\documentclass[final]{dmtcs-episciences}

\usepackage[utf8]{inputenc}
\usepackage[english]{babel}
\usepackage{t1enc}
\usepackage{dsfont}
\usepackage{amsmath,amssymb,amsthm}
\usepackage{bbold}
\usepackage{paralist}
\usepackage{mathtools}
\usepackage{hyperref}
\usepackage{stmaryrd}
\usepackage[ruled,lined,linesnumbered,procnumbered]{algorithm2e}
\usepackage{tikz}
\usepackage{multirow}

\usetikzlibrary{decorations.pathreplacing,automata,calc,positioning}
\tikzset{curlybrace/.style={decoration=brace,decorate}}

\DontPrintSemicolon
\SetKwFor{ForEvery}{forevery}{do}{end}
\SetKw{Var}{Variables}
\SetKw{Int}{Integer}
\SetKwComment{Comment}{}{}
\SetKw{Macro}{Macro}

\theoremstyle{definition}

\newtheorem{df}{Definition}[section]
\newtheorem{lm}[df]{Lemma}
\newtheorem{prop}[df]{Proposition}
\newtheorem{theo}[df]{Theorem}
\newtheorem{cor}[df]{Corollary}
\newtheorem{ob}[df]{Observation}

\newtheorem{ex}[df]{Example}

\newcommand{\A}{{\cal A}}
\newcommand{\B}{{\cal B}}
\newcommand{\D}{{\cal D}}
\newcommand{\E}{{\cal E}}
\newcommand{\F}{{\cal F}}

\newcommand{\K}{{\cal K}}
\newcommand{\M}{{\cal M}}
\newcommand{\Ne}{{\cal N}}
\newcommand{\T}{{\cal T}}
\newcommand{\cP}{{\cal P}}
\newcommand{\N}{\mathbb{N}}
\newcommand{\R}{{\cal R}}
\newcommand{\sem}[1]{[\![#1]\!]}

\newcommand{\init}{\mathrm{init}}
\newcommand{\run}{\mathrm{run}}
\newcommand{\rel}{\mathrm{rel}}

\newcommand{\V}{\mathcal{V}}

\newcommand{\C}{{\cal C}}
\newcommand{\0}{\mathbb{0}}
\newcommand{\1}{\mathbb{1}}
\newcommand{\FS}{\mathcal{F}(\Delta^*)}
\newcommand{\FSinf}{\mathcal{F}(\Delta^*_\infty)}

\DeclareMathOperator{\acc}{acc}
\DeclareMathOperator{\der}{der}
\DeclareMathOperator{\id}{id}
\DeclareMathOperator{\im}{im}
\DeclareMathOperator{\wt}{wt}
\DeclareMathOperator{\lcm}{lcm}
\DeclareMathOperator{\pos}{pos}

\DeclareMathOperator{\rk}{rk}
\DeclareMathOperator{\Ops}{Ops}

\DeclareMathOperator{\pr}{pr}
\DeclareMathOperator{\size}{size}
\DeclareMathOperator{\sizemod2}{size-mod-2}

\DeclareMathOperator{\supp}{supp}
\DeclareMathOperator{\maxrk}{maxrk}
\DeclareMathOperator{\lcp}{lcp}

\title{Crisp-determinization of weighted tree automata over strong bimonoids}

\author{Zolt\'an F\"ul\"op\affiliationmark{1}\thanks{Research of this author was supported by grant TUDFO/47138-1/2019-ITM of the Ministry for Innovation and Technology, Hungary.} \and D\'avid K\'osz\'o\affiliationmark{1}\thanks{Supported by the ÚNKP-19-3 New National Excellence Program of the Ministry for Innovation and Technology.} \and
Heiko Vogler\affiliationmark{2}}

\affiliation{
Department of Foundations of Computer Science, University of Szeged, Hungary\\
Faculty of Computer Science, Technische Universit\"at Dresden, Germany
}

\keywords{strong bimonoid, semiring, weighted tree automaton, determinization, undecidability}

\received{2019-12-06}

\revised{2021-02-03}

\accepted{2021-06-04}

\begin{document}
\publicationdetails{23}{2021}{1}{18}{5943}
\maketitle

\sloppy

\begin{abstract} 
We consider weighted tree automata (wta) over strong bimonoids and their initial algebra semantics and their run semantics. There are wta for which these semantics are different; however, for bottom-up deterministic wta and for wta over semirings, the difference vanishes. A wta is crisp-deterministic if it is bottom-up deterministic and each transition is weighted by one of the unit elements  of the strong bimonoid. We prove that the class of weighted tree languages recognized by crisp-deterministic wta is the same as the class of recognizable step mappings.  Moreover, we investigate the following two crisp-determinization problems: for a given wta~$\mathcal{A}$, (a) does there exist a crisp-deterministic wta which computes the initial algebra semantics of $\mathcal{A}$ and (b)  does there exist a crisp-deterministic wta which computes the run semantics of $\mathcal{A}$? We show that the finiteness of the Nerode algebra $\mathcal{N}(\mathcal{A})$ of $\mathcal{A}$ implies a positive answer for (a), and that  the finite order property of $\mathcal{A}$ implies a positive answer for  (b). 
  We show a sufficient condition which guarantees the finiteness of $\mathcal{N}(\mathcal{A})$ and a sufficient condition which guarantees the finite order property of $\mathcal{A}$. Also, we provide an algorithm for the construction of the crisp-deterministic wta according to (a) if $\mathcal{N}(\mathcal{A})$ is finite, and similarly for (b) if $\mathcal{A}$ has finite order property.  We prove that it is undecidable whether an arbitrary wta~$\mathcal{A}$ is crisp-determinizable. We also prove that both, the finiteness of $\mathcal{N}(\mathcal{A})$ and the finite order property of $\mathcal{A}$ are undecidable.
\end{abstract}

\section{Introduction}

The \emph{determinization problem} shows up  if one wants to specify a problem (\textit{e.g.}, a formal language) in a nondeterministic way and to calculate its solution (\textit{e.g.}, membership) in a deterministic way. More precisely, the determinization problem asks the following: for a given nondeterministic device $\A$ of a given type (or class) $\T$, does there exist a deterministic device $\B$ of the same type which is semantically equivalent to $\A$?

It is well known that the determinization problem is solved positively if $\T$ is the class of all finite-state (string) automata (cf., \textit{e.g.}, \cite[Thm.~2.1]{hopull79}), \textit{i.e.}, for each nondeterministic finite-state automaton $\A$ there is an equivalent deterministic finite-state automaton $\B$. The construction of $\B$ from $\A$ is called powerset construction. The same holds true for the class $\T$ of all finite-state tree automata \cite[Thm.~1]{thawri68}.

The situation changes drastically if one considers the class $\T$ of all weighted string automata (wsa), \textit{i.e.}, finite-state string automata in which each transition is weighted by some element of a semiring   \cite{sch61} (cf. \cite[Ch.~VI.6]{eil74} and \cite{salsoi78,kuisal86,sak09,drokuivog09}). 
More precisely, there exists a wsa such that there is no equivalent deterministic wsa  (see, \textit{e.g.}, \cite[Lm.~6.3]{borvog03} for a weighted tree automaton over a monadic alphabet with this property)\footnote{Weighted tree automata over monadic alphabets and wsa  are equivalent, cf., \cite[p.~324]{fulvog09}.}
On the other side, there are subclasses of $\T$  for which the determinization problem can be solved positively: the subclass of all wsa over locally finite semirings \cite[p.~293]{kirmae05}, the subclass of all trim unambiguous wsa over the tropical semiring having the twins property \cite[Thm.~12]{moh97}, and the subclass of all wsa over min-semirings having the twins property \cite[Thm.~5]{kirmae05}.
The same situation is present if $\T$ is the class of all weighted tree automata \cite{berreu82,aleboz87,kui98,esikui03}, and subclasses for the positive solution of the determinization problem were identified in \cite[Cor. 4.9 and Thm. 4.24]{borvog03}, \cite[Thm. 3.17]{fulvog09}, and  \cite[Thm. 5.2]{bucmayvog10}. In \cite{allmoh03a,beacarprisak03} results for deciding the twins property of wsa have been shown; we refer to \cite{buefis12} for results on deciding the twins property of weighted tree automata. 

Weighted string automata have been investigated for a number of different weight algebras, \textit{e.g.}, for semirings, lattices \cite{matsalsalyu95} (also cf. \cite{kliyua95,rah09}), valuation monoids \cite{dromei11,dromei12}, and strong bimonoids \cite{drostuvog10,cirdroignvog10}. 
 Roughly speaking, strong bimonoids are semirings without the distributivity laws. There exist  wsa over strong bimonoids such that the  initial algebra semantics and the run semantics are different \cite[Ex.~25~and~26]{drostuvog10}. However, if the strong bimonoid is right distributive, then both semantics coincide \cite[Lm.~4]{drostuvog10}.

A special case of determinization of wsa over strong bimonoids is when we require that 
the resulting deterministic wsa is \emph{crisp-deterministic}, \textit{i.e.}, each of its transitions is weighted by the additive zero or the multiplicative unit element of the strong bimonoid;\ that is, arbitrary weights can show up only at the final states.

Crisp-deterministic wsa are worth investigating because the class of weighted languages recognized by them is exactly the class of recognizable step mappings  \cite[Lm. 8]{drostuvog10}.
A recognizable step mapping is the sum of finitely many weighted languages, each of which is  constant over a recognizable language (called step language) and zero over the complement of that language. Therefore, it is easy to give a recognizable step mapping effectively by the direct product of the finite automata for the step languages and a simple weight mapping over the set of states of the direct product automaton. We mention that  recognizable step mappings play an important role 
in the characterization of recognizable weighted languages by weighted MSO-logic \cite{drogas05,drogas07,drogas09}. In fact, the semantics of the weighted MSO-formula $\forall x. \varphi$ is a recognizable weighted language if the semantics of $\varphi$ is a recognizable step mapping \cite[Lm.~5.4]{drogas09}; moreover, there is a weighted MSO-formula $\varphi$ of which the semantics is a recognizable weighted language and the semantics $\forall x. \varphi$ is not recognizable \cite[Ex.~3.6]{drogas09}. The same holds for weighted MSO-logic on trees \cite{drovog06}.

For the class $\T$ of all wsa over strong bimonoids, the \emph{crisp-determinization problem} asks the following: For a given wsa $\A$, (a) does there exist a crisp-deterministic wsa $\B$ which computes the initial algebra semantics of $\A$ and (b) does there exist a crisp-deterministic wsa $\B$ which computes the run semantics of $\A$?   In \cite{drostuvog10,cirdroignvog10} subclasses of $\T$ were identified for which the crisp-determinization problem is solved positively. However, in  \cite{drostuvog10,cirdroignvog10} no decidability results on the membership problem of that subclass is given.

In this paper we consider the class $\T$ of all weighted tree automata (wta) over strong bimonoids.  We will follow the lines of \cite{cirdroignvog10} and  identify subclasses of $\T$ for which the
crisp-determinization problem is solvable, \textit{i.e.}, for every wta $\A$ of that subclass, there exists a crisp-deterministic wta $\B$ such that $\A$ and $\B$ are i-equivalent, \textit{i.e.},  $\A$ and $\B$ have the same initial algebra semantics, and $\B$ can be constructed effectively. Also we deal with the modified problem in which initial algebra semantics and i-equivalence are replaced by run semantics and r-equivalence, respectively.
In fact, we generalize the corresponding results of the papers \cite{cirdroignvog10} to the tree case.

Moreover, we deal with decidability problems concerning crisp-determinization of wta. We show that it is undecidable whether, for an arbitrary given wta, there is an i-equivalent crisp-deterministic wta. Moreover, we show the undecidability of two properties of wta which  
are relevant for crisp-determinization. These are as follows. To each wta $\A$ we can associate an algebra $\V(\A)$ such that if
the image $\im(h_{\V(\A)})$ of the unique homomorphism $h_{\V(\A)}$ from the term algebra to $\V(\A)$ is finite, then 
a crisp-deterministic wta can be constructed which is i-equivalent to $\A$. For wsa over fields it is shown to be decidable whether  this image is finite \cite[Sec. IV. 2]{berreu88}. Moreover,  
in \cite{sei94} it was shown that for each wta $\A$ over the tropical or the arctic semiring it is decidable if $\A$ is bounded. Since these semirings are idempotent, the fact that $\A$ is bounded, implies that $\im(h_{\V(\A)})$ is finite. In this paper we show that for arbitrary wta $\A$ it is undecidable whether $\im(h_{\V(\A)})$ is finite. By  restricting this result to the case of monadic input trees (\textit{i.e.}, strings), we have solved partially the open problem stated in \cite[Sect.~12]{cirdroignvog10}. It would be interesting to find strong bimonoids such that it is decidable whether,  for arbitrary wta $\A$ over such bimonoids, (i)  $\A$ is crisp-determinizable or (ii)  $\im(h_{\V(\A)})$ is finite.
Finite order property of a wta $\A$ is  also important for crisp-determinization because if $\A$ has this property, then a crisp-deterministic wta can be constructed which is r-equivalent to $\A$. We also show that for arbitrary wta $\A$ it is undecidable whether $\A$ has the finite order property.

Our paper is organized as follows.
In Section \ref{sect:preliminaries} we recall the necessary definitions and concepts. We have tried to make the paper self-contained.
In Section \ref{sect:wta} we recall the concept of wta over strong bimonoid with its initial algebra semantics and its run semantics.
We give a complete proof for the result in the folklore that
\begin{compactitem}
 \item for bottom-up deterministic wta the two kinds of semantics coincide (Theorem \ref{theo:run-sem=initial-sem}).
  \end{compactitem}

  In Section \ref{sect:cdwbts} we introduce the auxiliary concept of algebras with root weights and define two basic constructions with them. These algebras may be infinite and the semantics of each algebra with root weights is a weighted tree language.

In Section \ref{sec:Crisp-deterministic weighted tree automata} we show that
crisp-deterministic wta and finite algebras with root weights are essentially the same concepts.
Moreover,
\begin{compactitem}
\item we prove a characterization of the class of weighted tree languages which are recognized by crisp-deterministic wta in terms of finite algebras with root weights, as well as in terms of  recognizable step mappings (Theorem \ref{lm:cdwta_finite_image}).
\end{compactitem}

In Section \ref{sect:initial-algebra-semantics} we consider the problem whether, given a wta $\A$, a crisp-deterministic wta can be constructed such that it is i-equivalent to $\A$. For each wta $\A$, we introduce the algebra $\Ne(\A)$ with root weights, which we call the Nerode algebra of $\A$. We show that $\A$ and $\Ne(\A)$ are semantically equivalent. As a consequence, 
\begin{compactitem}
\item we obtain that if $\Ne(\A)$ is finite, then  
  $\A$ and the crisp-deterministic wta $\rel(\Ne(\A))$ derived from $\Ne(\A)$ are i-equivalent (Theorem \ref{th:cd-init}),
\item we characterize the case that $\Ne(\A)$ is finite (Theorem \ref{theo:Nerode_conditions}),
\item we give an isomorphic representation of $\Ne(\A)$ (Theorem \ref{thm:characterizatio-of-Nerode-automaton}), and
  \item we show that if $\Ne(\A)$ is finite, then $\rel(\Ne(\A))$ is minimal among all crisp-deterministic wta which satisfy a certain condition concerning the initial algebra semantics of $\A$ (Theorem~\ref{theo:minimal}).
  \end{compactitem} (However, the last result  does not mean that $\rel(\Ne(\A))$ is minimal among all  crisp-deterministic wta which are i-equivalent to $\A$.) Moreover,
  \begin{compactitem}
  \item we give sufficient conditions which guarantee that $\Ne(\A)$ is finite (Corollary \ref{cor:cd-init}), and
  \item we design an algorithm of which the input is an arbitrary wta $\A$, and which terminates if $\Ne(\A)$ is finite and delivers the crisp-deterministic wta $\rel(\Ne(\A))$ (Algorithm \ref{alg:construct-rel-Ne(A)}).
  \end{compactitem}
  In particular, the algorithm terminates if the above mentioned sufficient conditions hold.

  In Section \ref{sect:run-semantics} we consider the problem whether, given a wta $\A$, a crisp-deterministic wta can be constructed such that
  it is r-equivalent to $\A$. We introduce the concept of finite order property for a wta $\A$. Then
  \begin{compactitem}
  \item we prove that if $\A$ has the finite order property, then a crisp-deterministic wta can be constructed which is r-equivalent to $\A$ (Theorem \ref{theo:A_pi_equals_A_run}),
  \item we give sufficient conditions which guarantee that $\A$ has the finite order property (Corollary \ref{thm:det-run}), and
  \item we give an algorithm of which the input is an arbitrary wta $\A$ which has the finite order property, and which  delivers the crisp-deterministic wta $\R(\A)$ which is r-equivalent to $\A$ (Algorithm \ref{alg:construct-R(A)} ).
    \end{compactitem}

    In Section \ref{sect:undecidability} we prove that it is undecidable whether
\begin{compactitem}
\item an arbitrary bottom-up deterministic wta  is crisp-determinizable (Theorem \ref{thm:A-init-rec-step-map}),
\item for an arbitrary bottom-up deterministic wta, its Nerode algebra is finite  (Theorem \ref{thm:finite-index-undec}), and
\item an arbitrary bottom-up deterministic wta  has the finite order property (Theorem~\ref{thm:fin-ord-prop-undec}).
\end{compactitem}


\section{Preliminaries}\label{sect:preliminaries}

\subsection{Basic concepts}
We denote by $\N$  the set of natural numbers $\{0,1,2,\ldots\}$ and by $\N_+$  the set $\N \setminus \{0\}$. For every $m,n \in \N$, we define $[m,n]= \{i\in \N \mid m \le i\le n\}$. We abbreviate $[1,n]$ by $[n]$. Hence, $[0]=\emptyset$.

Let $f:A \to B$ be a mapping.
We define the {\em image} of $f$ to be the set $\im(f) = \{f(a) \mid a \in A\}$.  Let $A'\subseteq A$. The {\em restriction of $f$ to $A'$} is the mapping $f|_{A'}: A' \to B$ defined by $f|_{A'}(a)=f(a)$ for each $a\in A'$. We denote the set of all mappings $f:A \to B$  by $B^A$.
For two mappings $f:A \to B$ and $g:B \to C$, the {\em composition of $f$ and $g$} is denoted by $g \circ f$ and is defined by $(g \circ f)(a) = g(f(a))$ for each $a \in A$. 

Let $A$ be a set. Then $|A|$ denotes the cardinality of $A$ and $\cP(A)$ its set of subsets. For each $k\in \N$, a mapping $f:A^k \to A$ is also called a {\em $k$-ary operation on $A$}. The set of all $k$-ary operations on $A$ is denoted by $\Ops^{(k)}(A)$ and we define $\Ops(A)=\bigcup_{k\in \N}\Ops^{(k)}(A)$. Let $A'\subseteq A$ and $O\subseteq \Ops(A)$. We say that $A'$ is closed under operations of $O$ if for every $f\in O$ of arity $k$ and $a_1,\ldots,a_k\in A'$ we have $f(a_1,\ldots,a_k)\in A'$.

An {\em alphabet} is a finite and nonempty set $X$ of symbols. A {\em string over $X$} is a finite sequence $x_1 \ldots x_n$ with $n \in \N$ and $x_i \in X$ for each $i \in [n]$. We denote by $\varepsilon$ the empty sequence (where $n=0$) and by $X^*$ the {\em set of all strings} (or {\em words}) {\em over $X$}.

\subsection{Trees and tree languages}

We assume that the reader is familiar with the fundamental concepts and results of the theory of tree automata and tree languages \cite{eng75-15,gecste84,tata07}. Here we only recall some basic definitions. 

A {\em ranked alphabet} is a tuple $(\Sigma,\rk)$ where $\Sigma$ is an alphabet and $\rk: \Sigma \to \N$ is a mapping called {\em rank mapping}. For each $k \in \N$, we define $\Sigma^{(k)} = \{\sigma \in \Sigma \mid \rk(\sigma) = k\}$. Sometimes we write $\sigma^{(k)}$ to mean that $\sigma \in \Sigma^{(k)}$. We denote $\max\{k \in \N \mid \Sigma^{(k)} \neq \emptyset\}$ by $\maxrk(\Sigma)$. When $\rk$ is irrelevant or it is clear from the context, then we abbreviate $(\Sigma,\rk)$ by $\Sigma$. 

Let $\Sigma$ be a ranked alphabet and $H$ a set such that $\Sigma \cap H=\emptyset$. The {\em set of $\Sigma$-trees over $H$}, denoted by $T_\Sigma(H)$, is the smallest set $T$ such that (i) $\Sigma^{(0)} \cup H \subseteq T$ and (ii) if $k \in \N_+$, $\sigma \in \Sigma^{(k)}$, and $\xi_1,\ldots,\xi_k \in T$, then $\sigma(\xi_1,\ldots,\xi_k) \in T$. We will abbreviate $T_\Sigma(\emptyset)$ by $T_\Sigma$.

\begin{quote}\em In the rest of this paper, $\Sigma$ will denote an arbitrary ranked alphabet if not specified otherwise. In addition, we assume that $\Sigma^{(0)} \neq \emptyset$.
  \end{quote}

  Each subset $L \subseteq T_\Sigma$ is called a {\em $\Sigma$-tree language} (or just: {\em tree language}). A tree language $L \subseteq T_\Sigma$ is {\em recognizable} if there is a {\em finite-state tree automaton over $\Sigma$} which recognizes $L$.

In order to avoid repetition of the quantifications of $k$, $\sigma$, and $\xi_1,\ldots,\xi_k$, we henceforth only write that we consider a "$\xi \in T_\Sigma$ of the form $\xi=\sigma(\xi_1,\ldots,\xi_k)$" or "for every $\xi=\sigma(\xi_1,\ldots,\xi_k)$".
For every $\gamma \in \Sigma^{(1)}$ and $\xi \in T_\Sigma$ we abbreviate  the tree $\gamma(\xi)$ by $\gamma \xi$. Moreover, we abbreviate the tree $\gamma(\ldots \gamma(\xi)\ldots)$ with $n$ occurrences of $\gamma$, by $\gamma^n \xi$.

We define the {\em set of positions} of trees as the mapping $\pos: T_\Sigma(H) \to \cP(\N_+^*)$ as follows: (i) for each $\alpha \in (\Sigma^{(0)} \cup H)$, we define $\pos(\alpha)=\{\varepsilon\}$ and (ii) for every $\xi=\sigma(\xi_1,\ldots,\xi_k)$, where $k\in \N_+$ we define $\pos(\xi)=\{\varepsilon\} \cup \{iv \mid i \in [k], v \in \pos(\xi_i)\}$. 

Now for every $\xi \in T_\Sigma(H)$ and $w \in \pos(\xi)$ the {\em label of $\xi$ at $w$}, denoted by $\xi(w)$, and the {\em subtree of $\xi$ at $w$}, denoted by $\xi|_w$, are defined as follows: (i) for each $\alpha \in (\Sigma^{(0)} \cup H)$, we define $\alpha(\varepsilon)=\alpha$, and $\alpha|_\varepsilon=\alpha$ and (ii) for every $\xi=\sigma(\xi_1,\ldots,\xi_k)$ with $k\in \N_+$, we define $\xi(\varepsilon)=\sigma$ and $\xi|_\varepsilon=\xi$, and for every $i\in [k]$ and $v \in \pos(\xi_i)$, we define $\xi(iv)=\xi_i(v)$ and $\xi|_{iv}=\xi_i|_v$.

For every $\Delta\subseteq (\Sigma\cup H)$
and $\xi \in T_\Sigma(H)$, we define $\pos_\Delta(\xi)=\{w \in \pos(\xi)\mid \xi(w)\in \Delta\}$.

Lastly, we define $\Sigma$-contexts. Let $\square$ be a special symbol such that $\square \not\in \Sigma$. For this, we define the notion of elementary context as follows: for every $k\in \N_+$, $\sigma\in \Sigma^{(k)}$, $i\in [k]$, and $\xi_1,\ldots,\xi_{i-1}, \xi_{i+1},\ldots, \xi_k \in T_\Sigma$, the tree 
\(\sigma(\xi_1,\ldots,\xi_{i-1}, \Box,\xi_{i+1},\ldots, \xi_k)\)
is an {\em elementary $\Sigma$-context}. The set of all elementary $\Sigma$-contexts is denoted by $C^{\mathrm e}_\Sigma$.

Then the \emph{set of $\Sigma$-contexts}, denoted by  $C_\Sigma$,  is the smallest set $C$ which satisfies the following two conditions:
\begin{itemize}
\item[(i)] $\Box \in C$, and
\item[(ii)] for every $e\in C^{\mathrm{e}}_\Sigma$ of the form $\sigma(\xi_1,\ldots,\xi_{i-1}, \Box,\xi_{i+1},\ldots, \xi_k)$ and for every $c\in C$, the tree $\sigma(\xi_1,\ldots,\xi_{i-1}, c,\xi_{i+1},\ldots, \xi_k)$ is in $C$.
\end{itemize}
Clearly, $C_\Sigma^{\mathrm{e}} \subseteq C_\Sigma$.

For every $c \in C_\Sigma$ and $\xi \in T_\Sigma \cup C_\Sigma$, we denote by $c[\xi]$ the tree  obtained from $c$ by replacing the unique occurrence of $\square$  by $\xi$. We note that $c[\xi] \in T_\Sigma$ if $\xi \in T_\Sigma$, and $c[\xi] \in C_\Sigma$ if $\xi \in C_\Sigma$.

\subsection{Algebraic structures}

We assume that the reader is familiar with the basic concepts and results of universal algebra \cite{gra68,bursan81} as well as basic concepts of semigroups and strong bimonoids \cite{drostuvog10,cirdroignvog10}. However, here we recall those concepts which we will use in the paper without any reference.

\paragraph{Universal algebra.}
A {\em $\Sigma$-algebra} is a pair $(A,\theta)$ which consists of a nonempty set $A$ and a $\Sigma$-indexed family $\theta=(\theta(\sigma)\mid \sigma \in \Sigma)$ over $\Ops(A)$ such that $\theta(\sigma): A^k \to A$ for every $k \in \N$ and $\sigma \in \Sigma^{(k)}$. Then $A$ is the {\em carrier set} and $\theta$ is the {\em $\Sigma$-interpretation} (or: {\em interpretation of $\Sigma$}), of that $\Sigma$-algebra. We call a $\Sigma$-algebra {\em finite} if its carrier set  is finite.
Next we show two examples of $\Sigma$-algebras.
\begin{ex}\hfill
  \begin{enumerate}
  \item The {\em $\Sigma$-term algebra} is the $\Sigma$-algebra $(T_\Sigma,(\overline{\sigma}\mid \sigma\in\Sigma))$, where  $\overline{\sigma}(\xi_1,\ldots,\xi_k)=\sigma(\xi_1,\ldots,\xi_k)$ for every $k\in \N$, $\sigma \in \Sigma^{(k)}$, and $\xi_1,\ldots,\xi_k \in T_\Sigma$.
        \item Let $\Sigma=\{+^{(2)}, *^{(2)}, -^{(2)}\}$. Then $(\mathbb{Z},\theta)$ is a $\Sigma$-algebra, where $\theta(+)$, $\theta(*)$, and $\theta(-)$ are the usual addition, multiplication, and substraction over integers, respectively. \hfill $\Box$
     \end{enumerate} 
\end{ex}

\begin{quote}\em As usual, if confusion is ruled out, we identify the $\Sigma$-algebra $(A,\theta)$ with its carrier set $A$.
\end{quote}

Let $(A,\theta)$ be a $\Sigma$-algebra.  A $\Sigma$-algebra $(A',\theta')$ is a {\em subalgebra of $A$}
if $A' \subseteq A$ and, for every $k \in \N$ and $\sigma \in \Sigma^{(k)}$,  we have $\theta'(\sigma)=\theta(\sigma)|_{(A')^k}$.
Let $X \subseteq A$ and $\Sigma'$ be a ranked alphabet such that $\Sigma'^{(k)}\subseteq \Sigma^{(k)}$ for each $k \in \N$. 
We denote by $\langle X \rangle_{\Sigma'}$ the smallest subset of $A$ which contains $X$ and is closed under operations of $\{\theta(\sigma)\mid \sigma \in \Sigma'\}$. Then  $(\langle X \rangle_{\Sigma'},\theta')$ is a  $\Sigma'$-algebra, where  $\theta'(\sigma)(a_1,\ldots,a_k)=\theta(\sigma)(a_1,\ldots,a_k)$ for every $k \in \N$, $\sigma \in \Sigma'^{(k)}$, and $a_1,\ldots,a_k\in \langle X \rangle_{\Sigma'}$.
In particular, $(\langle X \rangle_{\Sigma},\theta')$ is a  $\Sigma$-algebra, which we call the {\em subalgebra of $(A,\theta)$ generated by $X$}.
The \emph{smallest subalgebra of $(A,\theta)$} is its subalgebra generated by $\emptyset$.

We say that $(A,\theta)$ is {\em locally finite} if for each finite subset $X \subseteq A$ the set $\langle X \rangle_{\Sigma}$ is finite.

Let $((A_i,\theta_i)\mid i \in [n])$ be a family of $\Sigma$-algebras. The {\em direct product of $((A_i,\theta_i)\mid i \in [n])$} is the $\Sigma$-algebra $(A,\theta)$ where
\begin{compactitem}
\item $A = A_1 \times \ldots \times A_n$ and
\item for every $k\in\N$, $\sigma \in \Sigma^{(k)}$, $(a_{11},\ldots,a_{1n}),\ldots,(a_{k1},\ldots,a_{kn}) \in A$ we have that
\begin{equation}\label{eq:product-theta}
\begin{aligned}
&\theta(\sigma)((a_{11},\ldots,a_{1n}),\ldots,(a_{k1},\ldots,a_{kn}))\\
= & \ (\theta_1(\sigma)(a_{11},\ldots,a_{k1}),\ldots,\theta_n(\sigma)(a_{1n},\ldots,a_{kn}))\ .
\end{aligned}
\end{equation}
\end{compactitem}

Let $(A_1,\theta_1)$ and $(A_2,\theta_2)$ be two $\Sigma$-algebras and $h:A_1 \to A_2$  a mapping. Then $h$ is a {\em $\Sigma$-algebra homomorphism} (from $A_1$ to $A_2$) if for every $k\in\N$, $\sigma \in \Sigma^{(k)}$, and $a_1,\ldots,a_k \in A_1$, we have
\(h(\theta_1(\sigma)(a_1,\ldots,a_k))=\theta_2(\sigma)(h(a_1),\ldots,h(a_k))\).
If $h$ is bijective, then $h$ is a {\em $\Sigma$-algebra isomorphism}. If there is such an isomorphism, then we say that $A_1$ and $A_2$ {\em isomorphic}. We denote this fact by $A_1 \cong A_2$. 

Let $(A_1,\theta_1)$ and $(A_2,\theta_2)$ be two $\Sigma$-algebras and $h:A_1 \to A_2$  a $\Sigma$-algebra homomorphism.
Then $\im(h)$ is closed under the operations of $\{\theta_2(\sigma)\mid \sigma \in \Sigma\}$ and thus $(\im(h),\theta')$
is a subalgebra of $A_2$, where $\theta'(\sigma)(h(a_1),\ldots,h(a_k))=h(\theta_1(\sigma)(a_1,\ldots,a_k))$ for every $k \in \N$, $\sigma \in \Sigma^{(k)}$, and $a_1,\ldots,a_k\in A_1$. We call $(\im(h),\theta')$ the {\em $h$-image of $A_1$ (in $A_2$)}.

It is well known that the $\Sigma$-term algebra $T_\Sigma$ is {\em initial} in the class of all $\Sigma$-algebras, which means that
for every $\Sigma$-algebra $A$, there is a unique $\Sigma$-algebra homomorphism from $T_\Sigma$ to $A$. We denote by $h_A$ the unique homomorphism from $T_\Sigma$ to $A$.

\begin{prop} \label{prop:smallest-subalgebra}
Let $(A,\theta)$ be a $\Sigma$-algebra. The smallest subalgebra of $A$ is the $h_A$-image of $T_\Sigma$.
\end{prop}

For every $c \in C_\Sigma$, we define the mapping $c^A : A \to A$ by induction on $c$ as follows. \label{page:defintion-c^A}
\begin{itemize}
\item[(i)] If $c = \Box$, then $c^{A}(a) = a$ for each $a\in A$.

\item[(ii)] If $c = e[c']$ for some $e = \sigma(\xi_1,\ldots,\xi_{i-1}, \Box,\xi_{i+1},\ldots, \xi_k)$ in $C_\Sigma^{\mathrm{e}}$ and $c' \in C_\Sigma$, then
  \[
    c^A(a) = \theta(\sigma)(h_A(\xi_1),\ldots,h_A(\xi_{i-1}),(c')^A(a),h_A(\xi_{i+1}),\ldots, h_A(\xi_k))
  \]
 for each $a\in A$.
\end{itemize}

\begin{lm}(cf. \cite[Prop.~2.5]{gogthawagwri77b}) \label{lm:hom-decomposition} Let $(A,\theta)$ be a $\Sigma$-algebra. For every $c \in C_\Sigma$ and $\xi \in T_\Sigma$, we have $h_A(c[\xi])=c^A(h_A(\xi))$.
\end{lm}
\begin{proof} We prove by induction on $c$. For $c=\Box$ the proof is obvious.

Let $c=e[c']$, where $e=\sigma(\xi_1,\ldots,\xi_{i-1}, \Box,\xi_{i+1},\ldots, \xi_k)$ is in $C_\Sigma^{\mathrm{e}}$ and $c' \in C_\Sigma$. Then
\begin{align*}
h_A(c[\xi])& =h_A\big( \sigma(\xi_1,\ldots,\xi_{i-1}, c'[\xi],\xi_{i+1},\ldots, \xi_k)\big) \\
& = \theta(\sigma)\big(h_A(\xi_1),\ldots,h_A(\xi_{i-1}), h_A(c'[\xi]),h_A(\xi_{i+1}),\ldots, h_A(\xi_k)\big) \\
& \hspace*{5mm} \text{(since $h_A$ is a homomorphism)} \\
& = \theta(\sigma)\big(h_A(\xi_1),\ldots,h_A(\xi_{i-1}), (c')^A(h_A(\xi)),h_A(\xi_{i+1}),\ldots, h_A(\xi_k)\big)\\
& \hspace*{5mm} \text{(by I.H)}\\
& =c^A(h_A(\xi)). \qedhere
\end{align*}
\end{proof}

\paragraph{Strong bimonoids.}
Now we recall a particular class of $\Sigma$-algebras: strong bimonoids \cite{drostuvog10,cirdroignvog10,rad10}. This is specified by a particular $\Sigma$ and particular algebraic laws which involve its operations. Here, as usual,  we abbreviate $\theta(\sigma)$ by $\sigma$ for every $\sigma \in \Sigma$.

A {\em strong bimonoid} is an algebra $(B,\oplus,\otimes,\0,\1)$ where $(B,\oplus,\0)$ is a commutative monoid, $(B,\otimes,\1)$ is a monoid, $\0 \neq \1$, and $\0$ acts as multiplicative zero, \textit{i.e.}, $b \otimes \0 = \0 \otimes b = \0$ for every $b \in B$. We call the operations $\oplus$ and $\otimes$ summation and multiplication, respectively.   In order to avoid parentheses, we associate with multiplication higher priority than with summation. Then we may write, \textit{e.g.}, $a \oplus b \otimes c$ instead of $a \oplus (b \otimes c)$.

Let $(B,\oplus,\otimes,\0,\1)$ be a strong bimonoid. It is
\begin{compactitem}
  \item {\em commutative} if $\otimes$ is commutative,
  \item {\em right distributive} if $(a\oplus b)\otimes c = a\otimes c \oplus b \otimes c$  for every $a,b,c\in B$,
  \item {\em left distributive} if $a\otimes (b\oplus c) = a\otimes b \oplus a \otimes c$ for every $a,b,c\in B$, and
  \item {\em bi-locally finite} if $(B,\oplus,\0)$ and $(B,\otimes,\1)$ are locally finite.
  \end{compactitem}

  An element $b \in B$ is {\em additively idempotent} if $b\oplus b=b$. Moreover, $B$ is {\em additively idempotent} if each $b\in B$ is additively idempotent.

  Let $b \in B$ and  $n\in \N$. We abbreviate $b \oplus \ldots \oplus b$, where $b$ occurs $n$ times, by $nb$. In particular, $0b=\0$. We abbreviate $\langle b \rangle_{\{\oplus\}}=\{nb \mid n \in \N\}$ by $\langle b \rangle$.  
If $\langle b \rangle$ is finite, then we say that $b$ has a {\em finite order in $(B,\oplus,\0)$}.
In this case there is a least number $i \in \N_+$ such that 
$ib=(i+k)b$ for some $k \in \N_+$, and there is a least number $p \in \N_+$ such that $ib = (i+p)b$.
  We call $i$ and $p$ the {\em index (of $b$)} and the {\em period (of $b$)}, respectively, and denote them by $i(b)$ and $p(b)$, respectively. Moreover, we call $i+p-1$, \textit{i.e.}, the number of elements of $\langle b \rangle$, the {\em order of $b$}. We illustrate the index and the period of $b$ in Figure \ref{fig:ib-pb}, where the directed arrow means addition of $b$.

  \begin{figure}[t]
    \small
    \centering
    \begin{tikzpicture}
        \tikzset{element/.style={circle,fill=black,inner sep=0pt,minimum size=3pt}}
        \node[element,label={below:$b$}] at (0,0) (b) {};
        \node[element,label={below:$2b$}] at (1,0) (b2) {};
        \node at (2,0) (ld) {$\ldots$};
        \node[element,label={below:$ $}] at (3,0) (ib-1) {};
        \node[element,label={[xshift=1.25cm, yshift=-2em]$i(b)b =(i(b)+p(b))b$}] at (4,0) (ib) {};
        \node[element, shift={(150:.75)}] at (4,.75) (p4) {};
        \node[element, shift={(210:.75)}] at (4,.75) (p5) {};
        \node[element, shift={(330:.75)}] at (4,.75) (p1) {};
        \node[element, shift={(30:.75)}] at (4,.75) (p2) {};
        \node[shift={(90:.75)}] at (4,.75) (p3) {$\cdots$};
        \draw[->, shorten <=1mm, shorten >=1mm] (b) -- (b2);
        \draw[->, shorten <=1mm] (b2) -- (ld);
        \draw[->, shorten >=1mm] (ld) -- (ib-1);
        \draw[->, shorten <=1mm, shorten >=1mm] (ib-1) -- (ib);
        \draw[->, shorten <=1mm, shorten >=1mm] (ib) edge[bend right] (p1) ;
        \draw[->, shorten <=1mm, shorten >=1mm] (p1) edge[bend right] (p2) ;
        \draw[->, shorten <=1mm, shorten >=1mm] (p2) edge[bend right] (p3) ;
        \draw[->, shorten <=1mm, shorten >=1mm] (p3) edge[bend right] (p4) ;
        \draw[->, shorten <=1mm, shorten >=1mm] (p4) edge[bend right] (p5) ;
        \draw[->, shorten <=1mm, shorten >=1mm] (p5) edge[bend right] (ib) ;
        \draw[->] (4.75,0.175) arc (-40:220:1);
        \node at (5.5,1) {$p(b)$};
        \draw[decorate, decoration={brace,amplitude=10pt}] (4,-.65) -- (0,-.65) node[midway,below,yshift=-10pt] {$i(b)$};
    \end{tikzpicture}
    \caption{\label{fig:ib-pb} Illustration of the index $i(b)$ and the period $p(b)$ of $b$.}
\end{figure}
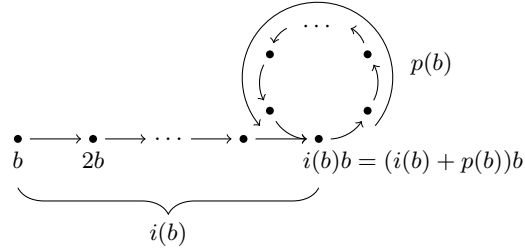
 
We extend $\oplus$ to every finite set $I$ and family $(a_i \mid i \in I)$ of elements of $B$. We denote the extended operation by $\bigoplus$ and define it as follows:
\[
  \bigoplus(a_i \mid i \in I) =
  \begin{cases}
    a_{i_1} \oplus \ldots \oplus a_{i_n} & \text{ if } I=\{i_1,\ldots,i_n\} \text{ for some $n \in \mathbb{N}_+$}\\
    \mathbb{0} & \text{ if } I=\emptyset \enspace.
  \end{cases}
\]
Since $\oplus$ is commutative, the sum above is well defined. Sometimes we abbreviate $\bigoplus(a_i \mid i \in I)$ by $\bigoplus_{i \in I}a_i$. If $I=[k]$ for some $k\in \N$, then we write $\bigoplus_{i=1}^k a_i$.

Moreover, we extend $\otimes$ to every $k\in \N$  and family $(a_i \mid i \in [k])$ of elements of $B$. We denote the extended operation by $\bigotimes$ and define it by:
\[
  \bigotimes(a_i \mid i \in [k]) =
  \begin{cases}
    a_{1} \otimes \ldots \otimes a_{k} & \text{ if } k \in \mathbb{N}_+\\
    \mathbb{1} & \text{ if } k=0 \enspace.
  \end{cases}
\]
We will abbreviate $\bigotimes(a_i \mid i \in [k])$ by $\bigotimes_{i=1}^k a_i$.

For every $B_1,B_2 \subseteq B$, we denote by $B_1 \otimes B_2$ the set $\{a \otimes b \mid a \in B_1, b \in B_2\}$.

A {\em semiring} is a strong bimonoid which is left distributive and right distributive. 

\begin{quote}
\em In the rest of this paper, $(B,\oplus,\otimes,\0,\1)$ will denote an arbitrary strong bimonoid if not specified otherwise.
\end{quote}

In the following example, we recall particular strong bimonoids and semirings which we will use later. We refer the reader for more examples of strong bimonoids (also those which are not semirings) to \cite[Ex.~1]{drostuvog10} and \cite[Ex.~2.1]{cirdroignvog10}.

\begin{ex} \hfill
\begin{enumerate}
\item The structure $\mathrm{TBM}=(\N_\infty,+,\min,0,\infty)$ \cite[Ex.~1(1)]{drostuvog10} with $\N_\infty = \N \cup \{\infty\}$ and the usual extensions of $+$ and $\min$ from $\N$ to $\N_\infty$ is a commutative strong bimonoid, called the {\em tropical bimonoid}. However, it is not bi-locally finite because $(\N_\infty,+,0)$ is not locally finite.
Moreover, it is not a semiring, because there are $a, b, c \in \N_\infty$ with $\min(a,b+c) \neq \min(a,b) + \min(a,c)$ (\textit{e.g.}, take $a = b = c \neq 0$).
\item The structure $\mathrm{TSR}=(\N_\infty,\min,+,\infty,0)$ is a semiring, called the {\em tropical semiring}.
\item The {\em Boolean semiring} is the semiring $(\mathbb{B},\vee,\wedge,0,1)$ where $\mathbb{B}=\{0,1\}$ (the truth values) and $\vee$ and $\wedge$ denote disjunction and conjunction, respectively. \hfill $\Box$
\end{enumerate}
\end{ex}

\subsection{Weighted tree languages} \label{subsect:weighted-tree-languages}
Let $H$ be a set disjoint with the ranked alphabet $\Sigma$ and the strong bimonoid $B$. A {\em weighted tree language over $\Sigma$, $H$ and $B$} is a mapping $r:T_\Sigma(H) \to B$. If $H = \emptyset$, then we say just weighted tree language over $\Sigma$ and $B$ or {\em $(\Sigma,B)$-weighted tree language}.

Let $r$ be  $(\Sigma,B)$-weighted tree language. The \emph{support of $r$}, denoted by $\supp(r)$, is the set $\{\xi \in T_\Sigma \mid r(\xi)\not= \0\}$.

Let $r$ and $r'$ be $(\Sigma,B)$-weighted tree languages and $b \in B$. We define the $(\Sigma,B)$-weighted tree languages $r\oplus r'$ and $b \otimes r$  by $(r\oplus r')(\xi)=r(\xi)\oplus r'(\xi)$ and $(b \otimes r)(\xi)=b\otimes r(\xi)$, respectively, for every $\xi \in T_\Sigma$. Moreover, we denote by $r^{-1}(b)$ the tree language defined by $r^{-1}(b) = \{\xi \in T_\Sigma \mid r(\xi) = b\}$.

Let $L\subseteq T_\Sigma$ be a tree language. The {\em characteristic mapping of $L$ with respect to $B$} is the mapping
$\mathbb{1}_{(B,L)}: T_\Sigma \to B$ defined for each $\xi \in T_\Sigma$ by:  $\1_{(B,L)}(\xi) = \1$ if $\xi \in L$ and $\1_{(B,L)}(\xi) = \0$ otherwise. 

Next let $r: T_\Sigma \to B$ be a $(\Sigma,B)$-weighted tree language and $\xi \in T_\Sigma$. The {\em quotient of $r$ with respect $\xi$}  is the weighted context language $\xi^{-1}r : C_\Sigma \to B$ defined by $(\xi^{-1}r)(c) = r(c[\xi])$ for every $c \in C_\Sigma$. In particular, $(\xi^{-1}r)(\square) = r(\square[\xi]) = r(\xi)$.


\section{Weighted tree automata}\label{sect:wta}

In this section we recall the definition of weighted tree automata \cite{fulvog09}, show examples, and compare the initial algebra semantics with the run semantics.

\subsection{The model}

A {\em weighted tree automaton over $\Sigma$ and $B$} (for short: $(\Sigma,B)$-wta, or just: wta) is a tuple $\A=(Q,\delta,F)$ where
\begin{compactitem}
\item $Q$ is a finite nonempty set ({\em states}),
\item $\delta = (\delta_k \mid k \in \N)$ is a family of mappings $\delta_k: Q^k \times \Sigma^{(k)} \times Q \to B$ ({\em transition mappings}), and
\item $F:Q \to B$ is a mapping ({\em root weight vector}).
\end{compactitem}

\begin{quote}\em In the rest of this paper, we abbreviate formulas of the form $\delta_k((q_1,\ldots,q_k),\sigma,q)$ by $\delta_k(q_1\ldots q_k,\sigma,q)$, for every $k \in \N$, $\sigma \in \Sigma^{(k)}$, and $q,q_1,\ldots,q_k \in Q$.
  \end{quote}

Let $\A=(Q,\delta,F)$ be a $(\Sigma,B)$-wta. Then $\A$ is {\em bottom-up deterministic} (for short bu-deterministic) if for every $k \in \N$, $\sigma \in \Sigma^{(k)}$, and $q_1,\ldots,q_k \in Q$ there is at most one $q \in Q$ such that $\delta_k(q_1 \ldots q_k,\sigma, q) \neq \0$. Moreover, $\A$ is {\em total} if for every $k \in \N$, $\sigma \in \Sigma^{(k)}$, and $q_1,\ldots,q_k \in Q$ there is at least one $q \in Q$ such that $\delta_k(q_1 \ldots q_k,\sigma, q) \neq \0$.

We call $\A$ a {\em crisp-deterministic} if  for every $k \in \N$, $\sigma \in \Sigma^{(k)}$ and $q_1,\ldots,q_k \in Q$ there is a unique $q \in Q$ such that $\delta_k(q_1\ldots q_k,\sigma,q)=\1$, and $\delta_k(q_1 \ldots q_k,\sigma,q')=\0$, for each $q' \in Q \setminus \{q\}$.  It is clear that each crisp-deterministic $(\Sigma,B)$-wta is bu-deterministic and total.

Before defining the semantics of $(\Sigma,B)$-wta, we introduce the following convention. We denote by  $\im(\delta)$
the set $\bigcup_{k\in \mathbb{N}}\im(\delta_k)$. Moreover, the elements of $B^Q$ are also called {\em $Q$-vectors over $B$}. For every $v \in B^Q$ and $q \in Q$, $v(q)$ is called the {\em $q$-component of $v$} and it is denoted by $v_q$ frequently.

\paragraph{Initial algebra semantics.} Let $\A$ be $(\Sigma,B)$-wta. The {\em vector algebra of $\A$} is the $\Sigma$-algebra $\V(\A)=(B^Q,\delta_\A)$, where, for every $k \in \N$, $\sigma \in \Sigma^{(k)}$, the $k$-ary operation $\delta_\A(\sigma): B^Q \times \cdots \times B^Q \to B^Q$ is defined by 
\[\delta_\A(\sigma)(v_1,\ldots,v_k)_q = \bigoplus_{q_1,\ldots,q_k \in Q} \left(\bigotimes_{i=1}^k (v_i)_{q_i}\right) \otimes \delta_k(q_1 \ldots q_k,\sigma,q)\]
for every $q\in Q$ and $v_1,\ldots,v_k \in B^Q$. We denote the unique $\Sigma$-algebra homomorphism from $T_\Sigma$ to $B^Q$ by $h_{\V(\A)}$. 

The {\em initial algebra semantics of $\A$} is the $(\Sigma,B)$-weighted tree language $\sem{\A}^\init$ such that
\[\sem{\A}^\init(\xi) = \bigoplus_{q \in Q} h_{\V(\A)}(\xi)_q \otimes F_q.\]
for every $\xi \in T_\Sigma$.

A $(\Sigma,B)$-weighted tree language $r$ is {\em initial algebra recognizable} (for short: i-recognizable) if there is a $(\Sigma,B)$-wta $\A$ such that $r = \sem{\A}^\init$. 

We note that i-recognizable $(\Sigma,\mathbb{B})$-weighted tree languages are the same as recognizable $\Sigma$-tree languages in the sense of \cite{eng75-15,gecste84,tata07}. Hence, we will specify a recognizable $\Sigma$-tree language $L$ by showing a bu-deterministic and total $(\Sigma,\mathbb{B})$-wta $\A=(Q,\delta,F)$ such that $L = \supp(\sem{\A}^{\init})$.

\paragraph{Run semantics.} Let $\A$ be $(\Sigma,B)$-wta and $\xi \in T_\Sigma$. A {\em run of $\A$ on $\xi$} is a mapping $\rho: \pos(\xi)\to Q$. If $\rho(\varepsilon) = q$ for some $q \in Q$, then we call $\rho$ a {\em $q$-run}. We denote by $R_\A(\xi)$ the {\em set of all runs of $\A$ on $\xi$} and by $R_\A(q,\xi)$ the {\em set of all $q$-runs of $\A$ on $\xi$}. For every $\rho \in R_\A(\xi)$ and $w \in \pos(\xi)$, the {\em run induced by $\rho$ at position $w$}, denoted by $\rho|_w \in R_\A(\xi|_w)$, is the mapping $\rho|_w:\pos(\xi|_w) \to Q$ defined by $\rho|_w(w') = \rho(ww')$ for every $w' \in \pos(\xi|_w)$. For every $\xi=\sigma(\xi_1,\ldots,\xi_k) \in T_\Sigma$, the {\em weight} $\wt_\A(\rho)$ of $\rho$ is the element of $B$ defined inductively by
\begin{equation}\label{eq:run-weight}
\wt_\A(\rho) = \left(\bigotimes_{i=1}^k\wt_\A(\rho|_i)\right)\otimes\delta_k(\rho(1)\ldots\rho(k),\sigma,\rho(\varepsilon)).
\end{equation}
The {\em run semantics of $\A$} is the $(\Sigma,B)$-weighted tree language $\sem{\A}^\run$ such that for every $\xi \in T_\Sigma$
\[\sem{\A}^\run(\xi) = \bigoplus_{\rho \in R_\A(\xi)} \wt_\A(\rho) \otimes F_{\rho(\varepsilon)} .\]

A $(\Sigma,B)$-weighted tree language $r$ is {\em run recognizable} (for short: r-recognizable) if there is a $(\Sigma,B)$-wta $\A$ such that $r = \sem{\A}^\run$.

\subsection{Examples}

In this subsection we show three examples of wta: a wta for which the initial algebra semantics and the run semantics are different, a bu-deterministic wta, and a crisp-deterministic wta.

Moreover, we visualize  wta  by  hypergraphs. A \emph{hypergraph over $\Sigma$} (for short: $\Sigma$-hypergraph) is a tuple $G = (V,E)$ where $V$ is a finite set (set of nodes) and $E \subseteq \bigcup_{k \in \mathbb{N}} V^k \times \Sigma^{(k)} \times V$ (set of hyperedges).

Each $\Sigma$-hypergraph $G = (V,E)$ can be illustrated by a figure as follows: we represent each node $v$ by a circle, and each hyperedge $\langle v_1,\ldots, v_k,\sigma,v\rangle$  by a box with $\sigma$ inscribed and with one outgoing arc leading to the node $v$ and with one ingoing arc coming from $v_i$ for each $i \in [k]$. In order to represent the order inherent in the list $v_1, \ldots, v_k$, the ingoing arcs are drawn such that, when traversing them counter-clockwise, starting from the outgoing arc, then the list of their source nodes is $v_1, \ldots, v_k$.
  
In particular, each $(\Sigma,B)$-wta $\A=(Q,\delta,F)$ can be represented by a $\Sigma$-hypergraph $G_\A= (Q,\delta')$ where $\delta' = \bigcup_{n \in \N} \supp(\delta_k)$. Moreover, in the illustration of  $G_\A$, we add to each hyperedge $e \in Q^k \times \Sigma^{(k)} \times Q$ the value $\delta_k(e)$ and to each node $q$ the value $F_q$.

\begin{figure}[t]
\begin{center}
\begin{tikzpicture}
\tikzset{node distance=7em, scale=0.6, transform shape}
\node[state, rectangle] (1) {$\alpha$};
\node[state, right of=1] (2) {$p_1$};
\node[state, rectangle, above of=2] (3) {$\gamma$};
\node[state, rectangle, above right of=2] (4) [right=5em] {$\gamma$};
\node[state, rectangle, below right of=2] (5) [right=5em] {$\gamma$};
\node[state, below right of=4] (6) [right=5em] {$p_2$};
\node[state, rectangle, above of=6] (7) {$\gamma$};
\node[state, rectangle, right of=6] (8) {$\alpha$};

\tikzset{node distance=2em}
\node[above of=1] (w1) {1};
\node[above of=2] (w2) {1};
\node[above of=3] (w3) {1};
\node[above of=4] (w4) {1};
\node[above of=5] (w5) {1};
\node[above of=6] (w6) {1};
\node[above of=7] (w7) {1};
\node[above of=8] (w8) {1};

\draw[->,>=stealth] (1) edge (2);
\draw[->,>=stealth] (2) edge[out=120, in=180, looseness=1.4] (3);
\draw[->,>=stealth] (3) edge[out=0, in=60, looseness=1.4] (2);
\draw[->,>=stealth] (2) edge[out=30, in=180, looseness=1.4] (4);
\draw[->,>=stealth] (4) edge[out=0, in=150, looseness=1.4] (6);
\draw[->,>=stealth] (6) edge[out=210, in=0, looseness=1.4] (5);
\draw[->,>=stealth] (5) edge[out=180, in=-30, looseness=1.4] (2);
\draw[->,>=stealth] (6) edge[out=120, in=180, looseness=1.4] (7);
\draw[->,>=stealth] (7) edge[out=0, in=60, looseness=1.4] (6);
\draw[->,>=stealth] (8) edge (6);
\end{tikzpicture}
\end{center}

\caption{\label{fig:hypgraph-sem-diff} The $\Sigma$-hypergraph for the $(\Sigma,\mathrm{TBM})$-wta $\A=(Q,\delta,F)$ such that $\sem{\A}^\init \neq \sem{\A}^\run$.}
\end{figure}
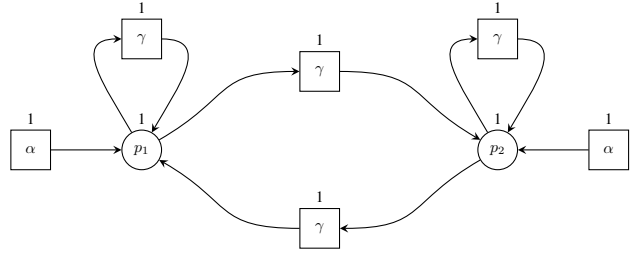

\begin{ex}\rm 
\label{ex:run-not=init} We consider the ranked alphabet $\Sigma = \{\gamma^{(1)}, \alpha^{(0)}\}$. Moreover, let $\A =(Q,\delta,F)$ be the $(\Sigma,\mathrm{TBM})$-wta with 
\begin{compactitem}
\item $Q = \{p_1,p_2\}$,
\item $\delta_0(\varepsilon,\alpha,q_1) = \delta_1(q_1,\gamma,q_2)= 1$ for every $q_1,q_2 \in Q$, and
\item $F_{p_1}= F_{p_2}=1$.
\end{compactitem}
Then $\A$ is not bu-deterministic because, \textit{e.g.}, $\delta_0(\varepsilon,\alpha,p_1)$ and $\delta_0(\varepsilon,\alpha,p_2)$ are not equal to~$0$. Figure \ref{fig:hypgraph-sem-diff} shows the $\Sigma$-hypergraph for $\A$.

Let $n \in \mathbb{N}$ and let us compute $\sem{\A}^\init(\gamma^n\alpha)$.
\[
\sem{\A}^\init(\gamma^n\alpha) = \sum_{q \in Q}{\min(h_{\V(\A)}(\gamma^n\alpha)_q, \ F_q)}
= \sum_{q \in Q}{\min(h_{\V(\A)}(\gamma^n\alpha)_q, \ 1)}
=^{(*)} \sum_{q \in Q}{1} = 2 \enspace,
\] 
where at $(*)$ we have used the following fact. For every $n \in \mathbb{N}$ and $q \in Q$: 
\[h_{\V(\A)}(\gamma^n\alpha)_{q}=
\left\{
\begin{array}{ll}
1 &\text{ if } n=0\\
2 &\text{ otherwise}\enspace.
\end{array}
\right.
\]
 This can be proved as follows. 
 If $n=0$, then $h_{\V(\A)}(\alpha)_q = \delta_0(\varepsilon,\alpha,q) = 1$. If $n \ge 1$, then:
\begin{align*}
h_{\V(\A)}(\gamma^n\alpha)_q 
=  \sum_{q' \in Q}{\min(h_{\V(\A)}(\gamma^{n-1}\alpha)_{q'}, \ \delta_1(q',\gamma,q))}
= \sum_{q' \in Q}{1}
= 2\enspace,
\end{align*}
where the second equality follows from (a) I.H. saying that $1 \le h_{\V(\A)}(\gamma^{n-1}\alpha)_{q'} \le 2$ and (b)~$\delta_1(q',\gamma,q)=1$. 

Now we compute $\sem{\A}^\run(\gamma^n\alpha)$. It is easy to see that $\wt_\A(\rho)=1$ for each run $\rho \in R_\A(\gamma^n\alpha)$.  Then 
\[
  \sem{\A}^\run(\gamma^n\alpha) = \sum_{\rho \in R_\A(\gamma^n\alpha)}{\min(\wt_\A(\rho), \ F_{\rho(\varepsilon)})} = \sum_{\rho \in R_\A(\gamma^n\alpha)}{1} = |R_\A(\gamma^n\alpha)| = 2^{n+1}\enspace.
\] 

Hence $\sem{\A}^\run \not= \sem{\A}^\init$. \hfill $\Box$
\end{ex}

\begin{ex}\label{ex:size} \rm Let $\Sigma = \{\sigma^{(2)}, \gamma^{(1)}, \alpha^{(0)}\}$. 
  We consider the mapping
  \[\size: T_\Sigma \to \mathbb{N}\]
  defined for each $\xi \in T_\Sigma$ by $\size(\xi)= |\pos(\xi)|$.
  As weight structure we use the tropical semiring $\mathrm{TSR}$. Thus, the mapping $\size$ is a $(\Sigma,\mathrm{TSR})$-weighted tree language.
  
We construct the $(\Sigma,\mathrm{TSR})$-wta $\C=(Q,\delta,F)$ such that its run semantics is $\size$, as follows.
\begin{compactitem}
\item $Q = \{q\}$ (intuitively, the state $q$ computes the size of the tree),

\item $\delta_2(qq,\sigma,q) = \delta_1(q,\gamma,q) = \delta_0(\varepsilon,\alpha,q) = 1$, and 
\item $F_q=0$.
\end{compactitem}
Clearly, $\C$ is bu-deterministic. However, it  is not crisp-deterministic, because $1$ is not the multiplicative unit element of  $\mathrm{TSR}$. Figure \ref{fig:hypgraph-size} shows the $\Sigma$-hypergraph for $\C$.

Let $\xi \in T_\Sigma$.
Since $q$ is the only state of $\C$, we have $R_\C(\xi)=\{\rho^\xi\}$,
where $\rho^\xi$ is the unique run on $\xi$ defined by $\rho^\xi(w) = q$ for each $w \in \pos(\xi)$. 
Thus
\[\sem{\C}^\run(\xi) = \min( \wt_\C(\rho) + F_{\rho(\varepsilon)} \mid  \rho \in R_\C(\xi)) = \wt_\C(\rho^\xi)= |\pos(\xi)| = \size(\xi).
\]
Hence $\sem{\C}^\run=\size$.
\hfill $\Box$
\end{ex}

\begin{figure}[t]
\begin{center}
\begin{tikzpicture}
\tikzset{node distance=7em, scale=0.6, transform shape}
\node[state, rectangle] (1) {$\gamma$};
\node[state, right of=1] (2) {$q$};
\node[state, rectangle, right of=2] (3) {$\sigma$};
\node[state, rectangle, below of=2] (4) {$\alpha$};

\tikzset{node distance=2em}
\node[above of=1] (w1) [right=0.05cm] {1};
\node[above of=2] (w2) {0};
\node[above of=3] (w3) {1};
\node[right of=4] (w4) {1};

\draw[->,>=stealth] (2) edge[out=-120, in=270, looseness=1.4] (1);
\draw[->,>=stealth] (1) edge[out=90, in=120, looseness=1.4] (2);
\draw[->,>=stealth] (2) edge[out=60, in=30, looseness=1.4] (3);
\draw[->,>=stealth] (2) edge[out=-60, in=-30, looseness=1.4] (3);
\draw[->,>=stealth] (3) edge (2);
\draw[->,>=stealth] (4) edge (2);
\end{tikzpicture}
\end{center}

\caption{\label{fig:hypgraph-size} The $\Sigma$-hypergraph for the $(\Sigma,\mathrm{TSR})$-wta $\C=(Q,\delta,F)$ which r-recognizes $\size$.}
\end{figure}
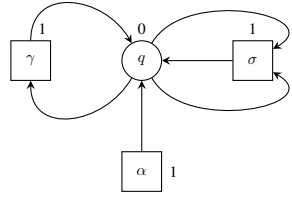

\begin{figure}[t]
  \begin{center}
\begin{tikzpicture}
\tikzset{node distance=7em, scale=0.6, transform shape}
\node[state] (1) {$e$};
\node[state, rectangle, right of=1] (2) [right=5em] {$\sigma$};
\node[state, rectangle, above of=2] (3)  {$\sigma$};
\node[state, rectangle, above of=3] (4)  {$\gamma$};
\node[state, rectangle, below of=2] (5)  {$\sigma$};
\node[state, rectangle, below of=5] (6)  {$\gamma$};
\node[state, right of=2] (7) [right=5em] {$o$};
\node[state, rectangle, above = 0.5em of $(2)!0.5!(7)$] (8) {$\alpha$};
\node[state, rectangle, right of=7] (9) {$\sigma$};

\tikzset{node distance=2em}
\node[above of=1] (w1) [left=0.05cm] {2};
\node[above of=2] (w2) {0};
\node[above of=3] (w3) {0};
\node[above of=4] (w4) {0};
\node[below of=5] (w5) {0};
\node[below of=6] (w6) {0};
\node[above of=7] (w7) [left=0.17cm] {3};
\node[above of=8] (w8) {0};
\node[above of=9] (w9) {0};

\draw[->,>=stealth] (1) edge[out=15, in=165, looseness=1.2] (2);
\draw[->,>=stealth] (1) edge[out=-15, in=195, looseness=1.2] (2);
\draw[->,>=stealth] (2) edge[out=0, in=195, looseness=1.4] (7);

\draw[->,>=stealth] (1) edge[out=75, in=120, looseness=1.4] (3);
\draw[->,>=stealth] (7) edge[out=105, in=60, looseness=1.4] (3);
\draw[->,>=stealth] (3) edge[out=270, in=45, looseness=1.4] (1);

\draw[->,>=stealth] (1) edge[out=90, in=180, looseness=1.4] (4);
\draw[->,>=stealth] (4) edge[out=0, in=90, looseness=1.4] (7);

\draw[->,>=stealth] (1) edge[out=-75, in=240, looseness=1.4] (5);
\draw[->,>=stealth] (7) edge[out=-105, in=300, looseness=1.4] (5);
\draw[->,>=stealth] (5) edge[out=90, in=-45, looseness=1.4] (1);

\draw[->,>=stealth] (7) edge[out=270, in=0, looseness=1.4] (6);
\draw[->,>=stealth] (6) edge[out=180, in=270, looseness=1.4] (1);

\draw[->,>=stealth] (8) edge[out=0, in=165, looseness=1.4] (7);

\draw[->,>=stealth] (7) edge[out=60, in=30, looseness=1.4] (9);
\draw[->,>=stealth] (7) edge[out=-60, in=-30, looseness=1.4] (9);
\draw[->,>=stealth] (9) edge (7);
\end{tikzpicture}
\end{center}

\caption{\label{fig:hypgraph-size-mod-2} The $\Sigma$-hypergraph for the $(\Sigma,\mathrm{TSR})$-wta $\A=(Q,\delta,F)$ which i-recognizes $\sizemod2$.}

\end{figure}

\begin{ex}\label{ex:size-mod2} \rm Let $\Sigma = \{\sigma^{(2)}, \gamma^{(1)},  \alpha^{(0)}\}$. 
  We consider the mapping
  $\sizemod2: T_\Sigma \to \mathbb{N}$
  defined for each $\xi \in T_\Sigma$ by
  \[
    \sizemod2(\xi) =
    \begin{cases}
      2 & \text{if  $|\pos(\xi)|$ is an even number }\\
      3 & \text{otherwise}\enspace.
      \end{cases}
    \]
    As weight structure we use the tropical semiring $\mathrm{TSR}$ again. 
  
We construct the $(\Sigma,\mathrm{TSR})$-wta  $\A=(Q,\delta,F)$ such that its initial algebra semantics is  $\sizemod2$ as follows. 
We  let 
\begin{compactitem}
\item $Q = \{e,o\}$,

  \item  $\delta_0(\varepsilon,\alpha,e) = \infty$ and $\delta_0(\varepsilon,\alpha,o) = 0$  and for every $q_1,q_2,q \in Q$ we let
   \[
    \delta_1(q_1,\gamma,q) =
    \begin{cases}
      0 & \text{ if $q_1\ne q$}\\
      \infty & \text{ otherwise},
      \end{cases}
    \]
    
  \[
    \delta_2(q_1q_2,\sigma,q) =
    \begin{cases}
      0 & \text{ if ($q_1=q_2$ and $q=o$) or ($q_1\not=q_2$ and $q=e$)}\\
      \infty & \text{ otherwise},
      \end{cases}
    \]
    and 
\item $F_e=2$ and $F_o = 3$.
\end{compactitem}
Then $\A$ is crisp-deterministic and consequently also bu-deterministic. Figure \ref{fig:hypgraph-size-mod-2} shows the $\Sigma$-hypergraph for $\A$.

Let $\xi \in T_\Sigma$ and $q \in Q$. It is clear that
\[
  h_{\V(\A)}(\xi)_q =
  \begin{cases}
    0 & \text{ if }  \ \big((q=e \text{ and } |\pos(\xi)| \text{ is even } ) \text{ or }  (q=o \text{ and } |\pos(\xi)| \text{ is odd } )\big)\\
    \infty & \text{ otherwise}.
    \end{cases}
\]
Thus
\begin{align*}
  \sem{\A}^\init(\xi) & = \min_{q \in Q} (h_{\V(\A)}(\xi)_q + F_q)  = \min \big( (h_{\V(\A)}(\xi)_e + F_e),  (h_{\V(\A)}(\xi)_o + F_o)\big)  \\ & = \min \big( (h_{\V(\A)}(\xi)_e + 2) , (h_{\V(\A)}(\xi)_o + 3)\big) = \sizemod2(\xi) \enspace. 
\end{align*}
\hfill $\Box$
\end{ex}

\subsection{Relationship between the initial algebra semantics and the run semantics.}

As it is illustrated by Example \ref{ex:run-not=init}, in general, the initial algebra semantics and the run semantics of wta are different (for the string case cf.  \cite{drostuvog10,cirdroignvog10}). However, if $B$ is a semiring, then the initial algebra semantics coincides with the run semantics. 

\begin{theo}\label{lm:run=initial} (\cite[Lm. 4.1.13]{bor04b} and \cite[Thm.~4.1]{rad10}, cf. also \cite[Lm.~4]{drostuvog10} for the string case)  Let $B$ be a semiring. Then $\sem{\A}^\init=\sem{\A}^\run$  for every  $(\Sigma,B)$-wta $\A$. 
\end{theo}

Also, the initial algebra semantics coincides with the run semantics for bu-deterministic $(\Sigma,B)$-wta.
We will prove this fact after the following preparation.

Let $\A=(Q,\delta,F)$ be a $(\Sigma,B)$-wta and let $\xi \in T_\Sigma$. We define the sets
\begin{align*}
Q_{\neq \0}^{h_{\V(\A)}}(\xi) &= \{q \in Q \mid h_{\V(\A)}(\xi)_q\not=\mathbb{0}\}, \text{ and} \\
Q_{\neq \0}^{R_\A}(\xi)& =\{q\in Q\mid (\exists \rho\in R_{\A}(q,\xi)) \text{ such that } \wt(\rho)\ne \mathbb{0}\}\enspace.
\end{align*}

\begin{lm}\rm\label{lm:properties-bu-wta}   
Let $\A=(Q,\delta,F)$ be a $(\Sigma,B)$-wta and $\xi \in T_\Sigma$ of the form $\xi=\sigma(\xi_1,\ldots,\xi_k)$. Then the following statement holds.
\begin{enumerate}
\item[(i)]
\begin{itemize}
\item[(a)] If $Q_{\neq \0}^{h_{\V(\A)}}(\xi_i)=\emptyset$ for some $i\in [k]$, then $Q_{\neq \0}^{h_{\V(\A)}}(\xi)=\emptyset$.
\item[(b)] If $Q_{\neq \0}^{R_\A}(\xi_i)=\emptyset$ for some $i\in [k]$, then $Q_{\neq \0}^{R_\A}(\xi)=\emptyset$.
\end{itemize}
\end{enumerate}
Moreover, if $\A$ is bu-deterministic, then also the following statements hold.
\begin{enumerate}
\item[(ii)] (a) $|Q_{\neq \0}^{h_{\V(\A)}}(\xi)|\le 1$ and (b) $|Q_{\neq \0}^{R_\A}(\xi)|\le 1$. 
\item[(iii)] Either (a) $Q_{\neq \0}^{h_{\V(\A)}}(\xi)=\emptyset=Q_{\neq \0}^{R_\A}(\xi)$ or (b) there is a $q\in Q$ such that 
$Q_{\neq \0}^{h_{\V(\A)}}(\xi)=\{q\}=Q_{\neq \0}^{R_\A}(\xi)$ and there is exactly one $\rho\in R_{\A}(q,\xi)$ with $\wt_\A(\rho)\ne \0$ and $h_{\V(\A)}(\xi)_{q} = \wt_\A(\rho)$.
\end{enumerate}
\end{lm}
\begin{proof} 
\underline{Proof of (i):} Let $q\in Q$ be arbitrary. To prove (a), we calculate as follows:
\begin{align*}
h_{\V(\A)}(\sigma(\xi_1,\ldots,\xi_k))_q &= h_{\V(\A)}(\bar{\sigma}(\xi_1,\ldots,\xi_k))_q = \delta_{\A}(\sigma)
(h_{\V(\A)}(\xi_1),\ldots, h_{\V(\A)}(\xi_k))_q \\
&= \bigoplus_{q_1,\dots,q_k \in Q} \Big( \bigotimes_{i=1}^k h_{\V(\A)}(\xi_i)_{q_i}\Big) \otimes \delta_k(q_1\dots q_k,\sigma,q),
\end{align*}
where $\bar{\sigma}$ is the operation of the $\Sigma$-term algebra associated to $\sigma$; the second equality holds, because $h_{\V(\A)}$ is a $\Sigma$-algebra homomorphism. By our assumption, $h_{\V(\A)}(\xi_i)_{q_i}=\0$ for each choice $q_i\in Q$, hence also $h_{\V(\A)}(\sigma(\xi_1,\ldots,\xi_k))_q=\mathbb{0}$.

To prove (b), let $\rho \in R_\A(q,\xi)$ and let us consider  Equation \eqref{eq:run-weight}. We have $\rho|_i\in R_\A(\rho(i),\xi_i)$, hence by our assumption $\wt_\A(\rho|_i)=\0$. Then also $\wt_\A(\rho)=\mathbb{0}$.

\underline{Proof of (ii):} We prove (a) by induction on $\xi$.  We assume that $|Q_{\neq \0}^{h_{\V(\A)}}(\xi_i)| \le 1$ for each $i \in [k]$, and we continue by case analysis.

\underline{Case (a1):} There is an $i \in [k]$ such that $|Q_{\neq \0}^{h_{\V(\A)}}(\xi_i)| = 0$. By Statement (i)(a) also \mbox{$|Q_{\neq \0}^{h_{\V(\A)}}(\xi)| =0$.}

\underline{Case (a2):} For each $i \in [k]$ we have  $|Q_{\neq \0}^{h_{\V(\A)}}(\xi_i)| = 1$, \textit{i.e.},
$Q_{\neq \0}^{h_{\V(\A)}}(\xi_i) = \{p_{i}\}$ for some $p_i\in Q$. Then for each $q \in Q$: 
\begin{align*}
h_{\V(\A)}(\sigma(\xi_1,\ldots,\xi_k))_q
=&\ \bigoplus_{q_1,\ldots,q_k \in Q} \Big( \bigotimes_{i=1}^k h_{\V(\A)}(\xi_i)_{q_i}\Big) \otimes \delta_k(q_1\ldots q_k,\sigma,q)\\
=&\ \Big( \bigotimes_{i=1}^k h_{\V(\A)}(\xi_i)_{p_{i}}\Big) \otimes \delta_k(p_{1}\ldots p_{k},\sigma,q)\enspace.
\end{align*}
Since $\A$ is bu-deterministic, there is at most one $q \in Q$ such that $\delta_k(p_{1}\ldots p_{k},\sigma,q) \not= \mathbb{0}$, and thus, $|Q_{\neq \0}^{h_{\V(\A)}}(\sigma(\xi_1,\ldots,\xi_k))| \le 1$. (Since $B$ may contain zero-divisors, the cardinality of this set can be $0$.) 

Statement (b) is proved in a very similar way to Statement (a).

\underline{Proof of (iii):} We prove by induction on $\xi$. If $Q_{\neq \0}^{h_{\V(\A)}}(\xi)=\emptyset=Q_{\neq \0}^{R_\A}(\xi)$, then we are done. 
Otherwise, first assume that  $Q_{\neq \0}^{h_{\V(\A)}}(\xi)\not=\emptyset$. Then, by Statement (ii)(a), $Q_{\neq \0}^{h_{\V(\A)}}(\xi)=\{q\}$ for some $q\in Q$ and, by Statement (i)(a), $Q_{\neq \0}^{h_{\V(\A)}}(\xi_i)\ne \emptyset$ for every $i\in [k]$. 
By the I.H, for every $i\in [k]$ there is a $q_i\in Q$ such that
$Q_{\neq \0}^{h_{\V(\A)}}(\xi_i)=\{q_i\}=Q_{\neq \0}^{R_\A}(\xi_i)$ and there is exactly one $\rho_i\in R_{\A}(q_i,\xi_i)$ with $\wt_\A(\rho_i)\ne \0$ and $h_{\V(\A)}(\xi_i)_{q_i} = \wt_\A(\rho_i)$.
 Now let the run $\rho \in R_\A(q,\xi)$ be defined by $\rho(\varepsilon)=q$, and
$\rho(iw)=\rho_i(w)$ for every $i\in[k]$ and $w\in \pos(\xi_i)$. Then we have
\begin{align*}
\wt_\A(\rho)= & \left(\bigotimes_{i=1}^k\wt_\A(\rho|_i)\right)\otimes\delta_k(\rho(1)\ldots\rho(k),\sigma,\rho(\varepsilon)) 
  =  \left(\bigotimes_{i=1}^k\wt_\A(\rho_i)\right)\otimes\delta_k(q_1\ldots q_k,\sigma,q)\\
= & \left(\bigotimes_{i=1}^k h_{\V(\A)}(\xi_i)_{q_i}\right)\otimes\delta_k(q_1\ldots q_k,\sigma,q)= h_{\V(\A)}(\xi)_q.
\end{align*}
Since $h_{\V(\A)}(\xi)_q\ne \0$, we have $\wt_\A(\rho)\ne  \0 $, \textit{i.e.}, $q\in Q_{\neq \0}^{R_\A}(\xi)$, and thus, by Statement (ii),  $Q_{\neq \0}^{R_\A}(\xi)=\{q\}$.
Lastly, let $\rho'\in R_\A(q,\xi)$ with $\rho'\ne \rho$. Using the definition of $\rho$, the assumption that $\rho_i$ is the only run in  $R_{\A}(q_i,\xi_i)$ with $\wt_\A(\rho_i)\ne \0$ for each $i\in [k]$, and the fact that $\A$ is bu-deterministic, we can easily show that $\wt_\A(\rho')= \0$.

The case $Q_{\neq \0}^{R_\A}(\xi)\not=\emptyset$ can be proved similarly.
\end{proof}

\begin{theo}\label{theo:run-sem=initial-sem} If $\A$ is a bu-deterministic $(\Sigma,B)$-wta, then $\sem{\A}^\init=\sem{\A}^\run$.
\end{theo}
\begin{proof} Let $\A=(Q,\delta,F)$. We show that $\sem{\A}^\init(\xi)=\sem{\A}^\run(\xi)$ for each $\xi\in T_\Sigma$. By Statement (iii) of Lemma \ref{lm:properties-bu-wta}, we can distinguish the following two cases.

\underline{Case 1:} $Q_{\neq \0}^{h_{\V(\A)}}(\xi) =\emptyset =Q_{\neq \0}^{R_\A}(\xi)$. Then
\[
\sem{\A}^\init(\xi)= \bigoplus_{q \in \emptyset} h_{\V(\A)}(\xi)_q \otimes F_q = \mathbb{0} =
\bigoplus_{q \in \emptyset}\Biggl(\bigoplus_{\rho \in R_{\A}(q,\xi)}\wt(\rho) \otimes F_q\Biggr) =
\sem{\A}^\run(\xi)\enspace.
\]

\underline{Case 2:} $Q_{\neq \0}^{h_{\V(\A)}}(\xi) =\{q\} =Q_{\neq \0}^{R_\A}(\xi)$ for some $q \in Q$ and there is exactly one $\rho\in R_{\A}(q,\xi)$ with $h_{\V(\A)}(\xi)_{q} = \wt_\A(\rho)$.
Then
\(
\sem{\A}^\init(\xi) = h_{\V(\A)}(\xi)_q \otimes F_q = \wt_\A(\rho) \otimes F_q = \sem{\A}^\run(\xi)\enspace.
\)
\end{proof}


\section{Algebras with root weights}\label{sect:cdwbts}

In this section we introduce the concept of algebra with root weights in order to study crisp-deterministic wta. This concept can be considered as generalization of weighted automata with infinitely many states \cite[p.~3502]{cirdroignvog10} to the tree case. The semantics of an algebra with root weights is a weighted tree language. Then we define two basic constructions with $(\Sigma,B)$-algebras. We will use both of them to give an isomorphic representation of the $(\Sigma,B)$-algebra $\mathcal{N}(\A)$ of a $(\Sigma,B)$-wta $\A$ (cf. Section~ \ref{sect:initial-algebra-semantics}, in particular, Theorem~\ref{thm:characterizatio-of-Nerode-automaton}).

\subsection{General concepts}

A {\em $\Sigma$-algebra with root-weight vector in $B$} (for short: $(\Sigma,B)$-algebra) is a triple $\K=(Q,\theta,F)$, where $Q$ is a (possibly infinite) set, 
$(Q,\theta)$ is a $\Sigma$-algebra, and $F:Q\to B$ is mapping. 

We denote  by $h_\K$ the unique $\Sigma$-algebra homomorphism from $T_\Sigma$ to $Q$. The {\em semantics of $\K$} is the weighted tree language  $\sem{\K}:T_\Sigma \to B$ defined by
\[
  \sem{\K}=F\circ h_\K\enspace.
  \]

Now we define some notions concerning $(\Sigma,B)$-algebras which we will use in the rest of this paper.
For this, let $\K=(Q,\theta,F)$ be a $(\Sigma,B)$-algebra.
A \emph{$(\Sigma,B)$-subalgebra of $\K$} is a $(\Sigma,B)$-algebra $\K'=(Q',\theta',F')$ such that $(Q',\theta')$ is a subalgebra of $(Q,\theta)$
and $F' = F|_{Q'}$.
We call $q \in Q$ {\em accessible (in $\K$)} if $q\in\im(h_\K)$. The \emph{accessible part of $\K$} is the
$(\Sigma,B)$-subalgebra  $\acc(\K)=(Q',\theta',F')$ of $\K$  such that $(Q',\theta')$ is the smallest subalgebra of $(Q,\theta)$. We call $\K$  {\em accessible} if $\im(h_\K) = Q$.

Let $\K'=(Q',\theta',F')$ be a further $(\Sigma,B)$-algebra. We say that {\em $\K$ and $\K'$ are isomorphic} if there is a $\Sigma$-algebra isomorphism $\varphi: Q \to Q'$ such that  $F_q=F'_{\varphi(q)}$ for every $q\in Q$. 
The following lemma can be proved by using standard arguments.

\begin{lm} If the $(\Sigma,B)$-algebra  $\K=(Q,\delta,F)$ and $\K'=(Q',\delta',F')$ are isomorphic, then $\sem{\K} = \sem{\K'}$. 
\end{lm}

\subsection{Direct product of algebras with root-weight vector}

Direct product of $(\Sigma,B)$-algebras is a generalization of the direct product of crisp-deterministic automata defined in \cite[p.~3503-3504]{cirdroignvog10}
to the tree case. It relies on the concept of direct product of $\Sigma$-algebras. We will use this concept to prove Theorem \ref{thm:characterizatio-of-Nerode-automaton}.

Let $\K = (\K_i\mid i \in [n])$ be a family of $(\Sigma,B)$-algebras with $\K_i=(Q_i,\theta_i,F_i)$. The {\em direct product} of $\K$ is the $(\Sigma,B)$-algebra $\Pi(\K)=(Q,\theta,F)$, where
\begin{compactitem}
\item $(Q,\theta)$ is the direct product of the $\Sigma$-algebras $((Q_i,\theta_i) \mid i \in [n])$ and
\item $F_{(q_1,\ldots,q_n)} = \bigotimes_{i=1}^n (F_i)_{q_i}$.
\end{compactitem}
Using \eqref{eq:product-theta}, we can easily see that
\begin{equation} \label{eq:direct_product_homomorphism}
  h_{\Pi(\K)}(\xi)= (h_{\K_1}(\xi),\ldots,h_{\K_n}(\xi)) \ \text{ for every $\xi \in T_\Sigma$}.
\end{equation}

\begin{lm}\label{lm:semantics-direct-prod} Let  $\K=(\K_i \mid i\in [n])$ be a family of $(\Sigma,B)$-algebras. For every $\xi\in T_\Sigma$ we have
\[\sem{\Pi(\K)}(\xi)=\bigotimes_{i=1}^n \sem{\K_i}(\xi)\enspace.\]
\end{lm}
\begin{proof} Let $\xi\in T_\Sigma$ and $\Pi(\K)=(Q,\theta,F)$. Then
\begin{align*}
\sem{\Pi(\K)}(\xi) = & (F \circ h_{\Pi(\K)})(\xi) = F(h_{\Pi(\K)}(\xi)) = F((h_{\K_1}(\xi), \ldots, h_{\K_n}(\xi))) \\ = & \bigotimes_{i=1}^n F_i(h_{\K_i}(\xi))=\bigotimes_{i=1}^n (F_i\circ h_{\K_i})(\xi)=\bigotimes_{i=1}^n \sem{\K_i}(\xi),
\end{align*}
where the third equality follows from \eqref{eq:direct_product_homomorphism}.
\end{proof}

Let $\Pi(\K)=(Q,\theta,F)$ be the direct product of the family $\K= (\K_i \mid i\in [n])$ with $(\Sigma,B)$-algebra $\K_i=(Q_i,\theta_i,F_i)$.
For every subset $Q'$ of $Q$ and $i \in [n]$, the mapping $\pr_i: Q' \to Q_i$ defined by $\pr_i(q_1,\ldots,q_n)=q_i$, for each $(q_1,\ldots,q_n) \in Q'$, is called the {\em $i$th projection mapping} of $Q'$ into $Q_i$. Any $(\Sigma,B)$-subalgebra $\K'=(Q',\theta',F')$ of $\Pi(\K)$ having the property that for each $i \in [n]$ the $i$th projection mapping $\pr_i: Q' \to Q_i$ is surjective is called a {\em subdirect product} of $\K$.

\begin{ob}\label{obs:subdirect-product}\rm The accessible part of $\Pi(\K)$ is a subdirect product of $\K$ if and only if $\K_i$ is accessible for each $i\in[n]$.
\end{ob}

\subsection{Derivative \texorpdfstring{$(\Sigma,B)$}{(\textbackslash Sigma,B)}-algebra of a weighted tree language}

The concept introduced below is a generalization of derivative automaton of \cite[p.~3504]{cirdroignvog10} to the tree case. It uses the concept of quotient of weighted tree language with respect to a tree defined in Subsection \ref{subsect:weighted-tree-languages}. We will apply it to prove Theorem \ref{thm:characterizatio-of-Nerode-automaton}.

Let $r$ be a $(\Sigma,B)$-weighted tree language. The {\em derivative $(\Sigma,B)$-algebra of $r$} is the $(\Sigma,B)$-algebra $\der(r)=(Q,\theta,F)$, where
\begin{compactitem}
\item $Q=\{\xi^{-1}r \mid \xi \in T_\Sigma\}$ (we recall that $\xi^{-1}r$ denotes the quotient of $r$ with respect to $\xi$),
\item $\theta(\sigma)(\xi^{-1}_1r, \ldots, \xi^{-1}_kr) = \sigma(\xi_1, \ldots, \xi_k)^{-1}r$
for every $k\in \N$, $\sigma \in \Sigma^{(k)}$, and $\xi_1,\ldots,\xi_k\in T_\Sigma$,
\item $F_{\xi^{-1}r}= r(\xi)$\enspace.
\end{compactitem}

We show that $\theta(\sigma)$ is a well defined mapping for each $k \in \N$ and $\sigma \in \Sigma^{(k)}$. For this, let $\zeta_1,\ldots,\zeta_k \in T_\Sigma$ be such that $\xi_i^{-1}r=\zeta_i^{-1}r$ for each $i \in [k]$.
To show that
\[\sigma(\xi_1,\xi_2,\ldots,\xi_k)^{-1}r=\sigma(\zeta_1,\zeta_2,\ldots,\zeta_k)^{-1}r\enspace,\] 
we first prove that 
\[\sigma(\xi_1,\xi_2,\ldots,\xi_k)^{-1}r=\sigma(\zeta_1,\xi_2,\ldots,\xi_k)^{-1}r\enspace.\]
Note that both of $\xi_1^{-1}r$ and $\zeta_1^{-1}r$ are mappings from $C_\Sigma$ to $B$. 
Let $c \in C_\Sigma$ and  $c' = c[\sigma(\square,\xi_2,\ldots,\xi_k)]$, \textit{i.e.},  $c'$ is the context which can be obtained by replacing $\square$ by $\sigma(\square,\xi_2,\ldots,\xi_k)$ in $c$. Then we have 
\[r(c[\sigma(\xi_1,\xi_2,\ldots,\xi_k)])=r(c'[\xi_1])=r(c'[\zeta_1])=r(c[\sigma(\zeta_1,\xi_2,\ldots,\xi_k)]),\]
where the second equality follows from $\xi^{-1}_1r=\zeta^{-1}_1r$. Thus, $\sigma(\xi_1,\xi_2,\ldots,\xi_k)^{-1}r=\sigma(\zeta_1,\xi_2,\ldots,\xi_k)^{-1}r$. By successive applications of the above reasoning, we obtain $\sigma(\xi_1,\xi_2,\ldots,\xi_k)^{-1}r=\sigma(\zeta_1,\zeta_2,\ldots,\zeta_k)^{-1}r$. Hence, $\theta(\sigma)$ is a well defined mapping.

Moreover, $F$ is also well defined  because if there are $\xi_1,\xi_2 \in T_\Sigma$ such that $\xi_1^{-1}r = \xi_2^{-1}r$, then for every $c \in C_\Sigma$ we have $r(c[\xi_1]) = r(c[\xi_2])$ and, by choosing $c= \Box$, we obtain $r(\xi_1) = r(\xi_2)$ and thus $F_{\xi_1^{-1}r} =F_{\xi_2^{-1}r}$.

For every $\xi=\sigma(\xi_1,\ldots,\xi_k)\in T_\Sigma$, we have that
\begin{equation}\label{eq:derivative-homomorphism}
\begin{aligned}
h_{\der(r)}(\xi) 
&= h_{\der(r)}(\sigma(\xi_1,\ldots,\xi_k))= \theta(\sigma)(h_{\der(r)}(\xi_1),\ldots,h_{\der(r)}(\xi_k))\\
&= \theta(\sigma)(\xi^{-1}_1r,\ldots,\xi^{-1}_kr)= \xi^{-1}r. 
\end{aligned}
\end{equation}
The above reasoning implies that $\der(r)$ is accessible.

\begin{lm}\label{lm:semantics-derivative} For every $r:T_\Sigma \to B$, we have that
$\sem{\der(r)}=r$.
\end{lm}
\begin{proof}  Let $\der(r)=(Q,\theta,F)$ and $\xi\in T_\Sigma$. Then
\[\sem{\der(r)}(\xi) = (F \circ h_{\der(r)})(\xi) = F(h_{\der(r)}(\xi)) = F(\xi^{-1}r) = r(\xi)\]
where the last but one equality follows from \eqref{eq:derivative-homomorphism}.
\end{proof}

\begin{ex}\label{ex:size-mod2-cont} \rm We reconsider the weighted tree language $\sizemod2$ from Example \ref{ex:size-mod2} and construct the derivative $(\Sigma,\mathrm{TSR})$-algebra $\der(\sizemod2)=(Q,\theta,F)$ of $\sizemod2$ as follows, where we abbreviate $\sizemod2$ by $r$. 
By definition, we have
\begin{compactitem}
\item $Q = \{\xi^{-1}r \mid \xi \in T_\Sigma\}$. 
\end{compactitem}
We analyze $Q$. Let $\xi \in T_\Sigma$. If $|\pos(\xi)|$ is even, then for every $c \in C_\Sigma$ we have
 \[
     \xi^{-1}r(c) =
    \begin{cases}
      2 & \text{if  $|\pos_\Sigma(c)|$ is an even number }\\
      3 & \text{otherwise}\enspace,
      \end{cases}
    \]
    where $|\pos_\Sigma(c)|$ is the set of positions of $c$ viewed as tree in $T_{\Sigma \cup \{\Box\}}$. If $|\pos(\xi)|$ is odd, then for every $c \in C_\Sigma$ we have
 \[
     \xi^{-1}r(c) =
    \begin{cases}
      2 & \text{if  $|\pos_\Sigma(c)|$ is an odd number }\\
      3 & \text{otherwise}\enspace.
      \end{cases}
    \]
Thus $|Q|=2$. Let us denote $ \xi^{-1}r$ by $e$ if $|\pos(\xi)|$ is even, and by $o$ otherwise.

Moreover, we have
\begin{compactitem}
\item $\theta(\alpha)()=e$, 
\item $\theta(\gamma)(e)=o$ and $\theta(\gamma)(o)=e$,
\item for every $q_1,q_2 \in Q$ we let $\theta(\sigma)(q_1,q_2)= e$ if $q_1 \neq q_2$, and $o$ otherwise,
\item $F_e=2$ and $F_o = 3$.
\end{compactitem}

Let $\xi \in T_\Sigma$. Then
\[
  h_{\der(r)}(\xi) = 
  \begin{cases}
    e &\text{if $|\pos(\xi)|$ is even }\\
    o &\text{otherwise}\enspace.
  \end{cases}
\]
and by Lemma \ref{lm:semantics-derivative} we have $\sem{\der(r)}=r$.
\hfill $\Box$
\end{ex}


\section{Crisp-deterministic weighted tree automata}
\label{sec:Crisp-deterministic weighted tree automata}

Finite $(\Sigma,B)$-algebras and crisp-deterministic $(\Sigma,B)$-wta are very close. Here we formalize their relationship.
For this let $\K = (Q,\theta,F)$ be a  finite $(\Sigma,B)$-algebra and $\A=(Q,\delta,F)$ be a  crisp-deterministic $(\Sigma,B)$-wta. We say that $\K$ and $\A$ are \emph{related} if for every $k \in \mathbb{N}$, $\sigma \in \Sigma^{(k)}$, and $q,q_1, \ldots,q_k \in Q$:
\[\theta(\sigma)(q_1,\ldots,q_k)=q \ \ \text{ if and only if} \ \ \delta_k(q_1\ldots q_k,\sigma,q)=\1.\]
Clearly, for each  finite $(\Sigma,B)$-algebra $\K$ there is exactly one  crisp-deterministic $(\Sigma,B)$-wta $\A$ such that $\K$ and $\A$ are related. We denote this $\A$ by $\rel(\K)$. Also vice versa, for each  crisp-deterministic $(\Sigma,B)$-wta $\A$ there  is exactly one  finite $(\Sigma,B)$-algebra $\K$ such that $\K$ and $\A$ are related.  We denote this $\K$ by $\rel(\A)$.

\begin{lm} \label{lm:cdwta_semantics} Let $\K$ be a  finite $(\Sigma,B)$-algebra and $\A$  be a  crisp-deterministic $(\Sigma,B)$-wta. If $\K$ and $\A$ are related, then
for each $\xi \in T_\Sigma$ and each $q \in Q$ we have that
$h_{\V(\A)}(\xi)_q = \mathbb{1}$ if $q= h_\K(\xi)$, and $\mathbb{0}$ otherwise. Moreover, $\sem{\K}=\sem{\A}^\init$.
\end{lm}
\begin{proof} For each $\xi \in T_\Sigma$ and each $q \in Q$ we show that  
\[
h_{\V(\A)}(\xi)_q = 
\left\{
\begin{array}{ll}
\mathbb{1} & \text{if } q= h_\K(\xi)\\
\mathbb{0} & \text{otherwise}\enspace.
\end{array}
\right.
\]
We prove the statement by induction on $\xi$. 
Let $\xi = \sigma(\xi_1,\ldots,\xi_k)$. Then we obtain
\begin{align*}
h_{\V(\A)}(\sigma(\xi_1,\ldots,\xi_k))_q & = \bigoplus_{q_1,\ldots, q_k\in Q}\Big(\bigotimes_{i=1}^k h_{\V(\A)}(\xi_i)_{q_i}\Big)\otimes \delta_k(q_1\ldots q_k,\sigma,q)\\
& = \Big(\bigotimes_{i=1}^k \1 \Big)\otimes \delta_k(h_\K(\xi_1) \ldots h_\K(\xi_k),\sigma,q)
\tag{by I. H., note that $h_{\V(\A)}(\xi_i)_{h_\K(\xi_i)}=\1$}\\[2mm] &=
\delta_k(h_\K(\xi_1) \ldots h_\K(\xi_k),\sigma,q) \\[2mm]
&  = 
\left\{
\begin{array}{ll}
\mathbb{1} & \text{if $\theta(\sigma)(h_\K(\xi_1), \ldots, h_\K(\xi_k))=q$ } \\[2mm]
\mathbb{0} & \text{otherwise\enspace.}
\end{array}
\right.
\end{align*}
Since $\theta(\sigma)(h_\K(\xi_1), \ldots, h_\K(\xi_k)) = h_\K(\sigma(\xi_1,\ldots,\xi_k))$, we have proved the statement.

Then for every $\xi \in T_\Sigma$ we have
\begin{align*}
\sem{\A}^\init(\xi)=\bigoplus_{q\in Q}h_{\V(\A)}(\xi)_q \otimes F_q = F_{h_\K(\xi)} =(F\circ h_\K)(\xi)= \sem{\K}(\xi). 
\end{align*}
\end{proof}

In the following we give characterizations for the weighted tree languages which are i-recognizable by crisp-deterministic wta in terms  of recognizable step mappings.

Let $r$ be a $(\Sigma,B)$-weighted tree language. Then $r$ is a {\em $(\Sigma,B)$-recognizable step mapping} if there are $n \in \N_+$, recognizable $\Sigma$-tree languages $L_1,\ldots,L_n \subseteq T_\Sigma$, and $b_1,\ldots,b_n \in B$ such that $r = \bigoplus_{i=1}^nb_i \otimes \mathbb{1}_{(B,L_i)}$. Each tree language $L_i$ is called {\em step language}. We say that a $(\Sigma,B)$-recognizable step mapping is in \emph{normal form} if the family of its step languages is a partitioning of $T_\Sigma$.

\begin{ex}\rm \label{ex:size-mod2-step} We consider again the weighted tree language $\sizemod2$ defined in Example~\ref{ex:size-mod2}.  It can be specified as follows:
    \[
\sizemod2 = \min(2 + 0_{(\mathrm{TSR},L_1)} , 3 + 0_{(\mathrm{TSR},L_2)})
\]
where $L_1 = \{\xi \in T_\Sigma \mid |\pos(\xi)| \text{ is even}\}$ and $L_2 = T_\Sigma \setminus L_1$. Clearly, $L_1$ and $L_2$ are recognizable $\Sigma$-tree languages, hence $\sizemod2$ is a recognizable step mapping in normal form.
 \hfill $\Box$
  \end{ex}

\begin{lm} (cf. \cite[Lm. 8. and Prop. 9.]{drostuvog10}) \label{lm:cdwta_finite_image}
Let $r: T_\Sigma \to B$. Then the following statements are equivalent.
\begin{enumerate}
\item[(i)] There is a crisp-deterministic $(\Sigma,B)$-wta $\A$ such that $r = \sem{\A}^\init$.
\item[(ii)] There is a finite $(\Sigma,B)$-algebra $\K$ such that $r = \sem{\K}$.
\item[(iii)] $r$ is a $(\Sigma,B)$-recognizable step mapping in normal form.
\item[(iv)] $r$ is a $(\Sigma,B)$-recognizable step mapping.
\item[(v)] $\im(r)$ is finite and for each $b \in B$ the $\Sigma$-tree language $r^{-1}(b)$ is recognizable.
\end{enumerate}
\end{lm}
\begin{proof} (i) $\Leftrightarrow$ (ii): For a given crisp-deterministic $(\Sigma,B)$-wta $\A$ it is trivial to construct a finite $(\Sigma,B)$-algebra $\K$ such that $\A$ and $\K$ are related. Then Lemma \ref{lm:cdwta_semantics} implies (i) $\Rightarrow$ (ii). In a similar way we can prove (ii) $\Rightarrow$ (i).

(ii) $\Rightarrow$ (iii): Let $\K = (Q,\theta,F)$ be a finite $(\Sigma,B)$-algebra. For each $q \in Q$, we let $\K^q = (Q,\theta,F^q)$ be the finite $(\Sigma,B)$-algebra such that $F^q(p) = \1$ if $q=p$, and $\0$ otherwise. Clearly, for every $\xi \in T_\Sigma$ and $q \in Q$ we have
  \begin{equation}
    h_\K(\xi)=q  \text{ iff } \ \sem{\K^q}(\xi)=\1 \text{ iff } \ \mathbb{1}_{(B,\supp(\sem{\K^q}))}(\xi)=\1\enspace. \label{equ:partitioning-state}
    \end{equation}
Let $\A^q=(Q,\delta,G^q)$ be the crisp-deterministic $(\Sigma,\mathbb{B})$-wta defined by 
\begin{compactitem}
\item $\delta_k(q_1\ldots q_k,\sigma,q)=1$ iff $\theta(\sigma)(q_1,\ldots, q_k)=q$ for every $k\in \N, \sigma \in \Sigma^{(k)}$, and $q,q_1,\ldots, q_k\in Q$, and
\item  $G^q(p) =1$ if $q=p$, and $0$ otherwise.
\end{compactitem}
Then it is easy to show that $\supp(\sem{\A^q}^\init)=\supp(\sem{\K^q})$, the proof is similar to the proof of Lemma  \ref{lm:cdwta_semantics}. Hence $\supp(\sem{\K^q})$ is a recognizable $\Sigma$-tree language (by definition). 
Moreover, the $Q$-indexed family  $(\supp(\sem{\K^q}) \mid q \in Q)$ is a partitioning of $T_\Sigma$.  Then, for each $\xi \in T_\Sigma$, we have 
  \[
    \sem{\K}(\xi) = F(h_\K(\xi)) = \bigoplus_{q \in Q} F_q \otimes \mathbb{1}_{(B,\supp(\sem{\K^q}))}(\xi) = \Big(\bigoplus_{q \in Q} F_q \otimes \mathbb{1}_{(B,\supp(\sem{\K^q}))}\Big)(\xi)\enspace,
  \]
  where the second equality follows from \eqref{equ:partitioning-state}. Hence $\sem{\K}$ is a $(\Sigma,B)$-recognizable step mapping in normal form.

  (iii) $\Rightarrow$ (iv) and (iv) $\Leftrightarrow$ (v): These are obvious by definition.

  (iv) $\Rightarrow$ (i): In a straightforward way, we generalize the direction $\Rightarrow$ of \cite[Lm.~8]{drostuvog10} from the string case to the tree case. Let $r= \bigoplus_{i=1}^n b_i \otimes \mathbb{1}_{(B,L_i)}$ be a recognizable step mapping. For each $i \in [n]$ we let $\A_i =(Q_i,\delta_i,F_i)$ be some bu-deterministic and total $(\Sigma,\mathbb{B})$-wta such that $L_i = \supp(\sem{\A_i}^\init)$. We define the $(\Sigma,B)$-wta $\A=(Q,\delta,F)$ such that
  \begin{compactitem}
  \item $Q = Q_1 \times \ldots \times Q_n$,
  \item for every $k \in \N$, $\sigma \in \Sigma^{(k)}$, and $\widetilde{q_1},\ldots, \widetilde{q_n}, \widetilde{q} \in Q$ we let
    \[
      \delta_k(\widetilde{q_1}\ldots \widetilde{q_n}, \sigma, \widetilde{q}) =
      \begin{cases}
        \1 & \text{if for each $i \in [n]$: } (\delta_i)_k((\widetilde{q_1})_i\ldots (\widetilde{q_n})_i, \sigma, (\widetilde{q})_i)=1\\
        \0 & \text{otherwise}
      \end{cases}     
    \]
    where $(\widetilde{q_j})_i$ denotes the $i$th component of $\widetilde{q_j}$, and similarly for $(\widetilde{q})_i$, and
    \item for every $\widetilde{q} \in Q$ we let $F_{\widetilde{q}} = \bigoplus_{\substack{i \in [n]:\\ (F_i)_{(\widetilde{q})_i}=1}} b_i$
    \end{compactitem}
    Clearly, $\A$ is crisp-deterministic. Let $\xi \in T_\Sigma$. By Lemma \ref{lm:cdwta_semantics}  there is a unique state $\widetilde{q} \in Q$ with $h_{\V(\A)}(\xi)_{\widetilde{q}}=\1$. Then $\xi \in L_i$ iff $(F_i)_{(\widetilde{q})_i}=1$ for each $i \in [n]$. Let $I_\xi = \{i \in [n] \mid \xi \in L_i\}$. Then $r(\xi) = \bigoplus_{i \in I_\xi} b_i = F_{\widetilde{q}} = \sem{\A}^\init(\xi)$.
  \end{proof}


\section{Crisp-determinization for the initial algebra semantics}\label{sect:initial-algebra-semantics}

We first introduce, for each $(\Sigma,B)$-wta $\A$, the $(\Sigma,B)$-algebra $\mathcal{N}(\A)$, which we call the Nerode algebra of $\A$. It is the generalization of the Nerode automaton of a wsa defined in \cite[Sect.~6]{cirdroignvog10} to the tree case. Then we show that $\A$ and $\mathcal{N}(\A)$ are semantically equivalent (cf. Lemma~\ref{prop:A_N_equals_A_init}). In general $\mathcal{N}(\A)$ is not finite, but if it is so, then we can derive the crisp-deterministic wta $\rel(\mathcal{N}(A))$ from $\mathcal{N}(A)$, which is i-equivalent to $\A$ (cf. Theorem~\ref{th:cd-init}). We prove two interesting properties of $\mathcal{N}(\A)$: (1) we characterize the case that $\mathcal{N}(\A)$ is finite (cf. Theorem~\ref{theo:Nerode_conditions}) and (2) we give an isomorphic representation of $\mathcal{N}(\A)$ (cf. Theorem~\ref{thm:characterizatio-of-Nerode-automaton}). For the latter, we will use the constructions direct product of $(\Sigma,B)$-algebras and derivative $(\Sigma,B)$-algebra of a weighted tree language defined in Section~\ref{sect:cdwbts}. Then we show that if (1) $B$ is locally finite or (2) $B$ is multiplicatively locally finite and $\A$ is bu-deterministic, then $\mathcal{N}(\A)$ is finite (cf. Corollary~\ref{cor:cd-init}). Finally, we present an algorithm of which the input is an arbitrary $(\Sigma,B)$-wta $\A$, and which terminates if $\mathcal{N}(\A)$ is finite and delivers the crisp-deterministic wta $\rel(\mathcal{N}(A))$ (cf. Algorithm \ref{alg:construct-rel-Ne(A)})).

\subsection{Finiteness of the Nerode algebra implies crisp-determinization}

Let $\A=(Q,\delta,F)$ be a $(\Sigma,B)$-wta. The {\em Nerode $(\Sigma,B)$-algebra of $\A$}, denoted by $\Ne(\A)$, is the $(\Sigma,B)$-algebra $\Ne(\A)=(Q_\Ne,\theta_\Ne,F_\Ne)$, where 
\begin{compactitem}
\item $(Q_\Ne,\theta_\Ne)$ is the smallest subalgebra of the vector algebra $\V(\A)$ of $\A$, and 
\item $(F_\Ne)_v = \bigoplus_{q\in Q}v_q \otimes F_q$ for each $v \in Q_\Ne$.
\end{compactitem}

\begin{quote}\em In the rest of this paper,  we denote the components of $\Ne(\A)$ by $Q_\Ne$, $\theta_\Ne$, and $F_\Ne$.
\end{quote}

The next proposition follows from Proposition \ref{prop:smallest-subalgebra} and the fact that $h_{\V(\A)}$ is the unique homomorphism from $T_\Sigma$ to $\V(\A)$.

\begin{prop}\label{prop:Nerode=image} For each wta $\A$ we have $Q_\Ne=\im(h_{\V(\A)})$.
  \end{prop}

\begin{lm} \label{prop:A_N_equals_A_init} (cf. \cite[Prop.~6.1]{cirdroignvog10})
Let $\A=(Q,\delta,F)$ be a $(\Sigma,B)$-wta. Then $\sem{\Ne(\A)} = \sem{\A}^\init$.
\end{lm}

\begin{proof} By Proposition \ref{prop:Nerode=image} we have that $Q_\Ne = \{h_{\V(\A)}(\xi) \mid \xi \in T_\Sigma\}$.
   Moreover, we note that $\theta_\Ne(\sigma)(h_{\V(\A)}(\xi_1),\ldots,h_{\V(\A)}(\xi_k))=h_{\V(\A)}(\sigma(\xi_1,\ldots,\xi_k))$
for every $k\in \N$, $\sigma\in \Sigma^{(k)}$, and $\xi_1,\ldots,\xi_k \in T_\Sigma$.
We recall that $h_{\Ne(\A)}$ is the unique $\Sigma$-algebra homomorphism from $T_\Sigma$ to $Q_\Ne$. This
$h_{\Ne(\A)}$ is a homomorphism also from $T_\Sigma$ to $\V(\A)$ because  $Q_\Ne$ is
a subalgebra of $\V(\A)$. Since $h_{\V(\A)}$ is unique, we have $h_{\Ne(\A)}=h_{\V(\A)}$.

Then we obtain
\[\sem{\Ne(\A)}(\xi) = (F_\Ne \circ h_{\Ne(\A)})(\xi) =  (F_\Ne)_{h_{\Ne(\A)}(\xi)} = \bigoplus_{q\in Q}h_{\V(\A)}(\xi)_q \otimes F_q = \sem{\A}^\init(\xi)\]
for every $\xi \in T_\Sigma$. Therefore, $\sem{\Ne(\A)} = \sem{\A}^\init$.
\end{proof}

\begin{theo}\label{th:cd-init} Let $\A$ be a $(\Sigma,B)$-wta. If the Nerode $(\Sigma,B)$-algebra $\Ne(\A)$ is finite, then for the crisp-deterministic wta $\rel(\Ne(\A))$ we have $\sem{\A}^\init= \sem{\rel(\Ne(\A))}^\init$.
  \end{theo}

  \begin{proof}   Since $\Ne(\A)$ and $\rel(\Ne(\A))$ are related, we obtain $\sem{\Ne(\A)}= \sem{\rel(\Ne(\A))}^\init$ by Lemma \ref{lm:cdwta_semantics}.
Since $\sem{\A}^\init= \sem{\Ne(\A)}$ (by Lemma \ref{prop:A_N_equals_A_init}), we eventually obtain   $\sem{\A}^\init= \sem{\rel(\Ne(\A))}^\init$.
    \end{proof}
    
     \begin{figure}[t]
        \centering
        \begin{tikzpicture}
            \tikzset{scale=0.6, transform shape}
            \node[state] (e) {$e$};
            \node[state] at (7,0) (o) {$o$};
            \node[state] at (3.5,6) (r) {$r$};

            \node[state,rectangle] at (0,-1.5)  (1) {$\sigma$};
            \node[state,rectangle] at (0,1.5)  (2) {$\sigma$};
            \node[state,rectangle] at (3.5,1.5) (3) {$\gamma$};
            \node[state,rectangle] at (3.5,0) (4) {$\sigma$};
            \node[state,rectangle] at (3.5,-1.5) (5) {$\gamma$};
            \node[state,rectangle] at (1,4) (6) {$\gamma$};
            \node[state,rectangle] at (5,4) (7) {$\gamma$};
            \node[state,rectangle] at (-2,4) (8) {$\sigma$};
            \node[state,rectangle] at (7,4) (9) {$\sigma$};
            \node[state,rectangle] at (-4,4) (10) {$\sigma$};
            \node[state,rectangle] at (9,4) (11) {$\sigma$};
            \node[state,rectangle] at (3.5,4) (12) {$\alpha$};
            \node[state,rectangle] at (8.8,.5) (13) {$\alpha$};
            \node[state,rectangle] at (7.75, -1.5) (14) {$\sigma$};

            \draw[->,>=stealth] (1) -- (e);
            \draw[->,>=stealth] (e) edge[out=240, in=260,looseness=3] (1);
            \draw[->,>=stealth] (o) edge[out=260, in=280,looseness=.75] (1);
            \draw[->,>=stealth] (e) edge[out=120, in=100, looseness=3] (2);
            \draw[->,>=stealth] (o) edge[out=105, in=75, looseness=.75] (2);
            \draw[->,>=stealth] (2) -- (e);
            \draw[->,>=stealth] (e) edge[out=60, in=180, looseness=1.4] (3);
            \draw[->,>=stealth] (3) edge[out=0, in=120, looseness=1.4] (o);
            \draw[->,>=stealth] (e) edge[out=10, in=170, looseness=1.1] (4);
            \draw[->,>=stealth] (e) edge[out=-10, in=190, looseness=1.1] (4);
            \draw[->,>=stealth] (4) -- (o);
            \draw[->,>=stealth] (o) edge[out=240, in=0, looseness=1.4] (5);
            \draw[->,>=stealth] (5) edge[out=180, in=-60, looseness=1.4] (e);
            \draw[->,>=stealth] (e) edge[out=130, in=180,looseness=1.4] (6);
            \draw[->,>=stealth] (6) edge[out=0, in=240,looseness=1.2] (r);
            \draw[->,>=stealth] (o) edge[out=90,in=270,looseness=1.2] (7);
            \draw[->,>=stealth] (7) edge[out=90,in=310,looseness=1.4] (r);
            \draw[->,>=stealth] (e) edge[out=150, in=290, looseness=1.4] (8);
            \draw[->,>=stealth] (e) edge[out=170, in=250, looseness=1.2] (8);
            \draw[->,>=stealth] (8) edge[out=90, in=210, looseness=1.2] (r);
            \draw[->,>=stealth] (o) edge[out=70, in=250,looseness=1.4] (9); 
            \draw[->,>=stealth] (o) edge[out=50, in=290,looseness=1.2] (9);
            \draw[->,>=stealth] (9) edge[out=90, in=340,looseness=1] (r);
            \draw[->,>=stealth] (e) edge[out=190, in=290,looseness=1] (10);
            \draw[->,>=stealth] (o) edge[out=270,in=250,looseness=1.5] (10);
            \draw[->,>=stealth] (10) edge[out=90,in=180,looseness=1] (r);
            \draw[->,>=stealth] (e) edge[out=230, in=290,looseness=2.5] (11);
            \draw[->,>=stealth] (o) edge[out=30, in=250, looseness=1] (11);
            \draw[->,>=stealth] (11) edge[out=90, in=0,looseness=1] (r);
            \draw[->,>=stealth] (12) -- (r);
            \draw[->,>=stealth] (13) edge[out=180, in=10, looseness=1] (o);
            \draw[->,>=stealth] (o) edge[out=320, in=10,looseness=1.1] (14);
            \draw[->,>=stealth] (o) edge[out=350, in=-10,looseness=1.2] (14);
            \draw[->,>=stealth] (14) edge[out=180, in=290,looseness=1] (o);

            \node at (0.55,0.45) {$\infty$};
            \node at (6.4,0.45) {$\infty$};
            \node at (4,6.5) {$0$};

            \node at (9.2,1.2) {$0$};
            \node at (3.9, 4.7) {$3$};
            \node at (3.9,2.2) {$0$};
            \node at (3.9,0.7) {$0$};
            \node at (3.9,-0.8) {$0$};
            \node at (1.4,4.7) {$3$};
            \node at (5.4,4.7) {$2$};
            \node at (7.5,-.8) {$0$};
            \node at (0.4,2.2) {$0$};
            \node at (0.25,-0.8) {$0$};
            \node at (-1.6,4.7) {$3$};
            \node at (7.4,4.7) {$3$};
            \node at (-3.6,4.7) {$2$};
            \node at (9.4,4.7) {$2$};
        \end{tikzpicture}
        \vspace*{-5em}
        \caption{\label{fig:ex-6.4} The $\Sigma$-hypergraph for the $(\Sigma,\mathrm{TSR})$-wta $\D=(Q,\delta,F)$ in Example \ref{ex:size-mod2-non-det}.}
    \end{figure}
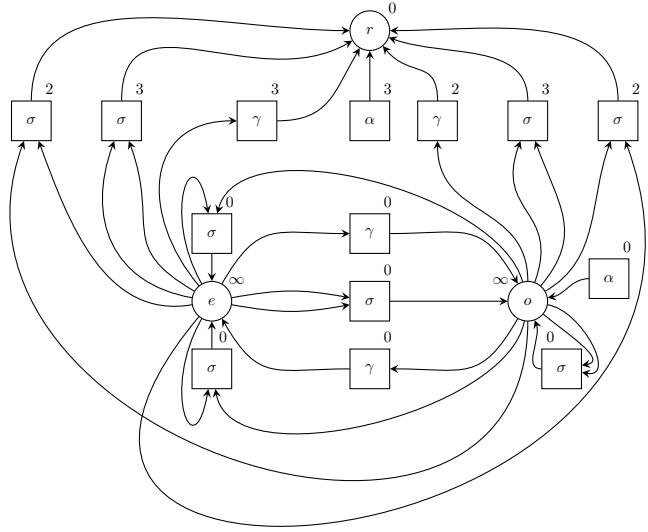

    \begin{ex} \label{ex:size-mod2-non-det} \rm 
 Let $\Sigma = \{\sigma^{(2)}, \gamma^{(1)},  \alpha^{(0)}\}$. 
 We consider the mapping $\sizemod2: T_\Sigma \to \mathbb{N}$ defined as in Example \ref{ex:size-mod2}. Here we will construct a wta $\D$ which is not bu-deterministic and which i-recognizes  $\sizemod2$, and we will analyze the Nerode algebra of $\D$.
 
    As weight structure we use the tropical semiring $\mathrm{TSR}=(\N_\infty,\min,+,\infty,0)$, \textit{i.e.}, the same algebra as in Examples \ref{ex:size} and \ref{ex:size-mod2}.  
We construct the $(\Sigma,\mathrm{TSR})$-wta  $\D=(Q,\delta,F)$ as follows.  
\begin{compactitem}
\item $Q = \{e,o,r\}$,

  \item  $\delta_0(\varepsilon,\alpha,e) = \infty$, $\delta_0(\varepsilon,\alpha,o) = 0$, and $\delta_0(\varepsilon,\alpha,r)=3$,  and for every $q_1,q_2,q \in Q$ we let
    \[
    \delta_1(q_1,\gamma,q) =
    \begin{cases}
      0 & \text{ if ($q_1=e$ and $q=o$) or ($q_1=o$ and $q=e$)}\\
      2 & \text{ if $q_1=o$ and $q=r$}\\
      3 & \text{ if $q_1=e$ and $q=r$}\\      
      \infty & \text{ otherwise},
      \end{cases}
    \]
    
  \[
    \delta_2(q_1q_2,\sigma,q) =
    \begin{cases}
      0 &\text{ if $q_1,q_2 \in \{e,o\}$ and}\\
      &\text{ \big(($q_1=q_2$ and $q=o$) or ($q_1\not=q_2$ and $q=e$)\big)}\\
      2 &\text{ if $q_1,q_2 \in \{e,o\}$, $q_1\not=q_2$, and  $q=r$}\\
      3 &\text{ if $q_1,q_2 \in \{e,o\}$, $q_1=q_2$, and $q=r$}\\
      \infty & \text{ otherwise},
      \end{cases}
    \]
    and 
\item $F_e= F_o =  \infty$ and $F_r= 0$.
\end{compactitem}
Figure \ref{fig:ex-6.4} shows the hypergraph for $\D$. We note that $\D$ is not bu-deterministic, and hence not crisp-deterministic. 
Moreover, for every  $\xi \in T_\Sigma$ and $q \in Q$, it is clear that
\[
  h_{\V(\D)}(\xi)_q =
  \begin{cases}
    0 & \text{ if }  \ \big((q=e \text{ and } |\pos(\xi)| \text{ is even } ) \text{ or }  (q=o \text{ and } |\pos(\xi)| \text{ is odd } )\big)\\
    2 & \text{ if }  \ q=r \text{ and } |\pos(\xi)| \text{ is even }  \\
    3 & \text{ if }  \ q=r \text{ and } |\pos(\xi)| \text{ is odd }  \\
    \infty & \text{ otherwise}\enspace.
    \end{cases}
\]
Thus
\begin{align*}
  \sem{\D}^\init(\xi) & = \min_{q \in Q} (h_{\V(\D)}(\xi)_q + F_q)  = h_{\V(\D)}(\xi)_r + F_r =  h_{\V(\D)}(\xi)_r = \sizemod2(\xi) \enspace. 
\end{align*}

Next we construct the Nerode algebra $\Ne(\D)=(Q_\Ne,\theta_\Ne,F_\Ne)$. By Proposition \ref{prop:Nerode=image}, we have 
      \[Q_\Ne= \im(h_{\V(\D)}) = \{ [0,\infty,2], [\infty, 0, 3]\},
             \]
             where the component at the left (middle, and right) is the $e$-component (respectively, $o$-component, and $r$-component) of the vectors.  Thus $\Ne(\D)$ is finite. Let us abbreviate the vectors
             \([0,\infty,2]\) and \([\infty, 0, 3]\)
by $E$ and $O$, respectively.
             Moreover, we have 
             \begin{align*}
\theta(\alpha)() &=  O, \
               \theta(\gamma)(O) = E, \
               \theta(\gamma)(E) = O, \\
               \theta(\sigma)(E,E) &= O, \ \theta(\sigma)(E,O) = E, \ \theta(\sigma)(O,E) = E, \ \theta(\sigma)(O,O) = O, 
             \end{align*}
             and
             \begin{align*}
               (F_\Ne)_E &= \min( E_q + F_q \mid q \in Q) = \min ( 0 + \infty, \infty + \infty, 2 + 0) = 2,\\
               (F_\Ne)_O &= \min( O_q + F_q \mid q \in Q) = \min ( \infty + \infty, 0 + \infty, 3 + 0) = 3\enspace.
             \end{align*}
             We note that the Nerode algebra $\Ne(\D)$ and the derivative $(\Sigma,\mathrm{TSR})$-algebra $\der(\sizemod2)$ of Example \ref{ex:size-mod2-cont} are isomorphic. (We will deal with the general relation between the Nerode algebra and derivative algebras in Theorem \ref{thm:characterizatio-of-Nerode-automaton}.)
             
             Since $\Ne(\D)$ is finite, we construct the crisp-deterministic $(\Sigma,\mathrm{TSR})$-wta $\rel(\Ne(\D)) =(Q_\Ne,\delta_\Ne,F_\Ne)$ by letting 
 $(\delta_\Ne)_0(\varepsilon,\alpha,E) = \infty$ and $(\delta_\Ne)_0(\varepsilon,\alpha,O) = 0$ , and for every $q_1,q_2,q \in Q_\Ne$ we let
   \[
    (\delta_\Ne)_1(q_1,\gamma,q) =
    \begin{cases}
      0 & \text{ if $q_1\ne q$}\\
      \infty & \text{ otherwise},
      \end{cases}
    \]
    
  \[
    (\delta_\Ne)_2(q_1q_2,\sigma,q) =
    \begin{cases}
      0 & \text{ if ($q_1=q_2$ and $q=O$) or ($q_1\not=q_2$ and $q=E$)}\\
      \infty & \text{ otherwise}.
      \end{cases}
    \]
    We realize that $\rel(\Ne(\D))$ and the crisp-deterministic wta of Example \ref{ex:size-mod2} are essentially the same.
    \hfill $\Box$
      \end{ex}

\subsection{Properties of the Nerode algebra}

In this section we show two properties of the Nerode $(\Sigma,B)$-algebra $\Ne(\A)$, cf. Theorem 
\ref{theo:Nerode_conditions}  and Theorem \ref{thm:characterizatio-of-Nerode-automaton}. For this, we introduce some preparatory concepts.

Let $\A=(Q,\delta,F)$ be a $(\Sigma,B)$-wta. For every $F':Q\to B$,  the $(\Sigma,B)$-wta $\B=(Q,\delta,F')$ is a \emph{final variant of $\A$}. 
We define $\mathrm{FV}(\A)=\{\B \mid \B \text{ is a final variant of } \A\}$. Moreover, we denote  the set $\{\sem{\B}^\init\mid \B \in \mathrm{FV}(\A)\}$ by $\sem{\mathrm{FV}(\A)}^\init$.
In a similar way, we define final variants of $(\Sigma,B)$-algebras. Let $\K=(Q,\theta,F)$ be a $(\Sigma,B)$-algebra. For every $F':Q\to B$, the $(\Sigma,B)$-algebra $F'(\K)=(Q,\theta,F')$ is a \emph{final variant of~$\K$}.

In addition, let $q\in Q$. We define the mapping $h_{\V(\A)}^q: T_\Sigma \to B$ by
\[h_{\V(\A)}^q(\xi) = h_{\V(\A)}(\xi)_q\]
for every $\xi \in T_\Sigma$. Then, let $\A^q=(Q,\delta,F^q)$ be the final variant of $\A$ defined by $(F^q)_q = \1$ and $(F^q)_p = \0$ for every $p \in Q \setminus \{q\}$. It is obvious that $h_{\V(\A)}^q=\sem{\A^q}^\init$, hence $h_{\V(\A)}^q \in \sem{\mathrm{FV}(\A)}^\init$.

In the following we give some characterizations for the fact that $\Ne(\A)$ is finite.

\begin{theo} \label{theo:Nerode_conditions} (cf. \cite[Thm.~6.3]{cirdroignvog10})
Let $\A=(Q,\delta,F)$ be a $(\Sigma,B)$-wta. Then the following statements are equivalent.
\begin{enumerate}
\item[(i)] $\Ne(\A)$ is finite.
\item[(ii)] For each final variant $\B$ of $\A$, the Nerode algebra $\Ne(\B)$ is finite and $\sem{\B}^\init$ is i-recognizable by a final variant of $\rel(\Ne(\A))$.
\item[(iii)] Each $r \in \sem{\mathrm{FV}(\A)}^\init$ is i-recognizable by some crisp-deterministic $(\Sigma,B)$-wta.
\end{enumerate}
\end{theo}

\begin{proof} (i) $\Rightarrow$ (ii): Let $\sem{\B}^\init \in \sem{\mathrm{FV}(\A)}^\init$ for some
final variant $\B=(Q,\delta,F')$ of $\A$.
By Lemma \ref{prop:A_N_equals_A_init}, $\sem{\B}^\init=\sem{\Ne(\B)}$. By our assumption, $\Ne(\A)$ is finite.
Moreover, $\Ne(\B)=(Q_\Ne,\theta_\Ne,F'_\Ne)$ for some $F'_\Ne : Q_\Ne \to B$, hence
$\Ne(\B)$ is finite. Since $\Ne(\B)$ is a final variant of $\Ne(\A)$, also $\rel(\Ne(\B))$ is a final variant of $\rel(\Ne(\A))$.
Moreover, $\sem{\Ne(\B)} = \sem{\rel(\Ne(\B))}^\init$ by Lemma \ref{lm:cdwta_semantics}, which proves the statement.

(ii) $\Rightarrow$ (iii): This implication is obvious because, by assumption (ii), $\Ne(\A)$ is finite and hence $\rel(\Ne(\A))$ is a crisp-deterministic wta.

(iii) $\Rightarrow$ (i): Let $q\in Q$. As we saw, $h_{\V(\A)}^q=\sem{\A^q}^\init$. By our assumption (iii), $h_{\V(\A)}^q$ is i-recognizable by a crisp-deterministic wta. Hence, $h_{\V(\A)}^q$ has a finite image by the implication (i) $\Rightarrow$ (iv) of Lemma \ref{lm:cdwta_finite_image} (by letting $r= h_{\V(\A)}^q$). 
Moreover. 
\[|\im(h_{\V(\A)}) |\le \prod_{q \in Q}|\im(h_{\V(\A)}^q)|,\]
hence $\Ne(\A)$ is finite.
\end{proof}

Before showing the second property of  $\Ne(\A)$, we exploit Theorem  \ref{theo:Nerode_conditions}
and show that the reverse of Theorem \ref{th:cd-init} does not hold. This theorem says in particular that, for each $(\Sigma,B)$-wta $\A$, if the Nerode algebra $\Ne(\A)$ is finite, then $\sem{\A}^\init$ is i-recognizable by a crisp-deterministic wta. However, the following also holds.

\begin{lm}\label{lm:N(A)-not-finite} There is a  $(\Sigma,\mathrm{TSR})$-wta $\A$ such that $\sem{\A}^\init$ is i-recognizable by a crisp-deterministic $(\Sigma,\mathrm{TSR})$-wta and $\Ne(\A)$ is not finite. 
\end{lm}
\begin{proof}
  Let $\A$ be the $(\Sigma,\mathrm{TSR})$-wta of Example \ref{ex:size} with the modification that $F_q=\infty$. (Note that $\A$ is even bu-deterministic.) Then we have $\sem{\A}^\init=\widetilde{\infty}$, where $\widetilde{\infty}$ is the weighted tree language which takes each tree to $\infty$.  Of course, $\sem{\A}^\init$ is i-recognizable by some crisp-deterministic $(\Sigma,\mathrm{TSR})$-wta.
  
Moreover, $\Ne(\A)$ is not finite because  Theorem \ref{theo:Nerode_conditions}(iii) does not hold for $\A$. 
In fact,  the bu-deterministic $(\Sigma,\mathrm{TSR})$-wta $\C$  of Example \ref{ex:size} is a final variant of $\A$  (the only difference is that $F_q=0$ in $\C$) and $\sem{\C}^\init=\size$. Since $\im(\size)$ is infinite, by Lemma \ref{lm:cdwta_finite_image} it is not true that $\sem{\C}^\init$ is i-recognizable by some crisp-deterministic $(\Sigma,\mathrm{TSR})$-wta.
\end{proof}

Next we give a representation of the Nerode algebra $\Ne(\A)$ of a wta $\A=(Q,\delta,F)$ in terms of the family $(h_{\V(\A)}^q\mid q\in Q)$ of weighted tree languages.

\begin{theo} \label{thm:characterizatio-of-Nerode-automaton} (cf. \cite[Thm.~6.5]{cirdroignvog10})
  Let  $\A=(Q,\delta,F)$ be a $(\Sigma,B)$-wta and $q_1,\ldots,q_n$ be an enumeration of the elements of $Q$. Then
\begin{quote}
there is a final variant $\C$ of  $\acc\big(\Pi\big((\der(h_{\V(\A)}^{q_i}) \mid i \in [n])\big)\big)$ such that $\Ne(\A)\cong\C$.
\end{quote}
\end{theo}
\begin{proof}
Let $r_i = h_{\V(\A)}^{q_i}$ and $\der(r_i)=(Q_{i},\theta_{i},F_{i})$ be the derivative $(\Sigma,B)$-algebra of $r_i$ for every $i\in[n]$.

Next, let $\Pi\big((\der(r_i) \mid i \in [n])\big)=(Q_\Pi,\delta_\Pi,F_\Pi)$ be the direct product of $(\der(r_i)\mid i\in [n])$. We abbreviate the $\Sigma$-homomorphism $h_{\Pi\big((\der(r_i) \mid i \in [n])\big)}$ by $h_\Pi$.

Let us recall that $\V(\A)=(B^Q,\delta_\A)$ and that $h_{\V(\A)}$ is the unique homomorphism from $T_\Sigma$ to $B^Q$.
For each context $c\in C_\Sigma$, we write just $c^\A$ for the mapping $c^{B^Q} : B^Q \to B^Q$ defined on page \pageref{page:defintion-c^A}.  Then we define the mapping $\varphi: Q_\Ne \to Q_\Pi$ by $\varphi(h_{\V(\A)}(\xi)) = h_\Pi(\xi)$ for every $\xi \in T_\Sigma$. For any $\xi,\zeta \in T_\Sigma$ we have that
\begin{align*}
 & h_{\V(\A)}(\xi) = h_{\V(\A)}(\zeta)\\
  \iff & (\forall c \in C_\Sigma)\ c^\A(h_{\V(\A)}(\xi)) = c^\A(h_{\V(\A)}(\zeta))
         \hspace{12mm} (\text{because $\Box\in C_\Sigma$ and $\Box^\A$ is the identity})\\
  \iff & (\forall c \in C_\Sigma)\ h_{\V(\A)}(c[\xi]) = h_{\V(\A)}(c[\zeta])
         \hspace*{21mm}(\text{by Lemma \ref{lm:hom-decomposition} with $(A,\theta)=(B^Q,\delta_\A)$)}\\
\iff & (\forall c \in C_\Sigma)(\forall i \in [n])\ h_{\V(\A)}(c[\xi])_{q_i} = h_{\V(\A)}(c[\zeta])_{q_i}\\
\iff & (\forall c \in C_\Sigma)(\forall i \in [n])\ r_i(c[\xi]) = r_i(c[\zeta])
    \hspace*{32mm}(\text{because $(\forall i \in [n])\ r_i = h_{\V(\A)}^{q_i}$})\\
  \iff & (\forall i \in [n])(\forall c \in C_\Sigma)\ (\xi^{-1}r_i)(c) = (\zeta^{-1}r_i)(c)
         \hspace*{7mm}(\text{by the definition of $(\xi^{-1}r_i)$ and $(\zeta^{-1}r_i)$})\\
  \iff & (\forall i \in [n])\ \xi^{-1}r_i = \zeta^{-1}r_i\\
  \iff & (\forall i \in [n])\ h_{\der(r_i)}(\xi) = h_{\der(i)}(\zeta)
         \hspace*{75mm}(\text{by \eqref{eq:derivative-homomorphism}})\\
  \iff & h_\Pi(\xi) = h_\Pi(\zeta)
         \hspace*{102mm}(\text{by  (\ref{eq:direct_product_homomorphism})})\\
  \iff & \varphi(h_{\V(\A)}(\xi)) = \varphi(h_{\V(\A)}(\zeta))
         \hspace*{60mm}(\text{by the definition of $\varphi$})
\end{align*}
and thus, $\varphi$ is well-defined and injective. Now, we define the $(\Sigma,B)$-algebra 
\[\K_\varphi=\big(\im(\varphi),\theta_\varphi,F_\varphi\big)\]
where, for every $k \in \N$, $\sigma \in \Sigma^{(k)}$, $q,q_1,\ldots,q_k \in \im(\varphi)$, we have $\theta_\varphi(\sigma)(q_1,\ldots,q_k) = \theta_\Pi(\sigma)(q_1,\ldots,q_k)$, and $(F_\varphi)_q= (F_\Ne)_{\varphi^{-1}(q)}$. (Note that $\im(\varphi)$ is closed under operations of $\{\theta_\Pi(\sigma) \mid \sigma \in \Sigma\}$.)
Then $\K_\varphi$ is a final variant of the accessible part of $\Pi\big((\der(r_i) \mid i \in [n])\big)$.

Finally, we show that $\varphi$ is a homomorphism from $\Ne(\A)$ to $\K_\varphi$. Let $k\in \N$, $\sigma\in \Sigma^{(k)}$, and $\xi_1,\ldots,\xi_k\in T_\Sigma$.  Then
\begin{align*}
& \varphi\big(\theta_\Ne(\sigma)(h_{\V(\A)}(\xi_1),\ldots,h_{\V(\A)}(\xi_k))\big) =\varphi(h_{\V(\A)}(\sigma(\xi_1,\ldots, \xi_k)))=h_\Pi(\sigma(\xi_1,\ldots, \xi_k)) \\
= \,& \theta_\Pi(\sigma)\big(h_\Pi(\xi_1),\ldots,h_\Pi(\xi_k) \big) = \theta_\Pi(\sigma)\big(\varphi(h_{\V(\A)}(\xi_1)),\ldots,\varphi(h_{\V(\A)}(\xi_k)) \big)\enspace.
\end{align*}
With this we proved that $\Ne(\A)$ and $\K_\varphi$ are isomorphic.
\end{proof}

\begin{cor} (cf. \cite[Thm.~6.5]{cirdroignvog10}) 
Let  $\A=(Q,\delta,F)$ be a $(\Sigma,B)$-wta and $q_1,\ldots,q_n$ be an enumeration of the elements of $Q$. Then $\Ne(\A)$ and a final variant of a subdirect product of $\big(\der(h_{\V(\A)}^{q_i}) \mid i \in [n]\big)$ are isomorphic.
\end{cor}

\begin{proof}
It follows from Theorem \ref{thm:characterizatio-of-Nerode-automaton} by Observation \ref{obs:subdirect-product}, because $\der(h_{\V(\A)}^q)$ is accessible for each $q \in Q$.
\end{proof}

Let $\A=(Q,\delta,F)$ be a $(\Sigma,B)$-wta such that $\Ne(\A)$ is finite. By  (i) $\Rightarrow$ (ii) of Theorem \ref{theo:Nerode_conditions},
$\rel(\Ne(\A))$ has the property that, for each $q\in Q$, the weighted tree language $h_{\V(\A)}^q$  is i-recognizable  by some final variant of $\rel(\Ne(\A))$ 
(because $h_{\V(\A)}^q\in\sem{\mathrm{FV}(\A)}^\init$). In the following we show that $\rel(\Ne(\A))$ is minimal among all
crisp-deterministic wta which have this property.

\begin{theo} \label{theo:minimal}(cf. \cite[Thm.~6.6]{cirdroignvog10})
Let $\A=(Q,\delta,F)$ be a $(\Sigma,B)$-wta such that $\Ne(\A)$ is finite. Then  $\rel(\Ne(\A))$ is minimal (with respect to the number of states) in the set
  \[
   U_\A = \{ \text{ crisp-deterministic wta } \B \mid (\forall q\in Q)(\exists \B'\in \mathrm{FV}(\B)): h_{\V(\A)}^q= \sem{\B'}^\init\}\enspace. 
  \]
\end{theo}

\begin{proof} As we saw, $\rel(\Ne(\A))\in U_\A$. Now let $\B=(Q',\delta',F')$ be an arbitrary crisp-deterministic $(\Sigma,B)$-wta in $U_\A$. 
By definition, for every $q \in Q$ there is a final variant $\B' = (Q',\delta',F'_q)$ of $\B$ with $\sem{\B'} = h_{\V(\A)}^q$. We will give a surjective mapping $\varphi: Q' \to Q_\Ne$. 

  Let $\rel(\B')=(Q', \theta',F'_q)$ be the $(\Sigma,B)$-algebra related to $\B'$. We can assume that $\rel(\B')$ is accessible, because otherwise if there is a state $q \in Q'$ which is not accessible, then there is a final variant $\B''$ of $\B$ with less states than $\B'$.

We define a mapping $\varphi$ by $\varphi(q') = h_{\V(\A)}(\xi)$ for each $q' \in Q'$, where $\xi \in T_\Sigma$ is such that $q'=h_{\rel(B')}(\xi)$. Then $\varphi$ is well-defined, which can be seen as follows. Let $q' \in Q'$ and $\xi,\zeta \in T_\Sigma$ such that $h_{\rel(B')}(\xi)=q'=h_{\rel(B')}(\zeta)$. Then \[h_{\V(\A)}(\xi)_q = (F'_q)_{h_{\rel(B')}(\xi)} = (F'_q)_{h_{\rel(B')}(\zeta)} = h_{\V(\A)}(\zeta)_q,\]
for every $q \in Q$ and hence, $h_{\V(\A)}(\xi) = h_{\V(\A)}(\zeta)$. In addition, $\varphi$ is surjective because $\rel(\B')$ is accessible. Thus, we conclude that $|Q_\Ne| \leq |Q'|$. Therefore, we have proved that $\rel(\Ne(\A))$ is a minimal crisp-deterministic  $(\Sigma,B)$-wta in the set $U_\A$. 
\end{proof}


\begin{algorithm}[t]
    \small
    \KwIn{a $(\Sigma,B)$-wta $\A=(Q,\delta,F)$}
    \KwOut{the crisp-deterministic $(\Sigma,B)$-wta $\rel(\Ne(\A))$}
    \BlankLine
    \Var $i \in \N$\;
    \hspace*{15mm} family $(G_i \mid i \in \mathbb{N})$ where $G_i=(V_i,E_i)$ is a $\Sigma$-hypergraph\label{line:var-family}\;
    \hspace*{15mm} $X \subseteq B^Q,\quad Y \subseteq \{(B^Q)^k \times \Sigma^{(k)} \times B^Q \mid k \in [0, \maxrk(\Sigma)]\},\quad v \in B^Q$\;
    \BlankLine
    $V_0 \leftarrow \emptyset$ and $E_0 \leftarrow \emptyset$ \Comment*[r]{\textrm{\% this forms the hypergraph $G_0$}}
    $i\leftarrow 0$\;
    
    \Repeat
    {
        $G_i = G_{i-1}$
    }
    {
        $X\leftarrow\emptyset$ and $Y\leftarrow\emptyset$\;
        \ForEvery{$k \in [0, \maxrk(\Sigma)]$, $\sigma \in \Sigma^{(k)}$, and $v_1,\ldots,v_k \in V_i$}
        {
            $v \leftarrow \delta_\A(\sigma)(v_1,\ldots,v_k)$, 
            $X \leftarrow X \cup \{v\}$, and 
            $Y \leftarrow Y \cup \{\langle v_1, \ldots, v_k,\sigma, v \rangle\}$\;
        }
        $V_{i+1} \leftarrow V_i \cup X$ and $E_{i+1} \leftarrow E_i \cup Y$\Comment*[r]{\textrm{\% this forms the hypergraph $G_{i+1}$}}
        $i \leftarrow i+1$\;
    }
    output the crisp-deterministic wta $(Q_\Ne,\delta_\Ne,F_\Ne)$ with\;
    \begin{itemize}
      \item $Q_\Ne = V_i$,
    \item $\delta_\Ne = ((\delta_\Ne)_k \mid k \in \N)$ and  $\supp((\delta_\Ne)_k)= \{(v_1\ldots v_k,\sigma,v)\mid  \langle v_1, \ldots, v_k,\sigma, v \rangle\in E_i\}$ and $\im((\delta_\Ne)_k) \in \{\0,\1\}$ and
      \item $F_\Ne: Q_\Ne \to B$ such that  $(F_\Ne)_v \leftarrow v \otimes F$ for each $v \in Q_\Ne$
      \end{itemize}
    \caption{Construction of the crisp-deterministic $(\Sigma,B)$-wta $\rel(\Ne(\A))$}
    \label{alg:construct-rel-Ne(A)}
    \end{algorithm}

\subsection{Sufficient conditions for finiteness and  the algorithmic construction of the Nerode algebra}

Next we give sufficient conditions for the strong bimonoid $B$ and the wta $\A$ which guarantee that the Nerode algebra $\Ne(\A)$ is finite. Moreover, we give an algorithm to construct the  crisp-deterministic wta which is i-equivalent to $\A$.

    \begin{cor}\label{cor:cd-init} Let $\A$ be a $(\Sigma,B)$-wta. If
      \begin{itemize}
      \item $B$ is locally finite or
      \item $B$ is multiplicatively locally finite and $\A$ is bu-deterministic,
      \end{itemize}
      then the Nerode $(\Sigma,B)$-algebra $\Ne(\A)$ is finite and for the crisp-deterministic wta $\rel(\Ne(\A))$ we have $\sem{\A}^\init= \sem{\rel(\Ne(\A))}^\init$.
  \end{cor}
  \begin{proof} Let $\A=(Q,\delta,F)$. 
    First we consider the case that $B$ is locally finite. Since  $\im(\delta)$ is finite and $B$ is locally finite, the carrier set $H$ of the subalgebra of $B$ generated by $\im(\delta)$ is finite. Since $\im(h_{\V(\A)})\subseteq H^Q$, also $\im(h_{\V(\A)})$ is finite. Thus, by Proposition \ref{prop:Nerode=image}, also $\Ne(\A)$ is finite. Then the result follows from Theorem~\ref{th:cd-init}.

    Second  we consider the case that $B$ is multiplicatively locally finite and $\A$ is bu-deterministic.
    Let $H=\langle \im(\delta) \rangle_{\{\otimes\}}$. Due to the fact that $\A$ is bu-deterministic we have $h_{\V(\A)}(\xi)\in H^Q$ for each $\xi \in T_\Sigma$. Thus $\im(h_{\V(\A)})$ is finite. Then we can finish as in the first case. 
    \end{proof}

Finally, we present the generalization (cf. Algorithm \ref{alg:construct-rel-Ne(A)})  of \cite[Algorithm~6.4]{cirdroignvog10} which we can use to construct $\rel(\Ne(\A))$ for a $(\Sigma,B)$-wta $\A$ if $\Ne(\A)$ is finite.
If Algorithm \ref{alg:construct-rel-Ne(A)} is given a $(\Sigma,B)$-wta $\A=(Q,\delta,F)$ as input and it terminates, then it outputs the crisp-deterministic $(\Sigma,B)$-wta $\rel(\Ne(\A))$. Algorithm \ref{alg:construct-rel-Ne(A)} terminates of input $\A$ if and only if $\Ne(\A)$ is finite.


\section{Crisp-determinization for the run semantics}\label{sect:run-semantics}

We introduce the concept of finite order property for a $(\Sigma,B)$-wta $\A$. Then we prove that if $\A$ has the finite order property, then we can construct the crisp-deterministic wta $\mathcal{R}(\A)$ which is r-equivalent to $\A$ (cf. Theorem~\ref{theo:A_pi_equals_A_run}). We show that if $B$ is bi-locally finite, then $\A$ has the finite order property (cf. Corollary~\ref{thm:det-run}). Next we present an algorithm  which, given  an arbitrary wta $\A$ with the finite order property,  delivers the crisp-deterministic wta $\mathcal{R}(\A)$ (cf. Algorithm~\ref{alg:construct-R(A)}). Lastly, we relate the number of states of $\mathcal{N}(\A)$ defined in Section~\ref{sect:initial-algebra-semantics} and $\mathcal{R}(\A)$ (cf. Theorem~\ref{thm:relating}).

\subsection{Finite order property implies crisp-determinization}

The following concepts are generalizations of the corresponding ones in \cite[Sect.~8]{cirdroignvog10} to the tree case.
Let   $\A=(Q,\delta,F)$ be a $(\Sigma,B)$-wta and let $H_\A$ denote the set $\langle \im(\delta) \rangle_{\{\otimes\}}$. Then $\A$ has the {\em finite order property} if 
\begin{compactitem}
\item[-] the set $H_\A$ is finite, and 
\item[-] each element $b \in H_\A \otimes \im(F)$ has a finite order in $(B,\oplus,\0)$. 
\end{compactitem}
If, \textit{e.g.}, $B$ is bi-locally finite, then each $(\Sigma,B)$-wta has the finite order property.

\begin{ex}\rm \label{ex:ex-finite-order-prop} The $(\Sigma,\mathrm{TSR})$-wta $\A=(Q,\delta,F)$ of Example \ref{ex:size-mod2} has the finite order property, because $H_\A= \langle \im(\delta) \rangle_{\{+\}} = \langle \{0,\infty\} \rangle_{\{+\}} = \{0,\infty\}$ is finite and each element in
  \[
    \{a+b \mid a \in H_\A, b \in \im(F)\} = \{a+b \mid a \in \{0,\infty\}, b \in \{2,3\}\} = \{2,3,\infty\}
    \]
    has finite order in $(\N,\min,\infty)$, because $\min$ is idempotent.

    The $(\Sigma,\mathrm{TSR})$-wta $\A=(Q,\delta,F)$ of Example \ref{ex:size-mod2-non-det} does not have the finite order property, because, \textit{e.g.}, $2 \in \im(\delta)$ and $\langle \{2\}\rangle_{\{+\}}$ is infinite, and thus $H_\A$ is infinite too.
    \hfill $\Box$
  \end{ex}

\begin{quote}\em In this subsection, let $\A=(Q,\delta,F)$ be a $(\Sigma,B)$-wta and we assume that $\A$ has the finite order property.
\end{quote}

Let $\xi\in T_\Sigma$. For each $\rho \in R_\A(\xi)$, we have $\wt_\A(\rho)\otimes F_{\rho(\varepsilon)} \in H_\A \otimes \im(F)$.
Thus $\sem{\A}^\run(\xi)$ is a sum over the finite set $H_\A \otimes \im(F)$. The fact that each element of $H_\A \otimes \im(F)$ has a finite order guarantees that any sum over $H_\A \otimes \im(F)$ is equal to a finite sum over this set. In the following we formalize this phenomenon.

We denote by $\lcm(K)$ the {\em least common multiple of $K$} for each finite subset $K \subseteq \N_+$. We define the integers
\[i_\A = \max\{i(b) \in \N_+ \mid b \in H_\A \otimes \im(F)\} \text{ and } p_\A=\lcm\{p(b) \in \N_+ \mid b \in H_\A \otimes \im(F)\},\] 
respectively, where $i(b)$ is the index of $b$ and $p(b)$ is the period of $b$ in $(B,\oplus,\0)$ for each $b \in H_\A \otimes \im(F)$. 

For each $n\in \N$, we define the number $J_\A(n)\in [0, i_\A+p_\A-1]$ by
\[J_\A(n) =
\begin{cases}
n &\text{if $n< i_\A$,}\\
i_\A+\big((n-i_\A)\!\!\!\mod p_\A\big) &\text{if $n \ge i_\A$,}
\end{cases}
\]
where $(n-i_\A)\!\!\!\mod p_\A$ is the remainder when $n - i_\A$ is divided by $p_\A$. In the first case $J_\A(n) < i_\A$, and in the second case $i_\A \le J_\A(n) \leq i_\A + p_\A - 1$, so $J_\A(n)$ is well-defined. In both cases $n \equiv_{p_\A} J_\A(n)$, where $\equiv_{p_\A}$ denotes the congruence modulo $p_\A$. Moreover, for every $b\in H_\A$, $b'\in \im(F)$, and $n\in \N$,
we have $n(b\otimes b')=J_\A(n)(b\otimes b')$.

Since $H_\A$ is finite by assumption,  the set $Q \times H_\A$ is also finite. We define the set of $q$-runs on $\xi$ of which the weight is $b$ by
\[R_\A(q,\xi,b)=\{\rho \in R_\A(q,\xi) \mid \wt_\A(\rho) = b\}\]
for every $(q,b) \in Q \times H_\A$ and $\xi \in T_\Sigma$.

Moreover, for every $\xi\in T_\Sigma$, let us define $p_\xi: Q \times H_\A \to \N$ and $\pi_\xi: Q \times H_\A \to [0, i_\A+p_\A-1]$ by 
\[p_\xi(q,b) = |R_\A(q,\xi,b)| \text{ and } \pi_\xi(q,b)=J_\A(p_\xi(q,b)).\]
Then for every $\xi\in T_\Sigma$, $(q,b) \in Q \times H_\A$, and $b' \in \im(F)$, we have
\begin{equation} \label{eq:pi_and_p_equal}
\pi_\xi(q,b)(b\otimes b') = p_\xi(q,b)(b\otimes b').
\end{equation}

Now we define  the crisp-deterministic $(\Sigma,B)$-wta $\R(\A) = (Q_\R,\delta_\R,F_\R)$ where
\begin{compactitem}
\item $Q_\R=\{\pi_\xi \mid \xi \in T_\Sigma\}$,
\item  $(\delta_\R)_k(\pi_{\xi_1} \ldots \pi_{\xi_k}, \sigma, \pi_\xi) =
\begin{cases}
\1 &\text{if $\xi=\sigma(\xi_1,\ldots,\xi_k)$}\\
\0 &\text{otherwise,}
\end{cases}$ \\
for every $k\in \N$, $\xi,\xi_1,\ldots,\xi_k \in T_\Sigma$, and $\sigma \in \Sigma^{(k)}$, and
\item $(F_\R)_{\pi_\xi}=\bigoplus_{(q,b) \in Q \times H_\A} \pi_\xi(q,b)(b\otimes F_q)$ for every $\xi \in T_\Sigma$.
\end{compactitem}

Since $|Q_\R| \leq (i_\A+p_\A)^{|Q| \cdot |H_\A|}$, the set $Q_\R$ is finite. If, in particular, each element in $H_\A \otimes \im(F)$ is additively idempotent, then we have  $|Q_\R| \leq 2^{|Q| \cdot |H_\A|}$.

Next we prove that $\delta_\R$ is well defined. For this we need some preparations.  For every $(q,b) \in Q \times H_\A$ and $\xi=\sigma(\xi_1,\ldots,\xi_k) \in T_\Sigma$, we define
\begin{equation}\label{eq:S_runs}
\begin{aligned}
S_\xi(q,b)=\{& (q_1,b_1)\ldots(q_k,b_k) \in (Q \times H_\A)^{k} \mid (\forall i \in [k]): \pi_{\xi_i}(q_i,b_i) > 0, \\ & \Big(\bigotimes_{i=1}^{k}b_i\Big)\otimes\delta_k(q_1 \ldots q_k,\sigma,q) = b\}\enspace,
\end{aligned}
\end{equation}where, and in the rest of this section, $(q_1,b_1)\ldots(q_k,b_k)$ abbreviates $((q_1,b_1),\ldots,(q_k,b_k))$.

\begin{lm} \label{lm:runs} (cf. \cite[Lm.~8.1]{cirdroignvog10})
For every $(q,b) \in Q \times H_\A$ and $\xi=\sigma(\xi_1,\ldots,\xi_k) \in T_\Sigma$, we have 
\begin{equation} \label{eq:runs}
p_\xi(q,b) = \sum_{(q_1,b_1)\ldots(q_k,b_k)\in S_\xi(q,b)} \,\prod_{i=1}^{k} p_{\xi_i}(q_i,b_i).
\end{equation}
\end{lm}
\begin{proof}
For each $(q_1,b_1)\ldots(q_k,b_k) \in (Q \times H_\A)^k$ let us define the set $P_{\xi_1,\ldots,\xi_k}\left((q_1,b_1)\ldots(q_k,b_k)\right)$ by
\[P_{\xi_1,\ldots,\xi_k}\left((q_1,b_1)\ldots(q_k,b_k)\right) = \{\rho \in R_\A(q,\xi) \mid (\forall i \in [k]): \rho|_i \in R_\A(q_i,\xi_i,b_i)\}.\] 
Then we have 
\[R_\A(q,\xi,b) = \bigcup_{(q_1,b_1)\ldots(q_k,b_k)\in S_\xi(q,b)} P_{\xi_1,\ldots,\xi_k}\big((q_1,b_1)\ldots(q_k,b_k)\big).\]
Clearly, $|P_{\xi_1,\ldots,\xi_k}\left((q_1,b_1)\ldots(q_k,b_k)\right)|=\prod_{i=1}^k|R_\A(q_i,\xi_i,b_i)|=\prod_{i=1}^k p_{\xi_i}(q_i,b_i)$. Moreover,
$P_{\xi_1,\ldots,\xi_k}\left((q_1,b_1)\ldots(q_k,b_k)\right) \cap P_{\xi_1,\ldots,\xi_k}\left((p_1,c_1)\ldots(p_k,c_k)\right) = \emptyset$ if there is an $i\in[k]$ with $(q_i,b_i) \neq (p_i,c_i)$.
Then (\ref{eq:runs}) follows. 
\end{proof}

Now we show that $(\delta_\R)_k$ is well defined for each $k\in[n]$. For this, let $\xi=\sigma(\xi_1,\ldots,\xi_k)$ and $\zeta=\sigma(\zeta_1,\ldots,\zeta_k)\in T_\Sigma$ such that $\pi_{\xi_i} = \pi_{\zeta_i}$ for each $i \in [k]$. It suffices to show
that $\pi_\xi=\pi_\zeta$. For this, let $(q,b) \in Q \times H_\A$. By (\ref{eq:S_runs}) we obtain that $S_\xi(q,b) = S_\zeta(q,b)$. Let us abbreviate $S_\xi(q,b)$ and  $S_\zeta(q,b)$ by $S$. Then by Lemma \ref{lm:runs} we obtain that
\[p_\xi(q,b) = \sum_{(q_1,b_1)\ldots(q_k,b_k)\in S} \prod_{i=1}^{k} p_{\xi_i}(q_i,b_i)
\text{ and }
p_\zeta(q,b) = \sum_{(q_1,b_1)\ldots(q_k,b_k)\in S} \prod_{i=1}^{k} p_{\zeta_i}(q_i,b_i).
\]
Next, we suppose that $\pi_\xi(q,b) < i_\A$. Then for every $(q_1,b_1)\ldots(q_k,b_k)\in S$ we have that
\[
\sum_{(q_1,b_1)\ldots(q_k,b_k)\in S} \prod_{i=1}^{k} p_{\xi_i} (q_i,b_i)  =p_\xi(q,b) = \pi_\xi(q,b) < i_\A,
\]

so $p_{\xi_i}(q_i,b_i) = \pi_{\xi_i}(q_i,b_i) = \pi_{\zeta_i}(q_i,b_i)$ and
$p_{\zeta_i}(q_i,b_i) = \pi_{\zeta_i}(q_i,b_i) = p_{\xi_i}(q_i,b_i)$ for each $i \in [k]$. Therefore,
\begin{equation*}
\begin{aligned}
p_\zeta(q,b) = \sum_{(q_1,b_1)\ldots(q_k,b_k)\in S} \prod_{i=1}^{k} p_{\zeta_i} (q_i,b_i) = \sum_{(q_1,b_1)\ldots(q_k,b_k)\in S} \prod_{i=1}^{k} p_{\xi_i} (q_i,b_i) = p_\xi(q,b),
\end{aligned}
\end{equation*}
which yields
\begin{equation} \label{eq:pi_xi_equal_pi_zeta}
\pi_\xi(q,b) = p_\xi(q,b) = p_\zeta(q,b) = \pi_\zeta(q,b).
\end{equation}
Similarly we prove that $\pi_\zeta(q,b) < i_\A$ implies (\ref{eq:pi_xi_equal_pi_zeta}).

Assume now that $i_\A \leq \pi_\xi(q,b), \pi_\zeta(q,b) \leq i_\A + p_\A - 1$. For any $(q_1,b_1)\ldots(q_k,b_k) \in S$ and all $i \in [k]$
\[p_{\xi_i}(q_i,b_i) \equiv_{p_\A} \pi_{\xi_i}(q_i,b_i) = \pi_{\zeta_i}(q_i,b_i) \equiv_{p_\A} p_{\zeta_i}(q_i,b_i).\]
Therefore,
\begin{equation} \label{eq:pi_xi_equiv_pi_zeta}
\begin{aligned}
\pi_\xi(q,b) 
\equiv_{p_\A} p_\xi(q,b)
&= \sum_{(q_1,b_1)\ldots(q_k,b_k)\in S}\;\prod_{i=1}^k p_{\xi_i}(q_i,b_i)\\
&\equiv_{p_\A}  \sum_{(q_1,b_1)\ldots(q_k,b_k)\in S}\;\prod_{i=1}^k p_{\zeta_i}(q_i,b_i) = p_\zeta(q,b) \equiv_{p_\A} \pi_\zeta(q,b).
\end{aligned}
\end{equation}
In (\ref{eq:pi_xi_equiv_pi_zeta}), we use that if $c \equiv_{p_\A} c'$ and $d \equiv_{p_\A} d'$, then $(c + d) \equiv_{p_\A} (c' + d')$ and $(c \cdot d) \equiv_{p_\A} (c' \cdot d')$.
Since $i_\A \leq \pi_\xi(q,b), \pi_\zeta(q,b) \leq i_\A + p_\A - 1$, we conclude that $\pi_\xi(q,b) = \pi_\zeta(q,b)$. Thus $(\delta_\R)_k$ is well defined.

\begin{theo} \label{theo:A_pi_equals_A_run} (cf. \cite[Thm.~8.2]{cirdroignvog10}) Let $\A=(Q,\delta,F)$ be a $(\Sigma,B)$-wta. If $\A$ has the finite order property, then $\sem{\R(\A)}^\run = \sem{\A}^\run$.
\end{theo}

\begin{proof}
Let $\rel(\R(\A)) = (Q_\R,\theta_\R,F_\R)$.
We show that $h_{\rel(\R(\A))}(\xi) = \pi_\xi$ for every $\xi\in T_\Sigma$. For this, let
 $\xi=\sigma(\xi_1,\ldots,\xi_k)\in T_\Sigma$. Then
\begin{align*}
h_{\rel(\R(\A))}(\xi)
&= h_{\rel(\R(\A))}(\sigma(\xi_1,\ldots,\xi_k))= \theta_\R(\sigma)(h_{\rel(\R(\A))}(\xi_1), \ldots, h_{\rel(\R(\A))}(\xi_k))\\
&= \theta_\R(\sigma)(\pi_{\xi_1}, \ldots, \pi_{\xi_k}) = \pi_\xi.
\end{align*}

Then for each $\xi \in T_\Sigma$ we have
\begin{equation}\label{eq:general-case}
\begin{aligned}
\sem{\R(\A)}^\run(\xi) &= \sem{\R(\A)}^\init(\xi) =^{(*)} \sem{\rel(\R(\A))}(\xi)  \\
&= (F_\R\circ h_{\rel(\R(\A))})(\xi) = (F_\R)_{h_{\rel(\R(\A))}(\xi)} = (F_\R)_{\pi_\xi} = \bigoplus_{(q,b) \in Q \times H_\A} \pi_\xi(q,b)(b \otimes F_q)\\
&=^{(**)} \bigoplus_{(q,b) \in Q \times H_\A} \left(\bigoplus_{\rho \in R_\A(q,\xi,b)} \wt_\A(\rho) \otimes F_q\right) = \bigoplus_{\rho \in R_\A(\xi)} \wt_\A(\rho) \otimes F_{\rho(\varepsilon)} = \sem{\A}^\run(\xi),
\end{aligned}
\end{equation}
where $(*)$ is due to Lemma \ref{lm:cdwta_semantics} and  $(**)$ is justified as follows.
By \eqref{eq:pi_and_p_equal} and the definitions of $p_\xi(q,b)$ and $R_\A(q,\xi,b)$, we have 
\[\pi_\xi(q,b)(b \otimes F_q) = p_\xi(q,b)(b \otimes F_q) = \bigoplus_{\rho \in R_\A(q,\xi,b)} \wt_\A(\rho) \otimes F_q \enspace.\]
This  proves $\sem{\R(\A)}^\run = \sem{\A}^\run$.
 \end{proof}

\begin{ex}
 We give a wta which has the finite order property, and construct a crisp-deterministic wta which is r-equivalent to that wta.
For this, let $\Sigma=\{\sigma^{(2)}, \gamma^{(1)}, \alpha^{(0)}\}$. As weight structure, we use the tropical semiring $\mathrm{TSR}=(\N_\infty,\min,+,\infty,0)$. The 
$(\Sigma,\mathrm{TSR})$-wta $\E=(Q,\delta,F)$ is defined as follows.
\begin{compactitem}
\item $Q=\{q_0,q_1,q_2\}$,
\item $\delta_0(\varepsilon,\alpha,q_0)=\delta(\varepsilon,\alpha,q_1)=\delta_1(q_0,\gamma,q_1)=\delta_2(q_1q_1,\sigma,q_2)=0$,
\item for every other $k\in \N$, $\eta \in \Sigma^{(k)}$, and $p,p_1,\ldots,p_k \in Q$, we have $\delta_k(p_1\ldots p_k,\eta,p)=\infty$, and
\item $F_{q_0}=F_{q_1}=F_{q_2}=1$.
\end{compactitem}
Figure \ref{fig:fin-ord-prop} shows the $\Sigma$-hypergraph for $\E$.  We have
    \[
    \sem{\E}^\run(\xi) =
    \begin{cases}
      1 & \text{if  $\xi \in \{\alpha,\gamma\alpha,\sigma(\alpha,\alpha), \sigma(\alpha,\gamma\alpha), \sigma(\gamma\alpha,\alpha), \sigma(\gamma\alpha,\gamma\alpha)\}$ }\\
      \infty & \text{otherwise}\enspace.
      \end{cases}
    \]
    
    \begin{figure}[t]
\begin{center}
\begin{tikzpicture}
\tikzset{node distance=7em, scale=0.6, transform shape}
\node[state, rectangle] (1) {$\alpha$};
\node[state, right of=1] (2) {$q_0$};
\node[state, rectangle, right of=2] (3) {$\gamma$};
\node[state, right of=3] (4)  {$q_1$};
\node[state, rectangle, below of=4] (5) {$\alpha$};
\node[state, rectangle, right of=4] (6) {$\sigma$};
\node[state, right of=6] (7) {$q_2$};

\tikzset{node distance=2em}
\node[above of=1] (w1) {0};
\node[above of=2] (w2) {1};
\node[above of=3] (w3) {0};
\node[above of=4] (w4) {1};
\node[above of=5] (w5) [right=0.05cm] {0};
\node[above of=6] (w6) {0};
\node[above of=7] (w7) {1};

\draw[->,>=stealth] (1) edge (2);
\draw[->,>=stealth] (2) edge[out=0, in=180, looseness=1.4] (3);
\draw[->,>=stealth] (3) edge[out=0, in=180, looseness=1.4] (4);
\draw[->,>=stealth] (5) edge (4);
\draw[->,>=stealth] (4) edge[out=30, in=150, looseness=1.4] (6);
\draw[->,>=stealth] (4) edge[out=-30, in=-150, looseness=1.4] (6);
\draw[->,>=stealth] (6) edge[out=0, in=180, looseness=1.4] (7);
\end{tikzpicture}
\end{center}

\caption{The $\Sigma$-hypergraph for the $(\Sigma,\mathrm{TSR})$-wta $\E=(Q,\delta,F)$ which is not bu-deterministic and which has the finite order property.}
\label{fig:fin-ord-prop}
\end{figure}
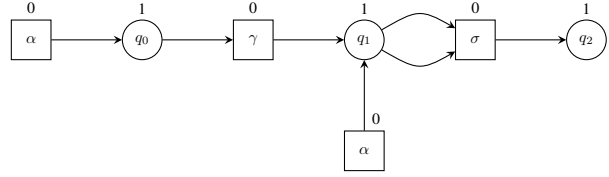

It is easy to see that $H_\E = \{0, \infty\}$. Thus, $\E$ has the finite order property because TSR is additively idempotent. 

Next we compute the values $p_\xi(q,b)$ and $\pi_\xi(q,b)$ for every $\xi \in T_\Sigma$ and $(q,b) \in Q \times H_\E$. Table \ref{table:p-pi-values} shows these values. By our computation, we have $\pi_{\sigma(\alpha,\alpha)}=\pi_{\sigma(\alpha,\gamma\alpha)}=\pi_{\sigma(\gamma\alpha,\alpha)}=\pi_{\sigma(\gamma\alpha,\gamma\alpha)}$, hence, we write $\pi_{\sigma(\alpha,\alpha)}$ for $\pi_{\sigma(\gamma\alpha,\alpha)}$, $\pi_{\sigma(\gamma\alpha,\gamma\alpha)}$, and $\pi_{\sigma(\gamma\alpha,\gamma\alpha)}$. Moreover, for each $\xi' \in (T_\Sigma \setminus \supp(\sem{\E}^\run))$, we have 
$\pi_{\xi'}=\pi_{\gamma^2\alpha}$, we write $\pi_{\gamma^2\alpha}$ for $\pi_{\xi'}$.

\begin{table}[t]
\centering
\begin{tabular}{l|l|cccccc}
& & $(q_0,0)$ & $(q_1,0)$ & $(q_2,0)$ & $(q_0,\infty)$ & $(q_1,\infty)$ & $(q_2,\infty)$\\
\hline
\multirow{2}{*}{$\alpha$}  & $p_\alpha$ & $1$ & $1$ & $0$ & $0$ & $0$ & $1$\\
& $\pi_\alpha$ &  $1$ & $1$ & $0$ & $0$ & $0$ & $1$\\
\hline
\multirow{2}{*}{$\gamma\alpha$}  & $p_{\gamma\alpha}$ & $0$ & $1$ & $0$ & $3$ & $2$ & $3$\\
& $\pi_{\gamma\alpha}$ &  $0$ & $1$ & $0$ & $1$ & $1$ & $1$\\
\hline
\multirow{2}{*}{$\sigma(\alpha,\alpha)$} & $p_{\sigma(\alpha,\alpha)}$ & $0$ & $0$ & $1$ & $9$ & $9$ & $8$\\
& $\pi_{\sigma(\alpha,\alpha)}$ &  $0$ & $0$ & $1$ & $1$ & $1$ & $1$\\
\hline
\multirow{2}{*}{$\sigma(\alpha,\gamma\alpha)$} & $p_{\sigma(\alpha,\gamma\alpha)}$ & $0$ & $0$ & $1$ & $27$ & $27$ & $26$\\
& $\pi_{\sigma(\alpha,\gamma\alpha)}$ &  $0$ & $0$ & $1$ & $1$ & $1$ & $1$\\
\hline
\multirow{2}{*}{$\sigma(\gamma\alpha,\alpha)$} & $p_{\sigma(\gamma\alpha,\alpha)}$ & $0$ & $0$ & $1$ & $27$ & $27$ & $26$\\
& $\pi_{\sigma(\gamma\alpha,\alpha)}$ &  $0$ & $0$ & $1$ & $1$ & $1$ & $1$\\
\hline
\multirow{2}{*}{$\sigma(\gamma\alpha,\gamma\alpha)$} & $p_{\sigma(\gamma\alpha,\gamma\alpha)}$ & $0$ & $0$ & $1$ & $81$ & $81$ & $80$\\
& $\pi_{\sigma(\gamma\alpha,\gamma\alpha)}$ &  $0$ & $0$ & $1$ & $1$ & $1$ & $1$\\
\hline
\multirow{2}{*}{$\xi' \in (T_\Sigma \setminus \supp(\sem{\E}^\run))$} & $p_{\xi'}$ & $0$ & $0$ & $0$ & $>\!0$ & $>\!0$ & $>\!0$\\
& $\pi_{\xi'}$ &  $0$ & $0$ & $0$ & $1$ & $1$ & $1$\\
\end{tabular}
\caption{ The values $p_\xi(q,b)$ and $\pi_\xi(q,b)$ for every $\xi \in T_\Sigma$ and $(q,b) \in Q \times H_\E$.}
\label{table:p-pi-values}
\end{table}

Then we construct the crisp-deterministic $(\Sigma,\mathrm{TSR})$-wta $\R(\E)=(Q_\R,\delta_\R,F_\R)$, where
\begin{compactitem}
\item $Q_\R=\{\pi_\alpha, \pi_{\gamma\alpha}, \pi_{\sigma(\alpha,\alpha)}, \pi_{\gamma^2\alpha}\}$,
\item 
\begin{itemize}
\item for each $\pi \in Q_\R$ we have\\
$(\delta_\R)_0(\varepsilon, \alpha,\pi)=\begin{cases}
0 &\text{if $\pi=\pi_\alpha$}\\
\infty &\text{otherwise},
\end{cases}$

\item for every $\pi, \pi_1 \in Q_\R$ we have\\
$(\delta_\R)_1(\pi_1,\gamma,\pi)=\begin{cases}
0 &\text{if ($\pi_1=\pi_\alpha$ and $\pi=\pi_{\gamma\alpha}$) or } \\
 &\text{($\pi_1 \neq \pi_\alpha$ and $\pi=\pi_{\gamma^2\alpha}$)}\\
\infty &\text{otherwise},
\end{cases}$

\item for every $\pi,\pi_1,\pi_2 \in Q_\R$ we have\\
$(\delta_\R)_2(\pi_1 \pi_2, \sigma, \pi) \begin{cases}
0 &\text{if ($\pi_1,\pi_2\in\{\pi_\alpha, \pi_{\gamma\alpha}\}$ and $\pi=\pi_{\sigma(\alpha,\alpha)}$) or}\\
 &\text{\big($\pi=\pi_{\gamma^2\alpha}$ and ($\pi_1 \not\in \{\pi_\alpha, \pi_{\gamma\alpha}\}$ and/or $\pi_2 \not\in \{\pi_\alpha, \pi_{\gamma\alpha}\}$)\big)}\\ 
\infty &\text{otherwise}\enspace.
\end{cases}$
\end{itemize}
\item for each $\pi \in Q_\R$ we have\\
$F_\R(\pi) = \begin{cases}
\infty &\text{if $\pi = \pi_{\gamma^2\alpha}$},\\
1 &\text{otherwise}.
\end{cases}$
\end{compactitem}

By Theorem \ref{theo:A_pi_equals_A_run}, $\R(\E)$ is a crisp-deterministic wta such that $\sem{\R(\E)}^\run = \sem{\E}^\run$.
\hfill $\Box$
\end{ex}

\subsection{Sufficient conditions for finite order property and the algorithmic construction of \texorpdfstring{$\R(\A)$}{R(A)}}

 \begin{cor} \label{thm:det-run} (cf. \cite[Thm.~11]{drostuvog10}) Let $\A$ be a $(\Sigma,B)$-wta. If
   $B$ is bi-locally finite,  then  we have that $\sem{\A}^\run= \sem{\R(\A)}^\run$.
   \end{cor}
   \begin{proof} Since $B$ is bi-locally finite, $\A$ has the finite order property. Then the result follows from Theorem~\ref{theo:A_pi_equals_A_run} .
   \end{proof}

   We note that the condition in Corollary \ref{thm:det-run} (that $B$ is bi-locally finite) is different from the condition in Corollary \ref{cor:cd-init} (that $B$ is locally finite).  Clearly, each locally finite strong bimonoid is also bi-locally finite.
In the following we give an example of a bi-locally finite strong bimonoid which is not locally finite.

\begin{ex} \cite[Ex.~2.2]{cirdroignvog10}
For each $\lambda \in \mathbb{R}$ with $0 < \lambda< \frac{1}{2}$,  let $(B,\oplus,\odot,0,1)$ be the algebra, where
$B =\{0\} \cup \{b \in \mathbb{R}\mid \lambda \le b \le 1\}$, $a \oplus b = \min(a + b, 1)$, and $a\odot b$ is $a \cdot b$ if $a \cdot b \ge \lambda$ and $0$ otherwise, and where $+$ and $\cdot$ are the usual addition and multiplication of real numbers, respectively.
Then $B$ is a commutative strong bimonoid.

It is easy to see that $B$ is bi-locally finite. However, for $\lambda = \frac{1}{4}$ it is not locally finite which can be seen as follows \cite{dro19}. 
Let  $(b_i \mid i \in \N)$ such that $b_0= \frac{1}{2}$ and, if $i$ is odd, then $b_i = \frac{1}{2} \cdot b_{i-1}$, and if $i$ is even and $i\not=0$, then $b_i = \frac{1}{2} + b_{i-1}$. Then, \textit{e.g.}, $b_0= 1/2$, $b_1=1/4$, $b_2= 1/2+1/4 = 3/4$, $b_3= 3/8$, $b_4= 1/2+3/8 = 7/8$, $b_5= 7/16$, $b_6 = 1/2+7/16 = 15/16$, $b_7 = 15/32$,
$b_8=1/2+15/32 = 31/32$, $b_9=31/64$, $b_{10}= 1/2+31/64 = 63/64$, etc. (We note that the subsequences $(b_i \mid i \in \N, i \text{ is even})$ and  $(b_i \mid i \in \N, i \text{ is odd})$ converge to $1$ and $\frac{1}{2}$, respectively.) It is easy to see that $b_i \in \langle \{\frac{1}{2}\}\rangle_{\{\oplus,\odot\}}$ for each $i \in \N$, and that  $b_i \not= b_j$ for every $i,j \in \N$ with $i\not=j$. Hence $(b_i \mid i \in \N)$ is an infinite family of elements in $\langle \{\frac{1}{2}\}\rangle_{\{\oplus,\odot\}}$, and thus $B$ is not locally finite.

We note that $B$ is not a semiring, because $\odot$ is not right distributive. For instance, for $a=b=0.9$, and $c=\lambda$, we have $(a \oplus b) \odot c=\lambda$, while $(a \odot c) \oplus (b \odot c)=0$ because $a \odot c = b \odot c = 0$.
\hfill $\Box$
\end{ex}

\begin{cor} Let $\Sigma$ be a ranked alphabet such that  $|\Sigma^{(1)}|\geq 2$. The following two statements are equivalent.
\begin{enumerate}
\item[(i)] $B$ is bi-locally finite.
\item[(ii)] For every $(\Sigma,B)$-wta $\A$, the weighted tree language $\sem{\A}^\run$ is r-recognizable by a crisp-deterministic wta.
\end{enumerate}
\end{cor}
\begin{proof} The proof of (i) $\Rightarrow$ (ii) follows from Corollary \ref{thm:det-run}. The proof of (ii) $\Rightarrow$ (i) can be obtained by an easy generalization of \cite[Lm. 12]{drostuvog10} from weighted string automata to wta.
\end{proof}

\begin{function}[t] 
    \small
    \SetAlgoFuncName{Algorithm}{Algorithm}
    \KwIn{$k \in \N$,\enspace $\sigma \in \Sigma^{(k)},$ and $\pi_1,\ldots,\pi_k \in [0, i_\A+p_\A-1]^{Q \times H_\A}$}
    \KwOut{$\pi: [0, i_\A+p_\A-1]^{Q \times H_\A}$}
    \BlankLine
    \Var $(q,b), (q_1, b_1), \ldots ,(q_k,b_k) \in Q \times H_\A$, \quad $q \in Q$,  and\quad $b \in H_\A$\;
    \hspace*{15mm} family $(n_{q,b} \in \N \mid q \in Q, b \in H_\A)$\;
    \BlankLine
    \lForEach{$(q,b) \in Q \times H_\A$}{
      $n_{q,b} \leftarrow 0$
    }
    \ForEach{$(q,b) \in Q \times H_\A$}{
      \ForEvery{$(q_1,b_1), \ldots, (q_k,b_k) \in Q \times H_\A$}{
        \If{$b = (\bigotimes_{i=1}^k b_i) \otimes \delta_k(q_1 \ldots q_k, \sigma, q)$}{
          $n_{q,b} \leftarrow n_{q,b} + J_\A\big( \prod_{i=1}^k \pi_i(q_i,b_i)\big)$
        }
      }
      $\pi(q,b) \leftarrow J_\A(n_{q,b})$
    }
    output $\pi$\;
    \caption{compute()}
    \label{map:compute}
\end{function}

Next we present the construction of $\R(\A)$ (cf. Algorithm \ref{alg:construct-R(A)})  which is the generalization of \cite[Algorithm~8.3]{cirdroignvog10}. 
The construction uses Algorithm \ref{map:compute} which is a generalization of \cite[Algorithm~8.4]{cirdroignvog10}. For these, the following result is necessary .

\begin{cor} \label{cor:pi-computing}
  For every $(q,b) \in Q \times H_\A$ and $\xi = \sigma(\xi_1,\ldots,\xi_k) \in T_\Sigma$, we have
  \[
    \pi_\xi(q,b) = J_\A \Big(\sum_{(q_1,b_1) \ldots (q_k,b_k) \in S_\xi(q,b)} J_\A \Big( \prod_{i=1}^k \pi_{\xi_i}(q_i,b_i)\Big)\Big)\enspace.
  \]
\end{cor}

\begin{proof}
  Let $(q,b) \in Q \times H_\A$ and $\xi = \sigma(\xi_1,\ldots,\xi_k) \in T_\Sigma$. Then we have
  \begingroup
  \allowdisplaybreaks
  \begin{align*}
    \pi_\xi(q,b)
    &= J_\A\big(p_\xi(q,b)\big) \tag{\text{by the definition of $\pi_\xi(q,b)$}}\\
    &= J_\A \Big(\sum_{(q_1,b_1) \ldots (q_k,b_k) \in S_\xi(q,b)} \prod_{i=1}^k p_{\xi_i}(q_i,b_i)\Big) \tag{\text{by Equality~\eqref{eq:runs} of Lemma~\ref{lm:runs}}}\\
    &= J_\A \Big(\sum_{(q_1,b_1) \ldots (q_k,b_k) \in S_\xi(q,b)} J_\A \Big( \prod_{i=1}^k p_{\xi_i}(q_i,b_i)\Big)\Big)
    \tag{\text{since $J_\A(m+n) = J_\A\big( J_\A(m) + J_\A(n) \big)$ for every $m,n \in \N$}}\\[1.25em]
    &= J_\A \Big(\sum_{(q_1,b_1) \ldots (q_k,b_k) \in S_\xi(q,b)} J_\A \Big( \prod_{i=1}^k J_\A\big(p_{\xi_i}(q_i,b_i)\big)\Big)\Big)
    \tag{\text{since $J_\A(m \cdot n) = J_\A\big( J_\A(m) \cdot J_\A(n) \big)$ for every $m,n \in \N$}}\\[1.25em]
    &= J_\A \Big(\sum_{(q_1,b_1) \ldots (q_k,b_k) \in S_\xi(q,b)} J_\A \Big( \prod_{i=1}^k \pi_{\xi_i}(q_i,b_i)\Big)\Big)
    \tag{\text{by the definition of $\pi_{\xi_i}(q_i, b_i)$}}
  \end{align*}
  \endgroup
\end{proof}

If Algorithm \ref{map:compute} is given the input $k \in \N$, $\sigma \in \Sigma^{(k)}$, and $\pi_{\xi_1}, \ldots, \pi_{\xi_k} \in [0, i_\A+p_\A-1]^{Q \times H_\A}$ for some $\xi_1,\ldots,\xi_k\in T_\Sigma$, then it outputs the mapping $\pi_\xi \in [0, i_\A+p_\A-1]^{Q \times H_\A}$, where $\xi=\sigma(\xi_1,\ldots,\xi_k)$ (cf. Corollary \ref{cor:pi-computing}).

    \begin{algorithm}[t]
        \small
        \KwIn{$(\Sigma,B)$-wta $\A=(Q,\delta,F)$ with finite order property}
        \KwOut{the crisp-deterministic $(\Sigma,B)$-wta $\R(\A)$}
        \BlankLine
        \Macro{\hspace*{5mm}${\cal F} = [0, i_\A+p_\A-1]^{Q \times H_\A}$}\;
        \Var $i \in \N$\;
        \hspace*{15mm} family $(G_i \mid i \in \mathbb{N})$ where $G_i=(V_i,E_i)$ is a $\Sigma$-hypergraph\;
        \hspace*{15mm} $X \subseteq {\cal F},\quad  Y \subseteq \{{\cal F}^k \times \Sigma^{(k)} \times {\cal F} \mid k \in [0, \maxrk(\Sigma)]\}, \text{ and}\quad \pi \in {\cal F}$\;
        \BlankLine
        $V_0 \leftarrow \emptyset$ and $E_0 \leftarrow \emptyset$ \Comment*[r]{\textrm{\% this forms the hypergraph $G_0$}}
        $i\leftarrow 0$\;
        
        \Repeat
        {
            $G_i = G_{i-1}$
        }
        {
            $X\leftarrow\emptyset$ and $Y\leftarrow\emptyset$\;
            \ForEvery{$k \in [0, \maxrk(\Sigma)]$, $\sigma \in \Sigma^{(k)}$, and $\pi_1,\ldots,\pi_k \in V_i$}
            {
                $\pi \leftarrow \mathrm{compute}(k,\sigma,\pi_1, \ldots,\pi_k)$\Comment*[r]{\textrm{\% cf. Algorithm~\ref{map:compute}}} 
                $X \leftarrow X \cup \{\pi\}$, and 
                $Y \leftarrow Y \cup \{\langle \pi_1, \ldots, \pi_k,\sigma, \pi \rangle\}$\;
            }
            $V_{i+1} \leftarrow V_i \cup X$ and $E_{i+1} \leftarrow E_i \cup Y$\Comment*[r]{\textrm{\% this forms the hypergraph $G_{i+1}$}}
            $i \leftarrow i+1$\;
        }
        output the crisp-deterministic wta $(Q_\R,\delta_\R,F_\R)$ with\;
        \begin{itemize}
          \item $Q_\R = V_i$,
        \item $\delta_\R = ((\delta_\R)_k \mid k \in \N)$ and  $\supp((\delta_\R)_k)= \{(\pi_1\ldots \pi_k,\sigma,\pi)\mid  \langle \pi_1, \ldots, \pi_k,\sigma, \pi \rangle\in E_i\}$ and $\im((\delta_\R)_k) \in \{\0,\1\}$ and
          \item $F_\R: Q_\R \to B$ such that  $(F_\R)_\pi \leftarrow \bigoplus_{(q,b) \in Q \times H_\A} \pi(q,b)(b \otimes F_q)$ for each $\pi \in Q_\R$
          \end{itemize}
        \caption{Construction of the crisp-deterministic $(\Sigma,B)$-wta $\R(\A)$}
        \label{alg:construct-R(A)}
    \end{algorithm}

\subsection{Relating the number of states of \texorpdfstring{$\Ne(\A)$}{N(A)} and \texorpdfstring{$\R(\A)$}{R(A)}}

\begin{theo} \label{thm:relating}(cf. \cite[Prop.~10.2]{cirdroignvog10})
Let $B$ be right distributive and $\A=(Q,\delta,F)$ be a $(\Sigma,B)$-wta such that the set $H_\A$ is finite and each of its elements has a finite order in $(B,\oplus,\0)$. 
Then $Q_\R$ and $Q_\Ne$ are finite and $|Q_\Ne| \leq |Q_\R|$.
\end{theo}

\begin{proof}
Since $B$ is right distributive, for every $d \in H_\A$ and $b \in B$, the element $d \otimes b$ has a finite  order in $(B,\oplus,\0)$. Thus any element of $H_\A \otimes \im(F)$ has a finite  order in $(B,\oplus,\0)$. 
Hence the assumptions of Theorem \ref{theo:A_pi_equals_A_run} are satisfied, so $Q_\R$ is finite. 

To finish the proof, it is sufficient to give a surjective mapping $\psi: Q_\R \to Q_\Ne$. We define it 
 by $\psi(\pi_\xi) = h_{\V(\A)}(\xi)$ for each $\xi \in T_\Sigma$.
We show that $\psi$ is well defined, \textit{i.e.}, that
\[\forall \xi,\zeta \in T_\Sigma : \pi_\xi = \pi_\zeta \implies h_{\V(\A)}(\xi) = h_{\V(\A)}(\zeta).\]
Let $\xi,\zeta \in T_\Sigma$ such that $\pi_\xi = \pi_\zeta$. Then we have 
\[h_{\V(\A)}(\xi)_q = \bigoplus_{b \in H_\A} p_\xi(q,b)b = \bigoplus_{b \in H_\A} \pi_\xi(q,b)b = \bigoplus_{b \in H_\A} \pi_\zeta(q,b)b = \bigoplus_{b \in H_\A} p_\zeta(q,b)b = h_{\V(\A)}(\zeta)_q\]
for every $q \in Q$, where in the first and the last equality we use that $B$ is right distributive. Hence $\psi$ is a well-defined. Moreover it is surjective obviously, so we obtain that $|Q_\Ne| \leq |Q_\R|$.
\end{proof}

\section{Undecidability results}\label{sect:undecidability}

The undecidability results of this section only make sense if we assume that the strong bimonoids we consider are computable.
A strong bimonoid $(B,\oplus,\otimes,\0,\1)$ is called {\em computable} if $B$ is a recursive set and the operations $\oplus$ and $\otimes$ are computable (\textit{e.g.}, by a Turing machine).

\begin{quote}\emph{In the rest of this section, we assume that all the mentioned strong bimonoids are computable.}
\end{quote}

We will show that each of the following problems is undecidable for arbitrary ranked alphabet $\Sigma$,  strong bimonoid $B$, and bu-deterministic $(\Sigma,B)$-wta $\A$:
\begin{itemize}
\item[(Pi)] Does there exist a crisp-deterministic $(\Sigma, B)$-wta which is i-equivalent to $\A$?
\item[(Pii)] Is the $(\Sigma,B)$-algebra $\Ne(\A)$ is finite?
\item[(Piii)] Does $\A$ have the finite order property?
\end{itemize}
Each of these results is based on the reduction to an undecidability result of Mealy machines. Thus we devote the first subsection to the repetition of Mealy machines and their simulation by wta.

\subsection{Mealy machines and their simulation by weighted tree automata}\label{subsect:Mealy-machine}

A {\em Mealy machine} is a tuple $\M=(Q,\Delta,\tau,\nu)$ where $Q$ is a finite nonempty set (\emph{states}), $\Delta$ is an alphabet, $\tau:~Q~\times~\Delta~\to~Q$ is a mapping (\emph{transition mapping}), and $\nu: Q \times \Delta \to \Delta$ is a mapping (\emph{output mapping}). 

As usual, we extend the transition mapping $\tau$ to a mapping $\tau^*: Q \times \Delta^* \to Q$ as follows: $\tau^*(q,\varepsilon)=q$ for each $q \in Q$, and $\tau^*(q,wa)=\tau(\tau^*(q,w),a)$ for every $q \in Q$, $w \in \Delta^*$ and $a \in \Delta$. For the sake of simplicity, we denote by $qw$ the state $\tau^*(q,w)$ for every $q \in Q$ and $w \in \Delta^*$.

The {\em mapping induced by $\M$ at state $q$}, denoted by $\nu_q$, is the mapping $\nu_q: \Delta^* \to \Delta^*$ defined by $\nu_q(\varepsilon)=\varepsilon$ and $\nu_q(wa)=\nu_q(w)\nu(qw,a)$ for every $w \in \Delta^*$ and $a \in \Delta$. 
The {\em monoid generated by $\M$}, denoted by $\langle \M \rangle_{\{\circ\}}$, is the submonoid 
$\langle \{\nu_q \mid q \in Q\}\rangle_{\{\circ, \id_{\FS}\}}$ of the monoid $(\FS,\circ, \id_{\FS})$, where $\FS$ denotes the set of all mappings $f:\Delta^{*}\to \Delta^{*}$, and  $\id_{\FS}$ is the identity mapping defined by 
  $\id_{\FS}(w)=w$ for each $w \in \Delta^*$.

We will prove our undecidability results by reducing them to the following one.

\begin{theo} \cite[Thm.~3.13]{gil14} \label{thm:undecidable-automaton-semigroup} 
It is undecidable whether,  for an arbitrary Mealy machine $\M$, the monoid  $\langle \M \rangle_{\{\circ\}}$ is finite. 
\end{theo}

In the proof of our undecidability results, we will simulate $\langle \M \rangle_{\{\circ\}}$ for an arbitrary  Mealy machine $\M$ with input alphabet $\Delta$ by a bu-deterministic wta $\A_\M$. The weight algebra of $\A_\M$ is a strong bimonoid which, cum grano salis, contains the monoid $(\FS,\circ, \id_{\FS})$ as multiplicative part. In order to guarantee later that $\A_\M$ has the finite order property, we will extend  $(\FS,\circ, \id_{\FS})$ into a strong bimonoid with an idempotent addition.

Formally, we let $\infty$ be a new symbol, \textit{i.e.}, $\infty \not\in \Delta^*$. Then we consider
the commutative monoid $(\Delta^*_\infty,\lcp,\infty)$ (cf. \cite[Ex.~1(5)]{drostuvog10}), where $\Delta^*_\infty= \Delta^* \cup \{\infty\}$,  $\lcp$ is the longest common prefix operation on $\Delta^*_\infty$ such that $\lcp(w,\infty)=w=\lcp(\infty,w)$ for each $w \in \Delta^*_\infty$. Then, by \cite[Ex.~1(4)]{drostuvog10}, the algebra
\[
  (\FSinf,\overline{\lcp},\circ,\infty_{\FSinf},\id_{\FSinf})
\]
is a strong bimonoid, where
\begin{compactitem}
  \item  $\FSinf$ is the set of all computable mappings $f: \Delta^{*}_\infty \to \Delta^{*}_\infty$ such that $f(\infty) = \infty$,
  \item $\overline{\lcp}$ is the extension of $\lcp$ for mappings defined by $(\overline{\lcp}(f, g))(w)=\lcp(f(w), g(w))$ for every $f,g \in \FSinf$ and $w \in \Delta^*_\infty$,
  \item $\circ$ is the composition of mappings,
\item $\infty_{\FSinf}$ is the constant mapping zero defined by $\infty_{\FSinf}(w)=\infty$  for each $w \in \Delta^*_\infty$, and
\item $\id_{\FSinf}$ is the identity mapping  over $\FSinf$. 
\end{compactitem}
We note that the condition $f(\infty) = \infty$ is needed in order to guarantee $f \circ \infty_{\FSinf} = \infty_{\FSinf}$. We also note that this strong bimonoid is additively idempotent.

Now let $\M=(Q,\Delta,\tau,\nu)$ be a Mealy machine. We construct the ranked alphabet $\Sigma_\M$ and the $(\Sigma_\M,\FSinf)$-wta $\A_\M=(\{*\},\delta_\M,F_\M)$ as follows.
\begin{compactitem}
  \item $\Sigma_\M=\Sigma_\M^{(1)}\cup \Sigma_\M^{(0)}$ is the ranked alphabet with $\Sigma_\M^{(1)}=Q$ and $\Sigma_\M^{(0)}=\{e^{(0)}\}$ where $e \not\in Q$ is a new symbol,
\item $(\delta_\M)_0(\varepsilon,e,*)=\id_{\FSinf}$ and for each $q \in \Sigma_\M^{(1)}$, we have $(\delta_\M)_1(*,q,*)=\nu'_q$,  where
  $\nu'_q \in \FSinf$ is the extension of $\nu_q$ defined by $\nu'_q|_{\Delta^*}=\nu_q$ and $\nu'_q(\infty)=\infty$, and
  \item $(F_\M)_* = \id_{\FSinf}$.
\end{compactitem}
Note that $\A_\M$ is bu-deterministic.

\begin{quote}
\emph{In the rest of this section, let $\M = (Q,\Delta,\tau,\nu)$ be an arbitrary Mealy machine and $\A_\M= (\{*\},\delta_\M,F_\M)$ be the  $(\Sigma_\M,\FSinf)$-wta constructed from $\M$ as above. Also we abbreviate $\sem{\A_\M}^\init$ by $r_\M$. For each $\xi \in T_{\Sigma_\M}$ we identify $h_{\V(\A_\M)}(\xi)$ with its only component $h_{\V(\A_\M)}(\xi)_*$.}
\end{quote}

Due to the determinism of $\A_\M$ and the way it is constructed we obtain the following connections between $\M$ and $\A_\M$.

\begin{lm}\rm \label{lm:equ-three-sets} We have $\im(r_\M)=\im(h_{\V(\A_\M)})= \langle \{\nu'_q \mid q \in Q\}\rangle_{\{\circ,\id_{\FS}\}} = H_{\A_\M}$. Moreover, $|\im(r_\M)|=| \langle \M \rangle_{\{\circ\}}|$.
  \end{lm}

  \begin{proof}  For every $k\in \N$ and $q_1,\ldots,q_k \in \Sigma_\M^{(1)}$, we have
\begin{equation}\label{eq:r-M}
\begin{aligned}
r_\M(q_1 \ldots q_k e)&  =h_{\V(\A_\M)}(q_1 \ldots q_k e)_*\circ (F_\M)_* \\
& =(\id_{\FSinf}\circ \nu'_{q_k} \circ \ldots \circ \nu'_{q_1}) \circ \id_{\FSinf} =\nu'_{q_k} \circ \ldots \circ \nu'_{q_1}\\
& =\id_{\FSinf}\circ \nu'_{q_k} \circ \ldots \circ \nu'_{q_1}= h_{\V(\A_\M)}(q_1 \ldots q_k e)_*\enspace.
\end{aligned}
\end{equation}
We note that $\nu'_{q_k} \circ \ldots \circ \nu'_{q_1}=\id_{\FSinf}$ for $k=0$.
Then we have
\begin{align*}
  \im(r_\M) &= \im(h_{\V(\A_\M)}) =  \{\nu'_{q_k} \circ \ldots \circ \nu'_{q_1} \mid  k \in \N, q_1,\ldots,q_k \in \Sigma_\M^{(1)}\}\\
&=\langle \{\nu'_q \mid q \in Q\}\rangle_{\{\circ,\id_{\FS}\}} = \langle \im(\delta_\M) \rangle_{\{\circ\}} =  H_{\A_\M} \enspace.
\end{align*}

Moreover, it is clear that  the mapping $\varphi: \im(r_\M)\to\langle \M \rangle_{\{\circ\}}$ defined for every 
$k \in \N$ and  $q_1,\ldots,q_k \in Q$ by
\(\varphi(\nu'_{q_k} \circ \ldots \circ \nu'_{q_1})=\nu_{q_k} \circ \ldots \circ \nu_{q_1}\)
is a bijection (and a monoid homomorphism). Hence $|\im(r_\M)|=| \langle \M \rangle_{\{\circ\}}|$.
    \end{proof}

\subsection{Undecidability of crisp-determinization under initial algebra semantics}

Here we show that problem (Pi) is undecidable. For this we introduce the $\Sigma_\M$-algebra $\F_\M=(\FSinf, (\sigma_\M \mid \sigma \in \Sigma_\M))$ by
\begin{compactitem}
\item $e_\M = \id_{\FSinf}$ and 
\item $q_\M(f) = f \circ \nu'_q$ for every $q \in \Sigma_\M^{(1)}$ and $f \in \FSinf$.
\end{compactitem}

\begin{lm} \label{lm:rM-is-homomorphism}
$r_\M$ is a $\Sigma_\M$-algebra homomorphism from $(T_{\Sigma_\M}, (\overline{\sigma} \mid \sigma \in \Sigma_\M))$ to $\F_\M$.
\end{lm}

\begin{proof}
  Obviously, $r_\M(\overline{e})=r_\M(e)=\id_{\FSinf}=e_\M$.
  Moreover, for every $q \in \Sigma_\M^{(1)}$ and $\xi \in T_{\Sigma_\M}$, we have
\[r_\M(\overline{q}(\xi))=r_\M(q(\xi))=r_\M(\xi) \circ \nu'_q =  q_\M(r_\M(\xi))\enspace ,\]
where the second equality follows from Equation \eqref{eq:r-M}.
Hence, $r_\M$ is a $\Sigma_\M$-algebra homomorphism.
\end{proof}

\begin{lm} \label{lm:rec-step-map-Mcirc} The following two statements are equivalent.
\begin{enumerate}
\item[(i)]$\im(r_\M)$ is finite and for each $f \in \FSinf$ the $\Sigma_\M$-tree language $r_\M^{-1}(f)$ is recognizable.
\item[(ii)] $\langle \M \rangle_{\{\circ\}}$ is finite.
\end{enumerate}
\end{lm}

\begin{proof}
(i) $\Rightarrow$ (ii): By Lemma \ref{lm:equ-three-sets}, the statement trivially holds. 

(ii) $\Rightarrow$ (i):
By Lemma \ref{lm:equ-three-sets},  $\im(r_\M)$ is finite.
Next we show that for each $f \in \FSinf$ the $\Sigma_\M$-tree language $r_\M^{-1}(f)$ is recognizable. 

If $f \in \FSinf \setminus \im(r_\M)$, then $r_\M^{-1}(f)=\emptyset$ which is obviously recognizable.

Now let $f \in \im(r_\M)$. Since $\im(r_\M)$ is finite, the $r_\M$-image of the term algebra $T_{\Sigma_\M}$ in 
$\F_\M$ is a finite $\Sigma_\M$-algebra. Moreover, $r_\M$ is a (surjective) $\Sigma_\M$-algebra homomorphism from $T_{\Sigma_\M}$
to this finite algebra. Thus, by \cite[Cor. 2.7.2]{gecste84}, $r_\M^{-1}(f)$ is recognizable.
\end{proof}

\begin{theo}\label{thm:A-init-rec-step-map}
It is undecidable whether, for arbitrary ranked alphabet $\Sigma$,  strong bimonoid $B$, and bu-deterministic $(\Sigma, B)$-wta $\A$, there is a crisp-deterministic $(\Sigma,B)$-wta $\B$ such that $\sem{\B}^\init=\sem{\A}^\init$.
\end{theo}
\begin{proof} 
We prove by contradiction. Thus, we assume that it is decidable whether, for arbitrary ranked alphabet $\Sigma$, strong bimonoid $B$, and bu-deterministic $(\Sigma,B)$-wta $\A$, there is a crisp-deterministic wta $\B$ such that $\sem{\B}^\init=\sem{\A}^\init$.

  Now let $\M$ be an arbitrary Mealy machine and let $\A_\M$ be the $(\Sigma_\M,\FSinf)$-wta constructed from $\M$ as above.
 By Lemma \ref{lm:cdwta_finite_image}, there is a crisp-deterministic wta $\B$ such that $\sem{\B}^\init=r_\M$, if and only if Condition (i) of Lemma \ref{lm:rec-step-map-Mcirc} holds.  
  Thus, by our assumption and Lemma \ref{lm:rec-step-map-Mcirc}, we can decide whether $\langle \M \rangle_{\{\circ\}}$ is finite  for an arbitrary Mealy machine $\M$. This contradicts to Theorem \ref{thm:undecidable-automaton-semigroup}, \textit{i.e.}, our assumption is wrong.
\end{proof}

\begin{cor} It is undecidable whether, for arbitrary ranked alphabet $\Sigma$, strong bimonoid $B$, and  $(\Sigma, B)$-wta $\A$, 
there is a crisp-deterministic $(\Sigma,B)$-wta $\B$ such that $\sem{\B}^\init=\sem{\A}^\init$.
\end{cor}

\subsection{Undecidability of finiteness of Nerode algebras}

Next we show that the problem (Pii) is undecidable.

\begin{theo}\label{thm:finite-index-undec}
It is undecidable whether, for arbitrary ranked alphabet $\Sigma$, strong bimonoid $B$, and bu-deterministic $(\Sigma,B)$-wta $\A$, the $(\Sigma,B)$-algebra $\Ne(\A)$ is finite.
\end{theo}
\begin{proof} The proof is by contradiction. Thus, we assume that it is decidable whether, for arbitrary ranked alphabet $\Sigma$, strong bimonoid $B$, and bu-deterministic $(\Sigma,B)$-wta $\A$, the $(\Sigma,B)$-algebra $\Ne(\A)$ is finite.

  Now let $\M$ be an arbitrary Mealy machine and let $\A_\M$ be the $(\Sigma_\M,\FSinf)$-wta constructed from $\M$ as above.
  Then we have
  \[\Ne(\A_\M)  \text{ is finite } \ \text{ iff } \ \im(h_{\V(\A_\M)}) \text{ is finite } \ \text{ iff } \ \langle \M \rangle_{\{\circ\}} \text{ is finite }\enspace,
  \]
where the second equivalence follows from Lemma \ref{lm:equ-three-sets}.
  Thus, by our assumption, we can decide whether $\langle \M \rangle_{\{\circ\}}$ is finite  for an arbitrary Mealy machine $\M$. This contradicts to Theorem \ref{thm:undecidable-automaton-semigroup}, and this means that our assumption is wrong.
  \end{proof}

\begin{cor}\label{cor:finite-index-undec}
It is undecidable whether, for arbitrary ranked alphabet $\Sigma$, strong bimonoid $B$, and $(\Sigma,B)$-wta $\A$, the $(\Sigma,B)$-algebra $\Ne(\A)$ is finite.
\end{cor}

\subsection{Undecidability of finite order property}

Lastly, we show that problem (Piii) is undecidable. Let us recall that $H_{\A_\M}$ denotes the set $\langle \im(\delta_\M) \rangle_{\{\circ\}}$

\begin{theo}\label{thm:fin-ord-prop-undec}
It is undecidable whether, for arbitrary ranked alphabet $\Sigma$, strong bimonoid $B$, and bu-deterministic $(\Sigma,B)$-wta $\A$, the wta $\A$ has the finite order property.
\end{theo}
\begin{proof} The proof is by contradiction. Thus, we assume that it is decidable whether, for arbitrary ranked alphabet $\Sigma$, strong bimonoid $B$, and bu-deterministic $(\Sigma,B)$-wta $\A$, the wta $\A$ has the finite order property.

Now let $\M$ be an arbitrary Mealy machine and let $\A_\M$ be the $(\Sigma_\M,\FSinf)$-wta constructed from $\M$ as above.
It is easy to see that the strong bimonoid $\FSinf$ is additively idempotent, hence, each element $f \in H_{\A_\M} \circ \im(F_\M)$ has order 1 in $(\FSinf,\overline{\lcp},\infty_{\FSinf})$. Therefore, $\A_\M$ has the finite order property if and only if $H_{\A_\M}$ is finite.
Thus, by our assumption and Lemma \ref{lm:equ-three-sets}, we can decide whether $\langle \M \rangle_{\{\circ\}}$ is finite  for an arbitrary Mealy machine $\M$. This contradicts to Theorem \ref{thm:undecidable-automaton-semigroup}, and this means that our assumption is wrong.
\end{proof}


\begin{cor}\label{cor:fin-ord-prop-undec}
It is undecidable whether, for arbitrary ranked alphabet $\Sigma$, strong bimonoid $B$, and $(\Sigma,B)$-wta $\A$, the wta $\A$ has the finite order property.
\end{cor}

\subsection{Undecidability for the string case}

Weighted string automata can be considered as  wta over monadic ranked alphabets \cite[p.~324]{fulvog09}, and vice versa. A ranked alphabet $\Sigma$ is \emph{monadic} if $\Sigma = \Sigma^{(0)} \cup \Sigma^{(1)}$ and $|\Sigma^{(0)}|=1$; say $\Sigma^{(0)}=\{e\}$. Each string $w$ over an alphabet $\Gamma$ can be considered as a tree $\mathrm{tree}(w)$ over the monadic ranked alphabet $\Sigma_\Gamma$ with $\Sigma_\Gamma^{(1)}=\Gamma$ and $\mathrm{tree}(\varepsilon)=e$. Obviously, $\mathrm{tree}: \Gamma^* \to T_{\Sigma_\Gamma}$ is a bijection. Then, a weighted string automaton  $\A=(Q,\lambda,\mu,\gamma)$ over $\Gamma$ with weights in $B$  \cite{cirdroignvog10} can be turned into the $(\Sigma_\Gamma,B)$-wta $\A_t=(Q,\delta,\gamma)$ where $\delta_0(\varepsilon,e,q)=\lambda_q$ and  $\delta_1(q,a,q')=\mu(a)_{q,q'}$.
Then $\sem{\A}^\init(w) = \sem{\A_t}^\init(\mathrm{tree}(w))$ and $\sem{\A}^\run(w) = \sem{\A_t}^\run(\mathrm{tree}(w))$ for every $w \in \Gamma^*$. Also, in a straightforward way, we can transform each wta over a monadic ranked alphabet into a weighted string automaton such that the corresponding equations hold.

Using the fact that weighted string automata are wta over monadic ranked alphabets, we can transfer our undecidability results to weighted string automata in the following way.

\begin{cor} Each of the following questions is undecidable for arbitrary alphabet $\Gamma$, 
 strong bimonoid $B$, and weighted string automaton $\A$ over $\Gamma$ and $B$:
\begin{itemize}
\item[(Ri)] Does there exist a crisp-deterministic weighted string automaton over $\Gamma$ and $B$ which is i-equivalent to $\A$?
\item[(Rii)] Is the Nerode $(\Sigma,B)$-algebra $\Ne(\A)$ is finite?
\item[(Riii)] Does $\A$ have the finite order property?
\end{itemize}
\end{cor}

\section{Open problems}

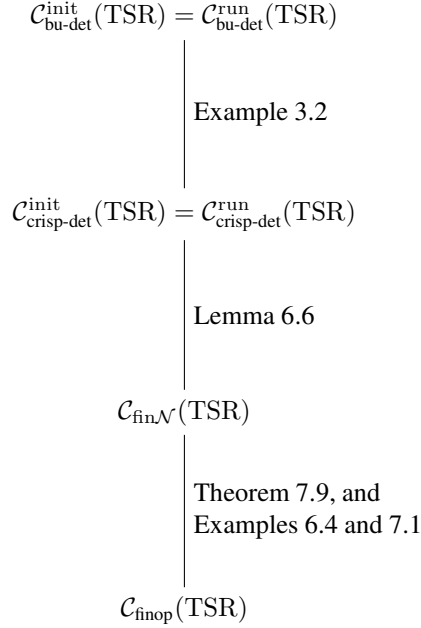
\begin{figure}[t]
\centering
\begin{tikzpicture}[node distance=7.5em]
\node (1) {$\C_\text{finop}(\mathrm{TSR})$};
\node[above of=1] (2) {$\C_{\mathrm{fin}\Ne}(\mathrm{TSR})$};
\node[above of=2] (3) {$\C^\init_\text{crisp-det}(\mathrm{TSR})=\C^\run_\text{crisp-det}(\mathrm{TSR})$};
\node[above of=3] (4) {$\C^\init_\text{bu-det}(\mathrm{TSR})=\C^\run_\text{bu-det}(\mathrm{TSR})$};

\path[-]
  (1) edge node[right, align=left] {Theorem \ref{thm:relating}, and \\Examples \ref{ex:size-mod2-non-det} and \ref{ex:ex-finite-order-prop}} (2)
  (2) edge node[right, align=left] {Lemma \ref{lm:N(A)-not-finite}} (3)
  (3) edge node[right, align=left] {Example \ref{ex:size}} (4)
;
\end{tikzpicture}
\caption{\label{fig:det-class} The determinization classification $(\mathrm{TSR},D)$.}
\end{figure}

Let  $\B$ be a class of strong bimonoids. For every $\mathrm{s} \in \{\init,\run\}$ and $\mathrm{m} \in \{\text{bu-det}, \text{crisp-det}\}$,
let 
\begin{compactitem}
\item $\C^\mathrm{s}_\mathrm{m}(\B)$ be the class of all $(\Sigma,B)$-wta $\A$ for some ranked alphabet $\Sigma$ and $B\in \B$ such that there is a $\mathrm{m}$ $(\Sigma,B)$-wta $\A'$ with $\sem{\A}^\mathrm{s} = \sem{\A'}^\mathrm{s}$.
\end{compactitem}
Moreover, let 
\begin{compactitem}
  \item $\C_{\mathrm{fin}\Ne}(\B)$ be the class of all $(\Sigma,B)$-wta $\A$ for some ranked alphabet $\Sigma$ and $B\in \B$ such that $\Ne(\A)$ is finite, and
  \item $\C_\mathrm{finop}(\B)$  the class of all $(\Sigma,B)$-wta $\A$ for some ranked alphabet $\Sigma$ and $B\in \B$ such that $\A$ has the finite order property.
  \end{compactitem}
Thus we have the following six classes of wta:
  \begin{equation}\label{eq:classes}
    \C^\init_\text{bu-det}(\B), \  \C^\init_\text{crisp-det}(\B), \ \C^\run_\text{bu-det}(\B), \ \C^\run_\text{crisp-det}(\B), \  
\C_{\mathrm{fin}\Ne}(\B), \ \text{ and } \ \C_\mathrm{finop}(\B) \enspace.
  \end{equation}
  The following inclusion relations between these classes hold for each class $\B$ of strong bimonoids:
  \begin{compactenum}
  \item[(i)] $\C^\mathrm{s}_\text{crisp-det}(\B) \subseteq \C^\mathrm{s}_\text{bu-det}(\B)$ for each $\mathrm{s} \in \{\init,\run\}$ by definition,
    \item[(ii)] $\C_{\mathrm{fin}\Ne}(\B) \subseteq \C^\init_\text{crisp-det}(\B)$ by Theorem \ref{th:cd-init}, and
    \item[(iii)] $\C_\mathrm{finop}(\B) \subseteq \C^\run_\text{crisp-det}(\B)$ by Theorem \ref{theo:A_pi_equals_A_run}.
\end{compactenum}
Moreover, for each class $\B$ of semirings:
\begin{compactenum}
    \item[(iv)]  $\C^\init_\text{bu-det}(\B)=\C^\run_\text{bu-det}(\B)$ and $\C^\init_\text{crisp-det}(\B)=\C^\run_\text{crisp-det}(\B)$, cf. Theorem  \ref{lm:run=initial}.
      \end{compactenum}

It would be nice to identify classes $\B$ of strong bimonoids for which a complete description of the inclusion relations can be given among the six classes \eqref{eq:classes} of wta. Let us form this problem more exactly.

A \emph{determinization classification} is a pair $(\B,D)$ such that
  \begin{compactitem}
  \item $\B$ is a class of strong bimonoids,
  \item $D$ is the Hasse diagram of the classes \eqref{eq:classes}.
  \end{compactitem}

Next we give some easy examples of  determinization classifications. For the sake of brevity, for a singleton class  $\{B\}$, we write just $B$.

For instance, $(\mathbb{B},D)$ is a determinization classification, where $D$  is the Hasse diagram which contains just one node and this node is labeled by all the six classes; indeed, all these classes are equal to the class of all $(\Sigma,\mathbb{B})$-wta for some ranked alphabet $\Sigma$.

  As another example, $(\mathrm{TSR},D)$ is also a determinization classification, where  $D$ is the Hasse diagram shown in Figure \ref{fig:det-class}.
The inclusions and equalities shown by $D$ were justified above, except the inclusion  $\C_\text{finop}(\mathrm{TSR})\subseteq \C_{\mathrm{fin}\Ne}(\mathrm{TSR})$, which follows from Theorem \ref{thm:relating}. Moreover, each inclusion is proper because
\begin{compactitem}
\item the wta $\D$ in Example \ref{ex:size-mod2-non-det} is in $\C_{\mathrm{fin}\Ne}(\mathrm{TSR}) \setminus \C_\text{finop}(\mathrm{TSR})$ (cf Example \ref{ex:ex-finite-order-prop}),
\item the wta $\A$ in the proof of Lemma \ref{lm:N(A)-not-finite} is in  $\C^\init_\text{crisp-det}(\mathrm{TSR})\setminus \C_{\mathrm{fin}\Ne}(\mathrm{TSR})$, and
\item the wta $\C$ of Example \ref{ex:size} is in $\C^\init_\text{bu-det}(\mathrm{TSR})\setminus \C^\init_\text{crisp-det}(\mathrm{TSR})$.
\end{compactitem}

\acknowledgements The authors would like to thank Manfred Droste for valuable discussions and the reviewers for their work and useful suggestions.

\bibliographystyle{alpha}
\bibliography{crisp19-bib}

\end{document}